\pdfoutput=1
\documentclass[format=acmsmall, review=false, screen=true]{acmart}
\setcopyright{none}\acmDOI{}\acmPrice{}
\makeatletter
\def\@permissionCodeOne#1#2#3#4#5#6#7#8#9{}
\def\@article@string{}
\fancypagestyle{firstpagestyle}{
  \fancyfoot[RO,LE]{}
}
\def\@formatdoi{}
\makeatother

\usepackage[utf8x]{inputenc}
\usepackage[T1]{fontenc}
\usepackage{bookmark}
\usepackage{xcolor}
\definecolor{linkcolor}{rgb}{0.9, 0.17, 0.31} 
\definecolor{citecolor}{rgb}{0.36, 0.54, 0.66}

\usepackage{amssymb}
\usepackage{amsmath}
\usepackage{multicol}
\usepackage{bbm}
\usepackage{tabu}

\usepackage{xspace}
\usepackage{enumitem}
\setlist[1]{leftmargin=*,itemsep=.4ex plus .1 ex minus .1ex}

\usepackage{microtype}

\def\inv{^{-1}}

\usepackage{tikz}
\usetikzlibrary{matrix,decorations.text,arrows,shapes,fit,shadows,calc,patterns,intersections}
\tikzstyle{marked}=[draw,minimum size=1pt,circle,fill=white]
\tikzstyle{pebble}=[draw,circle,fill=black]
\tikzstyle{ligne}=[line width=0.5pt,-]

\tikzstyle{non}=[inner sep=1pt]
\tikzstyle{lbox}=[rounded corners=5pt,draw=black!60,very thick,align=center]
\tikzstyle{wbox}=[rounded corners=5pt,align=center]
\tikzstyle{gbox}=[rounded corners=5pt,fill=green!20,align=center]
\tikzstyle{bbox}=[rounded corners=5pt,fill=blue!20,align=center]
\tikzstyle{rbox}=[rounded corners=5pt,fill=red!20,align=center]
\tikzstyle{obox}=[rounded corners=5pt,fill=yellow!35,align=center]
\tikzstyle{ledg}=[draw=black,very thick]
\tikzstyle{linc}=[fill=white,draw=black,inner sep=2pt,very thick,circle]

\newcommand\awfa{\ensuremath{\Abb_\alpha}\xspace}

\newcommand\wfwa[1]{\ensuremath{\text{\sc wf}_{\alpha}(#1)}\xspace}

\newcommand\wpart[2]{\ensuremath{[#1]_{#2}}\xspace}

\newcommand\ppart[1]{\wpart{#1}{\Pb}}
\newcommand\nat{\ensuremath{\mathbb{N}}\xspace}

\newcommand\Abb{\ensuremath{\mathbb{A}}\xspace}
\newcommand\frB{\ensuremath{\mathbb{B}}\xspace}

\newcommand\Lb{\ensuremath{\mathbf{L}}\xspace}

\newcommand\Kb{\ensuremath{\mathbf{K}}\xspace}

\newcommand\Pb{\ensuremath{\mathbf{P}}\xspace}

\newcommand\Gs{\ensuremath{\mathcal{G}}\xspace}
\newcommand\Cs{\ensuremath{\mathcal{C}}\xspace}
\newcommand\Ds{\ensuremath{\mathcal{D}}\xspace}

\newcommand\Is{\ensuremath{\mathcal{I}}\xspace}

\newcommand{\mso}{\ensuremath{\textup{MSO}}\xspace}

\newcommand{\sic}[1]{\ensuremath{\Sigma_{#1}}\xspace}
\newcommand{\sio}[1]{\ensuremath{\Sigma_{#1}(<)}\xspace}

\newcommand{\pio}[1]{\ensuremath{\Pi_{#1}(<)}\xspace}

\newcommand{\bso}[1]{\ensuremath{\mathcal{B}\Sigma_{#1}(<)}\xspace}

\newcommand{\sipe}[1]{\ensuremath{\Sigma_{#1}(<,+1)}\xspace}
\newcommand{\pipe}[1]{\ensuremath{\Pi_{#1}(<,+1)}\xspace}

\newcommand{\bspe}[1]{\ensuremath{\mathcal{B}\Sigma_{#1}(<,+1)}\xspace}

\newcommand{\ddp}{\ddp{2}}

\newcommand{\fow}{\ensuremath{\textup{FO}(<)}\xspace}
\newcommand{\fows}{\ensuremath{\textup{FO}(<,+1)}\xspace}

\newcommand{\fod}{\ensuremath{\textup{FO}^2}\xspace}
\newcommand{\fodw}{\ensuremath{\textup{FO}^2(<)}\xspace}
\newcommand{\fodp}{\ensuremath{\textup{FO}^2(<,+1)}\xspace}

\newcommand{\efgame}{Ehrenfeucht-Fra\"iss\'e\xspace}

\newcommand\fodeq[1]{\ensuremath{\cong_{#1}}\xspace}
\newcommand\kfodeq{\fodeq{k}}
\newcommand\fodeqp[1]{\ensuremath{\cong^{+}_{#1}}\xspace}
\newcommand\kfodeqp{\fodeqp{k}}

\newcommand\sieq[1]{\ensuremath{\preccurlyeq_{n,#1}}\xspace}
\newcommand\ksieq{\sieq{k}}
\newcommand\sieqp[1]{\ensuremath{\preccurlyeq^{+}_{n,#1}}\xspace}
\newcommand\ksieqp{\sieqp{k}}

\newcommand\simieq[1]{\ensuremath{\preccurlyeq_{n-1,#1}}\xspace}

\newcommand\simieqp[1]{\ensuremath{\preccurlyeq^{+1}_{n-1,#1}}\xspace}
\newcommand\ksimieqp{\simieqp{k}}

\newcommand{\at}{\textup{AT}\xspace}

\newcommand\su{\ensuremath{\textup{SU}}\xspace}
\newcommand\md{\ensuremath{\textup{MOD}}\xspace}

\newcommand\type[1]{\ensuremath{#1\text{-type}}\xspace}
\newcommand\types[1]{\ensuremath{#1\text{-types}}\xspace}
\newcommand\ktype{\type{k}}
\newcommand\ktypes{\types{k}}

\newcommand\eval{\ensuremath{\text{\sc eval}}\xspace}

\newcommand\croch[1]{\ensuremath{\eta(#1)}\xspace}
\newcommand\ucroch[1]{\ensuremath{\gamma(#1)}\xspace}
\newcommand\wcroch[1]{\ensuremath{\left\lceil #1 \right\rceil}\xspace}
\newcommand\pcroch[1]{\ensuremath{\left\lfloor #1 \right\rfloor}\xspace}

\newcommand\tcroch[1]{\ensuremath{\eta_p(#1)}\xspace}

\newcommand\rk[1]{\ensuremath{\textup{rank}(#1)}}

\newcommand\eqsu[1]{\ensuremath{\sim_{#1}}\xspace}
\newcommand\keqsu{\eqsu{k}}

\newcommand\fword{word\xspace}
\newcommand\flang{language\xspace}
\newcommand\fwords{words\xspace}
\newcommand\flangs{languages\xspace}

\newcommand\iword{\ensuremath{\omega}-word\xspace}
\newcommand\ilang{\ensuremath{\omega}-language\xspace}
\newcommand\iwords{\ensuremath{\omega}-words\xspace}
\newcommand\ilangs{\ensuremath{\omega}-languages\xspace}

\newcommand\isemi{\ensuremath{\omega}-semigroup\xspace}
\newcommand\isemis{\ensuremath{\omega}-semigroups\xspace}

\begin{document}
\fancyfoot[RO,LE]{}
\newtheorem{remark}[theorem]{Remark}
\newtheorem{fct}[theorem]{Fact}

\markboth{T. Place and M. Zeitoun}{Adding successor: A transfer theorem for separation and covering}

\title[Adding successor: A transfer theorem for separation and covering]{Adding successor: A transfer theorem for separation and covering}

\author{Thomas Place}
\affiliation{
	\institution{LaBRI, Bordeaux University}
	\country{France}}
\email{tplace@labri.fr}

\author{Marc Zeitoun}
\affiliation{
	\institution{LaBRI, Bordeaux University}
	\country{France}}
	\email{mz@labri.fr}


\begin{abstract}
  Given a class \Cs of word languages, the \Cs-separation problem asks for an algorithm that, given as input two regular languages, decides whether there exists a third language in \Cs containing the first language, while being disjoint from the second. Separation is usually investigated as a means to obtain a deep understanding of the class \Cs.

  In the paper, we are mainly interested in classes defined by logical formalisms. Such classes are often built on top of each other: given some logic, one builds a stronger one by adding new predicates to its signature.  A natural construction is to enrich a logic with the successor relation. In this paper, we present a transfer result applying to this construction: we show that for suitable logically defined classes, separation for the logic enriched with the successor relation reduces to separation for the original logic. Our theorem also applies to a problem that is stronger than separation: covering. Moreover, we actually present two reductions: one for languages of finite words and the other for languages of infinite words.
\end{abstract}

\maketitle
\section{Introduction}
\label{sec:intro}
\noindent
{\bf Context.} A central problem in formal languages theory is to characterize and understand the expressive power of high level specification formalisms. Monadic second order logic (MSO) is such a formalism, which is both expressive and robust. For several classes of structures, such as words or trees, it has the same expressive power as finite automata and defines the class of regular languages~\cite{BuchiMSO,BuchiMSOInf, ElgotMSO,TrakhMSO,Thatcher&Wright:1968,rabin1969}. In this paper, we investigate fragments of MSO over finite and infinite words. In this context, understanding the expressive power of a fragment is often associated to a decision problem: \emph{membership}. Given a logical fragment, one may associate the class \Cs of all word languages that can be defined by a sentence of this fragment. When \Cs is such a class, the \emph{\Cs-membership problem} asks for a decision procedure that tests whether some input regular language belongs to \Cs. Intuitively, setting such an algorithm requires a deep understanding of \Cs: it involves considering \emph{all} languages within \Cs.

Membership has been solved for many natural fragments of MSO, the most prominent one being \fow: first-order logic equipped with a predicate ``$<$'' for the linear ordering. For finite words, the solution was found by Schützenberger, McNaughton and Papert~\cite{sfo,mnpfo}. They characterized the regular languages that are definable in \fow by a syntactic, easily decidable property on a canonical recognizer of this language (such as its minimal automaton or its syntactic monoid). This result was later generalized to infinite words by~Perrin~\cite{pfo}. It now serves as a commonly followed template, which was used successfully to solve membership for many other logical formalisms.

Research on this topic is still ongoing and membership remains open for several fragments. A prominent example is the quantifier alternation hierarchy of first-order logic, which classifies it into levels \sio{n} and \bso{n}. Despite years of investigation, only the lower levels have been solved by Simon~\cite{simonthm}, Pin and Weil~\cite{pwdelta,pwdelta2}, the authors~\cite{pzqalt} and the first author~\cite{pseps3,pseps3j}. Furthermore, making progress has often required moving beyond the standard approach to membership questions. The latest results for the levels \hbox{\bso{2}}, \sio{3} and \sio{4} are based on decision problems that are stronger than membership: \emph{separation} and \emph{covering}. Given a class \Cs, the \emph{\Cs-separation problem} asks for a decision procedure that takes \emph{two} input regular languages and tests whether there exists a third one in \Cs containing the first language while being disjoint from the second one. Covering, which we defined in~\cite{pzcovering,pzcovering2} is even more general: it takes two different objects as input: a regular language $L$ and a finite set of regular languages \Lb. It asks whether there exists a \Cs-cover \Kb of $L$ (\emph{i.e.}, a finite set of languages in \Cs whose union includes $L$) such that no language $K \in \Kb$ intersects all languages in \Lb. Separation is just the special case when $\Lb$ is a singleton. Both problems are decidable for \fow as we showed~\cite{pzfo,pzfoj}, both for finite and infinite~words.

Because of these results, separation and covering have quickly replaced membership as the central question when trying to ``understand'' a given class of languages. However, the main motivation for considering separation and covering is more profound: while harder than membership, they are also more rewarding with respect to the knowledge gained on the investigated class \Cs. Intuitively, a membership algorithm only yields benefits for the languages of \Cs: we are able to detect them and to build a description witnessing this membership. On the other hand, separation and covering algorithms are universal: their benefits apply to \textbf{\emph{all}} languages. An insightful point of view is to see them as \emph{approximation problems}. For example, given an input pair $(L_{1},L_{2})$, the objective of separation is to over-approximate $L_{1}$ by a language in \Cs while $L_{2}$ is the specification of what an acceptable approximation is.

\smallskip

In the paper, we investigate separation and covering for several natural fragments of \fow. Specifically, we consider the levels \sio{n} and \bso{n} in the quantifier alternation hierarchy and the two-variable fragment \fodw. However, we shall not work with these fragments themselves. Instead, we are interested in stronger variants which are built from them in a natural way. A crucial observation is that for these fragments, the drop in expressive power forbids the use of natural relations that could be defined from the linear order in full first-order logic. The main example is ``$+1$'': the \emph{successor relation}. While \fow is powerful enough to express it (``$x+1=y$'' is equivalent to ``$x < y \wedge \neg \exists z (x < z < y)$''), this is not the case for \fodw, \sio{n} and \bso{n}. Hence, there are two natural variants for each of these fragments: a weak one which is only equipped with the linear ordering (denoted \fodw, \sio{n} and \bso{n}) and a strong one which is equipped with additional predicates such as successor (denoted \fodp, \sipe{n} and \bspe{n}). Our objective in this paper is to investigate separation and covering problems associated to strong variants.

\smallskip
\noindent
{\bf State of the art.} Naturally, these strong logical fragments were first investigated using the membership problem. However, this proved to be unexpectedly difficult.  In most cases, even
when the weak variant is known to have decidable membership, proving that this is also the case for the strong one can be highly nontrivial. Examples include the membership proofs of \bspe{1} and \sipe{2}, which involve difficult and intricate combinatorial arguments~\cite{knast83,glasserdd,DBLP:journals/ijfcs/KufleitnerL12} or a wealth of algebraic machinery~\cite{pwdelta2,PinWeilVD}. Another issue is that most proofs directly deal with the strong variant. Given the jungle of such logical fragments, it is desirable to avoid such an approach, treating each variant of the same fragment independently. Instead, a satisfying approach would be to first obtain a solution of the decision problems for the weak variant before lifting it to the strong one via a \emph{generic transfer result}.

This idea has first been investigated by Straubing~\cite{StrauVD} for the membership problem in the setting of finite words. He chose to formulate his approach using algebraic terminology. It is known that any class of languages satisfying appropriate properties is characterized by some algebraic variety~$\mathsf{V}$: a language is in the class if and only if its syntactic monoid belongs to $\mathsf{V}$. This result is the variety theorem of Eilenberg~\cite{EilenbergB}. Straubing's approach was to capture the intuitive connection between weak and strong fragments using a generic operation on algebraic varieties called \emph{wreath product}. Though this is nontrivial, it has been shown that for most logical fragments (including the ones we consider in the paper), if $\mathsf{V}$ is the variety corresponding to the weak variant, then the strong one corresponds to the variety $\mathsf{V} \circ \mathsf{D}$: the wreath product of $\mathsf{V}$ with $\mathsf{D}$ (where the $\mathsf{D}$ is a fixed variety). Thus, Straubing's approach was to show that the operation $\mathsf{V} \mapsto \mathsf{V} \circ \mathsf{D}$ preserves the decidability of membership.

Unfortunately, this is not true in general~\cite{DBLP:journals/ijac/Auinger10}. In fact, while decidability is preserved for all natural logical fragments, there is no generic result that captures them all. In particular, for the less expressive fragments, one has to use completely \emph{ad hoc} proofs. It turns out that in the separation setting, this approach is more robust: it has been shown by Steinberg~\cite{steindelay} that decidability of separation is preserved by the operation $\mathsf{V} \mapsto \mathsf{V} \circ \mathsf{D}$. However, this result has several downsides:

\begin{itemize}
\item Steinberg's theorem is not about separation: it states a purely algebraic property of varieties of the form $\mathsf{V} \circ \mathsf{D}$ (they have ``decidable pointlikes''). The connection with separation is indirect and made with another result by Almeida~\cite{almeidasep}. Therefore, while interesting when already starting from algebra, this approach is less satisfying from a logical point of view: it hides the logical intuitions, while our primary goal is to understand the expressiveness of logics.

\item Going from logic to algebra requires to be acquainted with new notions and vocabulary, as well as involved theoretical tools. One has to manipulate three objects of different nature simultaneously: logic, classes of languages and algebraic varieties. Proofs are also often nontrivial and require a deep understanding of complex objects, which may be scattered in the bibliography.

\item Steinberg's result only applies to classes of languages closed under complement (which excludes the fragments \sic{n} in the quantifier alternation hierarchy). This limitation is tied to the connection with algebraic varieties which only holds for classes closed under complement. While this connection may be lifted to a more general setting~\cite{pinordered,PinWeilVD}, this requires introducing even more algebraic vocabulary.

\item These results are specific to finite words while we intend to investigate both finite and infinite words.
\end{itemize}

\smallskip
\noindent
{\bf Contributions.} We present a new transfer theorem applying to all fragments presented above. For each of them, we show that separation and covering for the strong variant reduce to the same problem for the weak one. Our approach is generic and similar to the original one of Straubing described formerly.  However, rather than choosing algebra to formulate it, we use a pure language theoretic point of view. Specifically, we define a product between classes of languages, called \emph{enrichment}. Given two classes \Cs and \Ds, it builds a new one denoted by $\Cs \circ \Ds$: the \Ds-enrichment of \Cs. As the notation suggests, this operation is designed as the language theoretic counterpart of the wreath product. We then show the two following properties:
\begin{enumerate}
\item For all fragments that we consider, if \Cs is the class corresponding to the weak variant, then the strong one corresponds $\Cs \circ \su$ (\su is a fixed class: the \emph{suffix languages}).
\item Given any class \Cs satisfying standard closure properties, covering and separation for $\Cs \circ \su$ reduce to the same problem for \Cs.
\end{enumerate}

Using such a language theoretic approach has several important benefits over the algebraic one. Let us summarize them.
\begin{itemize}
\item The definition of enrichment is simple, and requires much less machinery than the wreath product. We avoid a lot of algebraic vocabulary, which we do not need. We only work with two objects: logic and classes of languages. The only needed piece of algebra is the elementary definition of regular languages in terms of finite monoids.

\item Our proof is self-contained and much simpler than previous ones. It only relies on basic notions on regular languages. A consequence is that our techniques yield much more intuition on the logical point of view.

\item Enrichment makes sense for \emph{any} class of language, even if it is not closed under complement. Furthermore, closure under complement is not required for applying our reduction theorem. Thus, contrary to~\cite{steindelay} our results capture the \sic{n} levels in the quantifier alternation hierarchy of first-order logic.

\item Our definitions and proofs adapt smoothly to the setting of infinite words. We have two definitions of enrichment and two reduction theorems: the first are for classes of languages of finite words and the second for classes of languages of infinite words.

\item In both settings of finite and infinite words, our results apply to two different problems: separation and covering.
\end{itemize}

It is already known that covering and separation are decidable for the weak variants of many logical fragments. Thus, when combining these algorithms with our results, we shall obtain new separation and covering procedures for several strong variants. Over words, it is known that both problems are decidable for \fodw~\cite{pvzmfcs13,pzcovering,pzcovering2}, \sio{1}~\cite{cmmptsep,pzcovering,pzcovering2}, \bso{1}~\cite{pvzmfcs13,cmmptsep,pzbpol}, \sio{2}~\cite{pzqalt,pzqaltj,pzbpol}, \bso{2}~\cite{pzbpol} and \sio{3}~\cite{pseps3,pseps3j}. Thus, we obtain the decidability of separation and covering for \fodp, \sipe{1}, \bspe{1}, \sipe{2}, \bspe{2} and \sipe{3} over words. Over infinite words the situation is more complicated: while the state of the art is roughly the same as for finite words, many of these results are yet unpublished. It was shown in~\cite{ppzinf16} that separation is decidable for \sio{2} and \sio{3}. Thus, we get that separation is decidable for \sipe{2} and \sipe{3} over infinite~words.

\medskip
\noindent
{\bf Organization of the Paper.} In Section~\ref{sec:prelims}, we set up the notation and present the separation and covering problems. In Section~\ref{sec:logic}, we define the logical fragments that we investigate in the paper. Section~\ref{sec:wfwords} is devoted to our main theorem for languages of finite words:  we define the enrichment operation on classes of languages of finite words, and we show that covering and separation for $\Cs \circ \su$ reduce to the corresponding problem for \Cs. The next two sections are devoted to applying this result to our logical fragments: we do so for two-variable first-order logic in Section~\ref{sec:fo2} and the quantifier alternation hierarchy of first-order logic in Section~\ref{sec:hiera}. Finally, in Section~\ref{sec:theiw}, we generalize our results to the setting of infinite words: we adapt \su-enrichment for classes of languages of infinite words and we lift our reduction theorem to this setting.

\medskip
This paper is the full version of~\cite{pzsucc}. From the conference version, the point of view has been changed from a purely logical one to a language theoretic one with the \su-enrichment operation. In particular, this means that while the underlying ideas are the same, the theorem presented in this full version is more general and applies to all classes built using \su-enrichment. Additionally, the reduction for infinite words is new.

\section{Preliminaries}
\label{sec:prelims}
In this section, we introduce the objects that we investigate in the paper. We first recall basic definitions about (finite and infinite) words and regular \flangs. Then, we present the two decision problems that we consider: covering and separation.

\subsection{Words and classes of \flangs} An alphabet is a \emph{finite} set $A$ of symbols, which are called \emph{letters}. We shall consider both finite and infinite words. Given some alphabet $A$, we denote by $A^+$ the set of all nonempty finite words and by $A^{*}$ the set of all finite words over $A$ (\emph{i.e.}, $A^* = A^+ \cup \{\varepsilon\}$). Moreover, we write $A^\omega$ for the set of all infinite words over $A$. Note that we shall always use the term ``\fword'' to mean a finite word (\emph{i.e.}, an element of $A^*$). On the other hand, we shall speak of an ``\iword'' when considering an infinite word, \emph{i.e.}, an element of $A^\omega$. Finally, we let $A^\infty = A^* \cup A^\omega$.

If $u \in A^*$ and $v \in A^\infty$ we write $u \cdot v \in A^\infty$ or $uv \in A^\infty$ for the concatenation of $u$ and~$v$. Note that if $v \in A^*$, then $uv \in A^*$ and if $v \in A^\omega$, then $uv\in A^\omega$. We shall also consider infinite products. Let $(u_n)_{n \in \nat}$ by a infinite family of \fwords (\emph{i.e.}, $u_n \in A^*$ for all $n \in \nat$), then we may construct a new word or \iword $u_0u_1u_2u_3 \cdots \in A^\infty$ by concatenating them all. Observe that $u_0u_1u_2u_3 \cdots \in A^\omega$ when there are infinitely many indices $n \in \nat$ such that $u_n \neq \varepsilon$. Otherwise, $u_0u_1u_2u_3 \cdots \in A^*$. Finally, when $u \in A^*$ is a single \fword, we denote by $u^\omega \in A^\infty$ the infinite concatenation $uuuu\cdots$.

The \emph{length} of a \fword $u \in A^*$, denoted by $|u|$, is its number of letters. When $u \in A^\omega$ is an \iword, we let $|u| = \infty$. Since we consider logic, we shall often view words and \iwords as linearly ordered sets of labeled \emph{positions}: the domain of a \fword $u \in A^*$ is $\{0,\dots,|u|-1\}$, while the domain of an \iword is simply \nat. In particular, we shall use the following notation. Let $u$ be a \fword or an \iword and let $i,j$ be two integers. We let $u[i,j] \in A^*$ be the following \fword:
\begin{enumerate}
\item If $i \leq j \leq |u|-1$, then $u[i,j]$ is the infix of $u$ obtained by keeping all positions from $i$ to $j$ in $u$. For example, if $u = a_0 \cdots a_{|u|-1}$ is finite, we have $u[i,j] = a_i \cdots a_j$.
\item Otherwise, $u[i,j] = \varepsilon$.
\end{enumerate}

\medskip
\noindent
{\bf Languages and classes.} A  \emph{\flang} over an alphabet $A$ is a subset of $A^*$. Similarly, an \emph{\ilang} is a subset of $A^\omega$. In the paper, we investigate classes of \flangs and classes of \ilangs. A \emph{class of \flangs} \Cs is a map $A \mapsto \Cs(A)$ associating a set $\Cs(A)$ of \flangs over $A$ to each alphabet $A$. Similarly, a \emph{class of \ilangs} is a map $A \mapsto \Cs(A)$ which associates a set $\Cs(A)$ of \ilangs over $A$ to each alphabet $A$.
\begin{remark}
  For the sake of simplifying the presentation, it is usual to abuse notation by making the alphabet implicit: when $A$ is clear from the context, one simply writes $L \in \Cs$ for $L \in \Cs(A)$. Note however that we shall often manipulate distinct alphabets simultaneously.
\end{remark}

In the paper, we work with regular languages. The regular \flangs are those that can be equivalently defined by \emph{nondeterministic finite automata}~(NFA), \emph{finite monoids} or \emph{monadic second-order logic} (\mso) interpreted on \fwords. Similarly, regular \ilangs are those that can be equivalently defined by \emph{nondeterministic Büchi automata} (NBA), \emph{finite \isemis} or \emph{\mso} interpreted on \iwords.  In the paper we work with the algebraic definition of regular \flangs and \ilangs in terms of monoids and \isemis.  We recall these notions in Sections~\ref{sec:wfwords} and~\ref{sec:theiw} respectively.

\subsection{Closure properties}

In the paper, we only consider classes satisfying robust closure properties that we present now. We define them for classes of \flangs (the corresponding definitions for \ilangs are analogous).

\medskip
\noindent
{\bf Boolean operations.} We only consider lattices. A \emph{lattice of \flangs} is a class of \flangs \Cs such that for any alphabet $A$, the two following properties are satisfied:
\begin{itemize}
\item {\it Closure under union.} For any $L_1,L_2 \in \Cs(A)$, we have $L_1 \cup L_2 \in \Cs(A)$. Moreover, $\Cs(A)$ contains the empty union: $\emptyset \in \Cs(A)$.
\item {\it Closure under intersection.} For any $L_1,L_2 \in \Cs(A)$, we have $L_1 \cap L_2 \in \Cs(A)$. Moreover, $\Cs(A)$ contains the empty intersection: $A^* \in \Cs(A)$.
\end{itemize}
A \emph{Boolean algebra of \flangs} is a lattice closed under complement: for any alphabet~$A$, if $L \in \Cs(A)$ then $A^* \setminus L \in \Cs(A)$.

\begin{remark}
  Note that since \ilangs are subsets of $A^\omega$, the empty intersection and complement are interpreted over $A^\omega$ for classes of \ilangs. For example, the empty intersection is $A^\omega$, for any alphabet $A$.
\end{remark}

\medskip
\noindent
{\bf Quotient.} We shall also consider closure under \emph{right quotient} (we do not need left quotient). Consider an alphabet $A$. Given $L \subseteq A^*$ and any $u \in A^*$, we define the right quotient $Lu\inv \subseteq A^*$ of $L$ by $u$ as the \flang,
\[
  L u\inv \stackrel{\text{def}}{=}\{w\in A^* \mid wu \in L\}.
\]
We say that a class of \flangs \Cs is \emph{closed under right quotient} when for any alphabet, any $L \in \Cs(A)$ and any $u \in A^*$, $u\inv L \in \Cs(A)$. We shall not consider closure under quotient for classes of \ilangs.

\medskip
\noindent
{\bf Inverse image.} Finally, we also consider closure under inverse image. For the definition, we need to introduce monoid morphisms. A \emph{semigroup} is a set $S$ equipped with an associative multiplication, written $s \cdot t$ or $st$. A \emph{monoid} is a semigroup $M$ having a neutral element $1_M$, \emph{i.e.}, such that $s \cdot 1_M = 1_M \cdot s = s$ for all $s \in M$. Moreover, a monoid morphism is a mapping $\alpha: M \to N$ from a monoid to another, which respects the algebraic structure: for all $s,s'\in M$, we have $\alpha(s\cdot s')=\alpha(s)\cdot\alpha(s')$ and $\alpha(1_M) = 1_N$. Observe that for any alphabet $A$, the sets $A^+$ and $A^*$ are respectively a semigroup and a monoid when equipped with concatenation (the neutral element of $A^*$ is $\varepsilon$). Therefore, given any two alphabets $A,B$, we may define morphisms $\alpha: A^* \to B^*$.

Given a class of \flangs \Cs, we say that \Cs is \emph{closed under inverse image} when for any two alphabets $A,B$, any morphism $\alpha: A^* \to B^*$ and any \flang $L \in \Cs(B)$, we have $\alpha\inv(L) \in \Cs(A)$. We shall also consider a weaker variant of closure under inverse image: \emph{alphabetic} inverse image. We say that a morphism $\alpha: A^* \to B^*$ is \emph{alphabetic} when $\alpha(a) \in B$ for any letter $a \in A$ (the image of a letter is a letter). A class of \flangs \Cs is \emph{closed under alphabetic inverse image} when for any two alphabets $A,B$, any \emph{alphabetic} morphism $\alpha: A^* \to B^*$ and any \flang $L \in \Cs(B)$, we have $\alpha\inv(L) \in \Cs(A)$.

\medskip

We finish by lifting the definition of inverse image to classes of \ilangs. Observe we may lift any morphism $\alpha: A^* \to B^*$ as a map $\alpha: A^\infty \to B^\infty$. Indeed, if $w \in A^*$, then $\alpha(w)$ is already defined and if $w = a_0a_1a_2 \cdots \in A^\omega$, then we may define,
\[
  \alpha(w) =\alpha(a_0)\alpha(a_1)\alpha(a_2) \cdots \in B^\infty.
\]

\begin{remark}
  Note that when $w \in A^\omega$, $\alpha(w)$ may belong to either $B^\omega$ or $B^*$. This depends on whether there are infinitely many indices $n\in \nat$ such that $\alpha(a_n) \neq \varepsilon$. On the other hand, given an \ilang $L \subseteq B^\omega$, its inverse image $\alpha\inv(L)$ is necessarily an \ilang as well, \emph{i.e.}, a subset of $A^\omega$.
\end{remark}

Given a class of \ilangs \Cs, we say that \Cs is closed under inverse image when for any two alphabets $A,B$, any map $\alpha: A^\infty \to B^\infty$ generated by a morphism and any \ilang $L \in \Cs(B)$, we have $\alpha\inv(L) \in \Cs(A)$.

\subsection{Decision problems} We turn to the two decision problems that we shall consider:  separation and covering. Both of them are parametrized by an arbitrary class of \flangs or \ilangs \Cs and their purpose is to serve as mathematical tools for analyzing \Cs. We only present the definition for classes of \flangs (adapting it to \ilangs  is immediate).

\medskip
\noindent
{\bf Separation.} Given three \flangs $K,L_1,L_2$, we say that $K$ \emph{separates} $L_1$ from $L_2$ if $L_1 \subseteq K$ and $K \cap L_2 = \emptyset$. Furthermore, if \Cs is some class of \flangs and $L_1,L_2$ are two \flangs, we say that $L_1$ is \emph{\Cs-separable} from $L_2$ when there exists $K \in \Cs$ that separates $L_1$ from $L_2$.

\begin{remark}
  Observe that when \Cs is closed under complement, $L_1$ is \Cs-separable from $L_2$ if and only if $L_2$ is \Cs-separable from $L_1$. However, this is not true for classes that are not closed under complement.
\end{remark}

Given a class of \flangs \Cs, we may now define the \Cs-separation problem as follows:

\smallskip
\begin{tabular}{rl}
  {\bf INPUT:}  &  Two regular \flangs $L_1$ and $L_2$. \\[.5ex]
  {\bf OUTPUT:} &  Is $L_1$ $\Cs$-separable from $L_2$?
\end{tabular}
\smallskip

When investigating separation for a particular class \Cs, one usually considers two complementary objectives: finding an algorithm that decides it and finding a generic for constructing a separator in \Cs when there exists one.

\begin{remark}
  Separation generalizes another well-known decision problem: \emph{membership}. Given a class \Cs, this problem asks whether an input regular \flang $L$ belongs to \Cs. This is equivalent to asking whether it is \Cs-separable from its complement (which is also regular). Indeed, in that case, there is only one candidate for being a separator: $L$ itself. In other words, \Cs-membership reduces to \Cs-separation.
\end{remark}

\medskip
\noindent
{\bf Covering.} We now present the covering problem which, we originally introduced in~\cite{pzcovering,pzcovering2} as a natural generalization of separation.

\begin{remark}
  One of the primary motivations for introducing covering is that even if one is only interested in separation, considering covering is required for many classes.
\end{remark}

Given a \flang $L$, a \emph{cover of $L$} is a \emph{\bf finite} set of \flangs \Kb such that $L \subseteq \bigcup_{K \in \Kb} K$. Moreover, given a class \Cs, a \Cs-cover of $L$ is a cover \Kb of $L$ such that all $K \in \Kb$ belong to \Cs. Additionally, given a finite multiset\footnote{We speak of multiset here for the sake of allowing several copies of the same \flang in $\Lb$. This is natural. Indeed, \Lb is an input of our problem: what we have in hand is a set of \emph{recognizers} for the \flangs in \Lb, and distinct recognizers may well define the same \flang.} of \flangs \Lb, we say that a finite set   of \flangs \Kb is \emph{separating} for \Lb if for any $K\in\Kb$, there exists $L \in \Lb$ such that $K \cap L = \emptyset$ (\emph{i.e.}, no element of $\Kb$ intersects all \flangs in \Lb).

Consider a class \Cs. Given a \flang $L_1$ and a finite multiset of \flangs $\Lb_2$, we say that the pair $(L_1,\Lb_2)$ is \Cs-coverable when there exists a \Cs-cover of $L_1$ which is separating for $\Lb_2$. The \emph{\Cs-covering problem} is as follows:

\smallskip
\begin{tabular}{rl}
  {\bf INPUT:}  &  A regular \flang $L_1$ and a finite multiset of regular \flangs $\Lb_2$. \\[.5ex]
  {\bf OUTPUT:} &  Is $(L_1,\Lb_2)$ is $\Cs$-coverable?
\end{tabular}
\smallskip

As for separation, one has usually two goals when investigating \Cs-covering: getting an algorithm that decides it and finding a generic method for building separating \Cs-covers when they exist.
We complete this definition by explaining why covering generalizes separation: the latter is special case of the former when the multiset $\Lb_2$ is a singleton (provided that the class \Cs is a lattice). We state this in the following fact whose proof is easy and given in~\cite{pzcovering2}.

\begin{fct} \label{fct:septocove}
  Let \Cs be a lattice and $L_1,L_2$ two \flangs. Then $L_1$ is \Cs-separable from $L_2$, if and only if $(L_1,\{L_2\})$ is \Cs-coverable.
\end{fct}

\subsection{Suffix languages}\label{sec:suffix-languages} We finish this preliminary section by presenting a specific class of \flangs: the \emph{suffix \flangs} (\su).   While simple, \su will be crucial in the paper: we use it in a generic construction which builds new classes on top of already existing ones.

We first define \su and then present a classification of the \flangs it contains. For any alphabet $A$, $\su(A)$ consists of all finite Boolean combinations of languages of the form $A^*w$ for some $w \in A^*$. It is immediate by definition that \su is a Boolean algebra. We state this in the following proposition.

\begin{proposition} \label{prop:suvar}
  \su is a Boolean algebra.
\end{proposition}

\begin{remark}
  While we shall not need this property, \su is also closed under quotient. On the other hand, it is not closed under inverse image.
\end{remark}

We also consider the following classification of the languages in \su (we call it a \emph{stratification} of \su). For any $k \in \nat$, we define a \emph{finite} class $\su_k$ (\emph{i.e.}, $\su_k(A)$ is a finite set for any $A$). Given an alphabet $A$, $\su_k(A)$ consists of all finite Boolean combinations of languages having the form $A^*w$ for some $w \in A^*$ such that $|w| \leq k$. One may verify that all strata $\su_k$ are finite Boolean algebras. Moreover, for any alphabet $A$, we have:
\[
  \su_k(A) \subseteq \su_{k+1}(A) \text{ for any $k \in \nat$} \qquad \text{and} \qquad \bigcup_{k \in \nat} \su_k(A) = \su(A).
\]
Our motivation for introducing this stratification of \su is the canonical equivalence that one may associate to each stratum $\su_k$. Consider an alphabet $A$. For any natural number $k \in \nat$ and any two words $w,w' \in A^*$, we write $w \keqsu w'$ if and only if the following condition holds:
\[
  \text{For any language $L \in \su_k(A)$,} \quad w\in L \Leftrightarrow w' \in L.
\]
By definition and since $\su_k$ is finite, \keqsu is an equivalence relation of finite index. Finally, since all strata $\su_k$ are Boolean algebras, the following lemma is immediate.

\begin{lemma} \label{lem:canoeq}
  Let $k \in \nat$. Then for any alphabet $A$, the languages in $\su_k(A)$ are exactly the unions of \keqsu-classes.
\end{lemma}

Finally, we shall need the following result which follows Lemma~\ref{lem:canoeq} and gives an alternate definition of \su, which is sometimes simpler to manipulate. One may verify that for any $k \in \nat$, the equivalence classes of \keqsu are all languages of the form $A^*w$ for $|w| = k$ or $\{w\}$ for some $|w| \leq k-1$. Thus, we have the following lemma.

\begin{lemma} \label{lem:eqclasses}
  Let $A$ be an alphabet, let $k \in \nat$ and let $L$ be a language over $A$. Then, $L \in \su_k(A)$ if and only if $L$ is a union of languages having one the two following forms:
  \begin{enumerate}
  \item $A^*w$ for some $w \in A^*$ such that $|w| = k$.
  \item $\{w\}$ for some $w \in A^*$ such that $|w| \leq k-1$.
  \end{enumerate}
\end{lemma}

\section{Fragments of first-order logic}
\label{sec:logic}
In this section, we briefly recall the definition of first-order logic over words and \iwords. Moreover, we introduce the various fragments that we intend to investigate.

\subsection{First-order logic and fragments}

We first briefly define first-order logic. Consider an alphabet $A$. Recall that we view a word $w \in A^*$ as a linearly ordered set of labeled positions $\{0,\dots,|w|-1\}$. In first-order logic (\fow), one can quantify over these positions and use the following predicates:
\begin{itemize}
\item[$\bullet$] \emph{The label predicates}: for each $a \in A$, a	unary predicate ``$a(x)$'' selects all positions labeled with an $a$.
\item[$\bullet$] \emph{Linear order}: a binary predicate ``$x < y$'' interpreted as the (strict) linear order over the positions.
\end{itemize}

Each first-order sentence defines the \flang of all words satisfying it. For example, the sentence ``$\exists x \exists y\ (x < y \wedge a(x) \wedge b(y))$'' defines the \flang $A^*aA^*bA^*$. We shall freely use the name ``\fow'' to denote both first-order logic and the class of \flangs that may be defined by a first-order sentence.

Moreover, \fow also defines a class of \ilangs. Recall that the set of positions in an \iword is simply \nat. Thus, we may interpret \fow sentences on \iwords. For example, the sentence ``$\exists x \exists y\ (x < y \wedge a(x) \wedge b(y))$'' also defines the \ilang $A^*aA^*bA^\omega$. Therefore, \fow defines two classes: a class of \flangs and a class of \ilangs. We speak of \fow over \fwords and \fow over \iwords. We shall adopt a similar terminology for all fragments that we consider.

First-order logic itself is well-understood. The solution to the membership problem for \fow over \fwords is due to~Schützenberger~\cite{sfo}, McNaughton and Papert~\cite{mnpfo}. It is considered as a seminal result for this research field. It was later lifted to \iwords by Perrin~\cite{pfo}. Separation and covering were considered much later. Both problems were solved for \fwords and \iwords by the authors~\cite{pzfo,pzfoj}. However, the focus of our investigation in the present paper is not first-order logic itself. Instead, we are interested in specific fragments of first-order logic that we define now.

\medskip
\noindent
{\bf Two-variable first-order logic.} This fragment is denoted by \fodw. It restricts \fow sentences to those containing at most {\bf two} distinct variables. Note however that these two variables may be reused. For example, the sentence
\[
  \exists x \exists y\  x < y \wedge a(x) \wedge b(y) \wedge (\exists
  x\ y < x \wedge c(x))
\]
\noindent
is an \fodw sentence defining the \flang $A^*aA^*bA^*cA^*$ and the \ilang $A^*aA^*bA^*cA^\omega$. It is folklore and simple to verify that over words, \fodw is a Boolean algebra closed under quotient and inverse image. Similarly, over \iwords, \fodw is a Boolean algebra closed under inverse image.

\medskip
\noindent
{\bf Quantifier alternation hierarchy.} We shall also consider fragments within the quantifier alternation of first-order logic. It is natural to classify first-order sentences by counting the number of alternations between $\exists$ and $\forall$ quantifiers in their prenex normal form. More precisely, given a natural number $n \geq 1$, an \fow sentence is \sio{n} (resp.\ \pio{n}) when its prenex normal form has $(n -1)$ quantifier alternations (that is, $n$ blocks of quantifiers) and starts with an $\exists$ (resp.~ a\ $\forall$) quantifier. For example, a sentence whose prenex normal form is
\[
  \exists x_1 \exists x_2  \forall x_3 \exists x_4
  \ \varphi(x_1,x_2,x_3,x_4) \quad \text{(with $\varphi$ quantifier-free)}
\]
\noindent
is {\sio 3}. Observe that the sets of \sio{n} and \pio{n} sentences are not closed under negation: negating a \sio{n} sentence yields a \pio{n} sentence and vice versa. Thus, one also considers \bso{n} sentences: Boolean combinations of \sio{n} sentences. This yields hierarchies of classes of languages and \ilangs, and both are strict~\cite{BroKnaStrict}. The hierarchy for \fwords is depicted in Figure~\ref{fig:hieraover}. Colors depict the status of each fragment: green ($\mbox{\sio1}$, $\mbox{\sio2}$, $\mbox{\sio3}$, $\mbox{\bso1}$, $\mbox{\bso2}$) means that covering is decidable (hence also separation and membership); blue ($\mbox{\pio1}$, $\mbox{\pio2}$, $\mbox{\pio3}$) means that separation is decidable, while the status for covering is unknown; yellow ($\sio4$, $\pio4$) means that membership is decidable for \fwords and unknown for \iwords\footnote{Actually, for \iwords, very few results on separation and covering have been published, see Section~\ref{sec:theiw}.}, while the status for separation and covering is unknown; finally, red means that even the status for membership is unknown (which is also the case for $\bso4$ and all fragments~above).

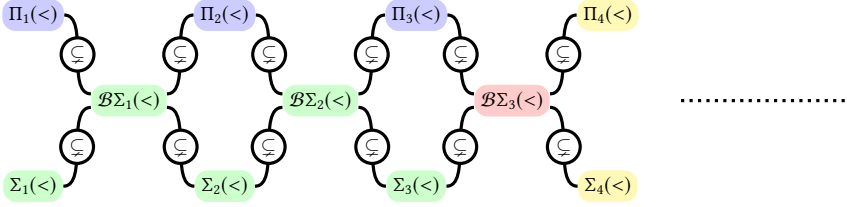
\begin{figure}[!htb]
  \begin{center}

    \begin{tikzpicture}

      \begin{scope}[scale=0.75,every node/.style={scale=0.75}]

        \node[obox] (s4) at (10.2,-1.5) {$\sio{4}$};
        \node[obox] (p4) at (10.2,1.5) {$\pio{4}$};

        \node[rbox] (b3) at (8.5,0.0) {$\bso{3}$};

        \node[gbox] (s3) at (6.8,-1.5) {$\sio{3}$};
        \node[bbox] (p3) at (6.8,1.5) {$\pio{3}$};

        \node[gbox] (b2) at (5.1,0.0) {$\bso{2}$};

        \node[gbox] (s2) at (3.4,-1.5) {$\sio{2}$};
        \node[bbox] (p2) at (3.4,1.5) {$\pio{2}$};

        \node[gbox] (b1) at (1.7,0.0) {$\bso{1}$};

        \node[gbox] (s1) at (0.0,-1.5) {$\sio{1}$};
        \node[bbox] (p1) at (0.0,1.5) {$\pio{1}$};

        \draw[ledg] (s1) to [out=0,in=180,looseness=0.8] node[linc] {$\subsetneq$} (b1.190);
        \draw[ledg] (p1) to [out=0,in=180,looseness=0.8] node[linc] {$\subsetneq$} (b1.170);
        \draw[ledg] (b1.-10) to [out=0,in=180,looseness=0.8] node[linc] {$\subsetneq$} (s2);
        \draw[ledg] (b1.10) to [out=0,in=180,looseness=0.8] node[linc] {$\subsetneq$} (p2);
        \draw[ledg] (s2) to [out=0,in=180,looseness=0.8] node[linc] {$\subsetneq$} (b2.190);
        \draw[ledg] (p2) to [out=0,in=180,looseness=0.8] node[linc] {$\subsetneq$} (b2.170);
        \draw[ledg] (b2.-10) to [out=0,in=180,looseness=0.8] node[linc] {$\subsetneq$} (s3);
        \draw[ledg] (b2.10) to [out=0,in=180,looseness=0.8] node[linc] {$\subsetneq$} (p3);
        \draw[ledg] (s3) to [out=0,in=180,looseness=0.8] node[linc] {$\subsetneq$} (b3.190);
        \draw[ledg] (p3) to [out=0,in=180,looseness=0.8] node[linc] {$\subsetneq$} (b3.170);

        \draw[ledg,dotted] ($(b3)+(3.0,0.0)$) to [out=0,in=180,looseness=0.4]
        ($(b3)+(6.0,0.0)$);

        \draw[ledg] (b3.-10) to [out=0,in=180,looseness=0.8] node[linc] {$\subsetneq$} (s4);
        \draw[ledg] (b3.10) to [out=0,in=180,looseness=0.8] node[linc] {$\subsetneq$} (p4);
      \end{scope}

    \end{tikzpicture}
    \caption{Quantifier alternation hierarchy within \fow.}
    \label{fig:hieraover}
  \end{center}
\end{figure}

It is folklore and simple to verify that over words, all levels \sio{n} (resp. \bso{n}) are lattices (resp. Boolean algebras) closed under quotient and inverse image. Similarly, over \iwords, all levels \sio{n} (resp. \bso{n}) are lattices (resp. Boolean algebras) closed under inverse image. We shall come back to the quantifier alternation hierarchy of \fow in Section~\ref{sec:hiera}.

\subsection{Enriched signatures}

Observe that one may define a seemingly stronger variant \fows of \fow by enriching its signature with the following natural predicates:
\begin{itemize}
\item A binary predicate  ``$x+1 =y$'' interpreted as the successor relation between positions.
\item A unary predicate ``$min(x)$'' which selects the leftmost position.
\item A unary predicate ``$max(x)$'' which selects the rightmost position (in a  finite word).
\item A constant ``$\varepsilon$'' which holds for the empty \fword.
\end{itemize}

\begin{remark}
  Naturally, ``$max$'' and ``$\varepsilon$'' are only useful when interpreting first-order sentences over finite words: \iwords cannot have a rightmost position nor be empty.
\end{remark}

Note however that \fows is only stronger in the \emph{syntactic} sense: it is known and simple to verify that \fow and \fows have the same expressive power, \emph{i.e.},  the corresponding classes of \flangs (resp.\ of \ilangs) are the same. In other words, the four above predicates $\min$, $\max$, ${+}1$ and $\varepsilon$ may be defined from the linear
order. For example $x+1 = y$ is defined by the formula $x < y \wedge \neg(\exists z\ x < z \wedge z <y)$.

\medskip

Nonetheless, this remark is crucial for the paper: we are not interested in first-order logic itself but in its \emph{fragments}. It turns out that for \fodw and the levels \sio{n}, \bso{n} in the quantifier alternation hierarchy, adding the above predicates to the signature yields strictly more expressive logics. Therefore, for each fragment, we are able to define \textbf{two} natural ``variants'': a weak one (whose signature consists of the linear order $<$ and of the letter predicates $a()$) and a strong one (obtained by enriching the signature of the weak one with the above predicates). We write \fodp for the strong variant of \fodw and $\sipe{n},\bspe{n}$ for the strong variants of \sio{n}, \bso{n}.

\begin{example}
  It is not possible to express successor in the two-variables restriction of first-order logic (\fodw). Intuitively, this is because it requires quantifying over a third variable (on the other hand, it is simple to express ``$min$'', ``$max$'' and``$\varepsilon$'').
\end{example}

\begin{remark}
  While the predicates ``$min$'', ``$max$'', ``$\varepsilon$'' are not explicitly mentioned in our notation, they are allowed in all strong variants. We omit them in the notation since they can be defined from ``$<$'' for all fragments except for $\sio{1}, \pio{1}$ and $\bso{1}$.
\end{remark}

Our objective in the paper is to investigate the covering and separation problems associated to strong variants. Our main contribution is a \emph{generic reduction technique}. It is designed to exploit the intuitive relationship between weak and strong variants: we reduce covering for the latter to the same problem for the former. We are then able to obtain covering and separation algorithms for several strong variants as corollaries of already existing results for the corresponding weak variants. In fact, we have two similar reduction theorems: one for finite words (presented in Section~\ref{sec:wfwords}) and one for \iwords (presented in Section~\ref{sec:theiw}).

Let us sketch our approach using the case of finite words. The reduction technique is based on two ingredients. The first one is a generic operation that can be applied to a class of languages: \emph{\su-enrichment} (here, \su denotes the class of suffix languages defined in Section~\ref{sec:prelims}). Given a class of languages \Cs, \su-enrichment constructs a larger class denoted $\Cs \circ \su$. The key idea is that \su-enrichment captures the intuitive relationship between weak and strong variants as a formal and generic connection between the associated classes: if \Cs is the class corresponding to a weak variant, then $\Cs \circ \su$ corresponds to the strong variant.

\begin{remark}
  The connection between weak and strong variants is not a new result, even formalized as a generic operation defined on classes of languages. For example, this was observed by Straubing~\cite{StrauVD} for the quantifier alternation hierarchy. However, an important point is that in the literature, these results are usually formulated using algebraic terminology. In contrast, in this paper, we work directly at the level of language classes. We shall come back to this point in the next section.
\end{remark}

The second ingredient is formulated as a generic reduction theorem. Given a lattice of languages \Cs closed under right quotient and inverse image, it states that $(\Cs \circ \su)$-covering reduces to $\Cs$-covering. By combining the two ingredients, we get a reduction that is generic to all logical fragments outlined above.

\section{Reducing strong to weak variants}
\label{sec:wfwords}
In this section, we present our generic reduction for languages of finite words. First, we define an operation which combines two class of \flangs \Cs and \Ds into a larger one $\Cs \circ \Ds$: the ``\emph{\Ds-enrichment of \Cs}''. This operation is crucial: as we shall prove later, when $\Ds = \su$, it captures the intuitive connection existing between weak and strong logical fragments.

\begin{remark}
  Enrichment is the language theoretic counterpart of an algebraic operation defined between varieties of semigroups: the \emph{wreath product} (see~\cite{StrauVD} for \fwords and~\cite{cartonwreath} for \iwords). In fact, we use the same notation: ``$\circ$''. We have two motivations for using a pure language theoretic point of view here:
  \begin{enumerate}
  \item This is much simpler: we avoid a lot of algebraic machinery.
  \item Manipulating this definition in proofs is more natural since we are only dealing with one object: \emph{classes of languages}. On the other hand, using the algebraic definition requires handling varieties of semigroups and classes of languages simultaneously.
  \end{enumerate}
  Our approach makes it necessary to prove the connection with logic, \emph{i.e.}, that \su-enrichment captures the link between weak and strong fragments.
\end{remark}

In this section, we first define enrichment and then present our reduction theorem: given any class of \flangs \Cs (satisfying appropriate properties), $(\Cs \circ \su)$-covering reduces to \Cs-covering. The remainder of the section will be devoted to proving this result.

\subsection{The enrichment of a class of \flangs} We are now ready to define enrichment. We start with a preliminary notion.

\medskip
\noindent
{\bf \Pb-taggings.} Let $A$ be an alphabet and let \Pb be a finite partition of $A^*$. Observe that $\Pb \times A$ is also a finite set. We use it as an extended alphabet and define a canonical map $\tau_\Pb: A^* \to (\Pb \times A)^*$. Given an arbitrary \fword $u \in A^*$, we denote by $\ppart{u}$ the unique language in the partition \Pb that contains $u$.

Let $w \in A^*$ be a \fword. Consider the decomposition of $w$ as a concatenation of letters: $w = a_1 \cdots a_n \in A^*$. We let $\tau_\Pb(w)$ be the \fword $\tau_\Pb(w) = b_1 \cdots b_n \in (\Pb \times A)^*$ where,
\[
  b_1 = (\ppart{\varepsilon},a_1) \quad \text{~~and~~} \quad b_i = (\ppart{a_1 \cdots a_{i-1}},a_i) \quad \text{for $2 \leq i \leq n$}.
\]
Note that when $w$ is empty, then $\tau_\Pb(\varepsilon) = \varepsilon$. Given any $w \in A^*$, we call $\tau_\Pb(w)$ the \emph{\Pb-tagging} of $w$. Observe that the \Pb-tagging of $w$ is simply a relabeling: each position $i$ in $w$ is given a new label encoding its original label in $A$ and the unique \flang in \Pb containing the prefix $w[1,i-1]$.

\begin{example} \label{ex:tag}
  Let $A = \{a,b\}$ and consider the \flangs $P_\varepsilon = \{\varepsilon\}$, $P_a = A^*a$ and $P_b = A^*b$. Clearly, $\Pb = \{P_\varepsilon,P_a,P_b\}$ is a partition of $A^*$. Let $w = babba \in A^*$, the \Pb-tagging of $w$ is $\tau_\Pb(w) = (P_\varepsilon,b)(P_b,a)(P_a,b)(P_b,b)(P_b,a)$.
\end{example}

\begin{remark}
  The map $w \mapsto \tau_\Pb(w)$ is not a morphism. Moreover, it is not surjective in general. Indeed, there are usually compatibility constraints between consecutive positions in $\tau_\Pb(w)$, as can be observed in Example~\ref{ex:tag} (on the other hand, $\tau_\Pb$ is injective).
\end{remark}

While $\tau_\Pb$ is not a morphism, it will often be convenient to decompose the image $\tau_{\Pb}(w)$ of some \fword $w \in A^*$. For this, we shall use a second map $\delta_\Pb: A^* \times A^* \to A^*$. Let $u,w \in A^*$ and consider the decomposition of $w$ as a concatenation of letters: $w = a_1 \cdots a_n \in A^*$. We let $\delta_\Pb(u,w) = b_1 \cdots b_n \in (\Pb \times A)^*$ where,
\[
  b_1 = (\ppart{u},a_1) \quad \text{~~and~~} \quad b_i = (\ppart{ua_1 \cdots a_{i-1}},a_i) \quad \text{for $2 \leq i\leq n$}.
\]
The following lemma may be verified from the definitions of $\tau_{\Pb}$ and $\delta_\Pb$.

\begin{lemma} \label{lem:deltadef}
  Let $A$ be an alphabet, \Pb a finite partition of $A^*$. Then given any $u \in A^*$ and any $w \in A^*$, we have $\tau_\Pb(uw) = \tau_\Pb(u) \cdot \delta_\Pb(u,w)$.
\end{lemma}

\noindent
{\bf Definition of enrichment.} Consider two classes of \flangs \Cs and \Ds. We define a new class of \flangs $\Cs \circ \Ds$ called the \emph{\Ds-enrichment of \Cs}. Given an alphabet $A$,  a \emph{\Ds-partition of $A^*$} is a \emph{finite} partition of $A^*$ into \flangs of \Ds.

We define $\Cs \circ \Ds$ as the following class of languages. Let $A$ be some alphabet. Given any language $L \subseteq A^*$, we have $L \in (\Cs \circ \Ds)(A)$ if and only if there exists a \Ds-partition \Pb of $A^*$ and \flangs $L_P \in \Cs(\Pb \times A)$ for all $P \in \Pb$ such that,
\[
  L = \bigcup_{P \in \Pb} \left(P \cap \tau_\Pb\inv(L_P)\right).
\]

\begin{remark}
  Strictly speaking, this definition makes sense for any two classes of \flangs \Cs and \Ds. On the other hand, one needs a few hypotheses for it to be robust. Specifically, it is natural to require \Cs to be closed under alphabetic inverse image: since we use \Cs for distinct alphabets, it makes sense to have a property connecting what \Cs is for these distinct alphabets. Similarly, requiring \Ds to be a Boolean algebra is natural since we use \Ds-partitions of $A^*$.

  It turns out that when \Cs is a lattice of \flangs closed under alphabetic inverse image and \Ds is a Boolean algebra of \flangs, $\Cs \circ \Ds$ is a lattice of \flangs containing both \Cs and \Ds (since we never use this property, its proof is left to the reader).
\end{remark}

\begin{example}
  Consider the class \at of alphabet testable \flangs: for any alphabet $A$, $\at(A)$ contains all Boolean combinations of \flangs $A^*aA^*$ for $a \in A$. We describe a \flang in $\at \circ \su$. Let $A =  \{a,b\}$. We claim that $A^*abA^*a$ belongs to $(\at \circ \su)(A)$. Consider the \su-partition of Example~\ref{ex:tag}: $\Pb = \{P_\varepsilon,P_a,P_b\}$ where $P_\varepsilon = \{\varepsilon\}$, $P_a = A^*a$ and $P_b = A^*b$. Clearly, $L = (\Pb \times A)^* \cdot (P_a,b) \cdot (\Pb \times A)^*$ belongs to $\at(\Pb \times A)$. Moreover, one may verify that:
  \[
    A^*abA^*a = (P_\varepsilon \cap \tau_\Pb\inv(\emptyset)) \cup (P_a\cap \tau_\Pb\inv(L)) \cup (P_b\cap \tau_\Pb\inv(\emptyset)) \in \at \circ \su.
  \]
  In fact, $\at \circ \su$ is exactly the class of locally testable \flangs, which is well-known in the literature~\cite{zalclt72,bslt73,mcnlt74}.
\end{example}

\begin{remark}\label{rem:stratification}
  If the class \Ds can be written $\Ds=\bigcup_{k\in\nat}\Ds_k$, then we obtain from the definition that $\Cs\circ\Ds=\bigcup_{k\in\nat}\Cs\circ\Ds_k$. In the case of \su-enrichment, we will use this remark with the natural stratification $(\su_k)_{k\in\nat}$ of \su defined in Section~\ref{sec:suffix-languages}.
\end{remark}

As we explained, we are mainly interested in the special case of \su-enrichment as it captures the intuitive connection between strong and weak logical fragments. Specifically, we prove in Section~\ref{sec:fo2} that \fodp is the \su-enrichment of \fodw. Moreover, we show in Section~\ref{sec:hiera} that for any $n \geq 1$ \sipe{n} and \bspe{n} are respectively the \su-enrichments of \sio{n} and \bso{n}.

\medskip

We now turn to the main theorem of the paper (more precisely, to its variant for finite words): given any lattice of languages \Cs closed under \emph{right quotient} and \emph{inverse image}, $(\Cs \circ \su)$-covering reduces to \Cs-covering. The remainder of the section is devoted to presenting this reduction. Let us start with an outline of the different steps it involves.

\smallskip
The reduction works with the algebraic definition of regular \flangs, which we recall. As we explained, given any alphabet $A$, the universal \flang $A^*$ is a monoid. Given an arbitrary monoid $M$ and a \flang $L \subseteq A^*$, we say that $L$ is \emph{recognized by $M$} if there exists a monoid morphism $\alpha: A^* \rightarrow M$ and a set $F \subseteq M$ such that $L = \alpha\inv (F)$. It is well-know that a \flang is regular if and only if it can be recognized by a \emph{finite} monoid. Moreover, if $L$ is regular, one may compute such a morphism recognizing $L$ from any representation of $L$ (such as an NFA or an \mso sentence).

\smallskip

Consider an input pair $(L,\Lb)$ for the covering problem: $L$ is a regular \flang and \Lb is a finite multiset of regular \flangs. Our reduction requires starting from a single monoid morphism $\alpha:  A^* \to M$ recognizing all \flangs in $\{L\} \cup \Lb$. This is mandatory, as the reduction is parametrized by $\alpha$.

\begin{remark} \label{rem:wfwords:computealpha}
  This requirement is not restrictive: it is simple to build such a morphism. Assume that $\{L\} \cup \Lb =  \{L_1,\dots,L_n\}$.  For each $i \leq n$, we may build a morphism $\alpha_i: A^* \to M_i$ recognizing $L_i$. It then suffices to use the morphism $\alpha: A^* \to M_1 \times \cdots \times M_n$ defined by $\alpha(w) = (\alpha_1(w),\dots,\alpha_n(w))$, which recognizes all \flangs $L_i$.
\end{remark}

Once we have the morphism $\alpha:  A^* \to M$ in hand, we use it in a generic construction which builds two objects. The first one is a new alphabet \awfa: the \emph{alphabet of well-formed \fwords}. The second one is a map $L \mapsto \wfwa{L}$ which associates a new regular \flang over \awfa to any \flang $L \subseteq A^*$ recognized by $\alpha$. We also extend this map to multisets.

Assume that \Cs is a lattice of \flangs closed under right quotient and inverse image. The reduction states that for any pair $(L,\Lb)$ such that all \flangs in $\{L\} \cup \Lb$ are recognized by $\alpha$ , the two following properties are equivalent:
\begin{enumerate}
\item $(L,\Lb)$ is $(\Cs \circ \su)$-coverable.
\item $(\wfwa{L},\wfwa{\Lb})$ is \Cs-coverable.
\end{enumerate}

This concludes our outline. The remainder of this section is organized as follows. We first present the construction which builds an alphabet of well-formed \fwords from an arbitrary morphism. Then, we state the reduction theorem and prove it.

\subsection{Languages of well-formed \fwords} We describe a generic construction which takes as input a morphism $\alpha: A^* \to M$ into a finite monoid $M$. It builds the following objects:
\begin{enumerate}
\item An alphabet \awfa, called the \emph{alphabet of well-formed \fwords associated to $\alpha$}.
\item A map associating to any \flang $L$ over $A^*$ \emph{recognized by $\alpha$} a regular \flang $\wfwa{L}$ over \awfa. We call $\wfwa{L}$ the \emph{\flang of well-formed \fwords associated to $L$}.
\end{enumerate}

We denote by $S$ the semigroup $S = \alpha(A^+)$, that is, the image in $M$ of all \textbf{nonempty} \fwords. Moreover, we write $E(S)$ for the set of idempotent elements in $S$ (\emph{i.e.}, $E(S)$ consists of all $e \in S$ such that $ee = e$). We also write $S^1$ for $S\cup\{1_M\}$ (notice that it may happen that $S=S^1$). Finally, we let ``$\square$'' be some symbol that does not belong to $M$. The \emph{alphabet of well-formed \fwords associated to $\alpha$}, denoted by \awfa, is defined as follows:
\[
  \awfa = (E(S) \cup \{\square\}) ~~\times~~ S^1 ~~\times~~ (E(S) \cup \{\square\}).
\]
Note that since $S$ depends on $\alpha$, so does \awfa. We are not interested in all \fwords of $\awfa^*$, but only in those that are ``well-formed''. Given a \fword $w \in \awfa^*$, we say that $w$ is \emph{well-formed} if it is {\bf nonempty} (\emph{i.e.}, $w \in \awfa^+$) and has the following form:
\[
  w = (\square,s_1,f_1) \cdot (e_2,s_2,f_2) \cdots (e_{n-1},s_{n-1},f_{n-1})	\cdot (e_n,s_n,\square)
\]
with $f_{i} = e_{i+1} \in E(S)$ for all $1\leq i \leq n-1$. In other words,
\[
  w = (\square,s_1,e_2) \cdot (e_2,s_2,e_3) \cdots (e_{n-1},s_{n-1},e_n)	\cdot (e_n,s_n,\square).
\]
In particular, well-formed \fwords of length $1$ are of the form $(\square,s,\square)$ with $s \in S^1$.

\begin{remark}
  The definition requires that $f_{i} = e_{i+1} \in E(S)$ for all $1\leq i \leq n-1$. This means that the idempotents $e_2,\dots,e_n$ must be the image under $\alpha$ of a {\bf nonempty} \fword. This is easy to miss. However, this property is crucial for proving the main theorem.
\end{remark}

It is straightforward to build an automaton recognizing the \flang of well-formed \fwords over \awfa. Thus, the following simple fact is immediate.

\begin{fct} \label{fct:wfwords:wfregular1}
  The \flang of all well-formed \fwords over \awfa is regular.
\end{fct}

We now associate a new \flang over \awfa to each \flang $L$ recognized by $\alpha$: we call it the \emph{\flang of well-formed \fwords associated to $L$}. As the name suggests, it is made exclusively of well-formed \fwords.

One defines a canonical morphism $\eval: \awfa^* \to M$ by defining the image of each letter from \awfa (there are four kinds of such letters). For $s \in S^1$ and $e,f \in E(S)$, we let:
\[
  \left\{
    \begin{array}{lllllll}
      \eval((e,s,f))       & = & esf, & \qquad & \eval((\square,s,f))       & = & sf, \\
      \eval((e,s,\square)) & = & es,  & \qquad & \eval((\square,s,\square)) & = & s.
    \end{array}
  \right.
\]
Consider a \flang $L$ recognized by $\alpha$. The \emph{\flang of well-formed \fwords associated to $L$}, denoted by \wfwa{L}, is defined as follows:
\[
  \wfwa{L} = \bigl \{w \in \awfa^* \mid w \text{ is well-formed and } \eval(w) \in \alpha(L)\bigr\} \subseteq \awfa^+.
\]
Observe that by definition, \wfwa{L} is the intersection of some \flang recognized by \eval with the \flang of well-formed \fwords. Hence, we have the following fact.

\begin{fct} \label{fct:wfwords:wfregular2}
  For any \flang $L \subseteq A^*$ recognized by $\alpha$, $\wfwa{L} \subseteq \awfa^*$ is regular.
\end{fct}

Finally, as explained in the outline, we extend the notation \wfwa{} to multisets, by setting $\wfwa{\Lb} = \{\wfwa{L} \mid L \in \Lb\}$ for any multiset \Lb consisting of \flangs recognized by $\alpha$.

\subsection{\texorpdfstring{Main theorem: reducing $(\Cs \circ \su)$-covering to \Cs-covering}{Main theorem: reducing (C o SU)-covering to C-covering}}

We may now state the theorem that reduces $(\Cs \circ \su)$-covering to \Cs-covering. It is restricted to classes of \flangs \Cs that are nontrivial:  that is, there should exist some alphabet $A$ such that $\Cs(A)$ contains a \flang $L$ which is neither empty nor universal ($L  \neq \emptyset$ and $L \neq A^*$).

\begin{theorem}[Reduction theorem]\label{thm:wfwords}
  Let $\alpha: A^* \to M$ be a morphism and let~\Cs be a nontrivial lattice of \flangs closed under right quotient and inverse image. Moreover, let $L$ be a \flang and let \Lb be a multiset of \flangs, all recognized by $\alpha$. Then, the following properties are equivalent:
  \begin{enumerate}
  \item $(L,\Lb)$ is $(\Cs \circ \su)$-coverable.
  \item $(L,\Lb)$ is $(\Cs \circ \su_{2|M|})$-coverable.
  \item $(\wfwa{L},\wfwa{\Lb})$ is \Cs-coverable.
  \end{enumerate}
\end{theorem}

Before we prove Theorem~\ref{thm:wfwords}, let us discuss its consequences. As announced, the theorem yields a generic reduction from $(\Cs \circ \su)$-covering to \Cs-covering for any nontrivial lattice of \flangs \Cs closed under right quotient and inverse image.

An important remark is that the proof is constructive. Since we intend to use the theorem as a reduction, this is of particular interest for the direction $(3) \Rightarrow (2)$. It is proved by exhibiting a generic construction: given as input a separating \Cs-cover of $(\wfwa{L},\wfwa{\Lb})$, we explain how to build a separating $(\Cs \circ \su_{2|M|})$-cover of $(L,\Lb)$. Thus, we actually get reductions for the two objectives associated to the covering problem: getting an algorithm that decides it, and finding a generic construction for building separating covers.

Note that one may adapt the statement of Theorem~\ref{thm:wfwords} to accommodate natural restrictions of covering, such as separation. Recall that this is the special case of inputs $(L,\Lb)$ where \Lb is a singleton. Thus, we have the following immediate corollary.

\begin{corollary} \label{thm:wfwords:sep}
  Let $\alpha: A^* \to M$ be a morphism and \Cs be a nontrivial lattice of \flangs closed under right quotient and inverse image. Moreover, let $L_1,L_2$ be two \flangs recognized by $\alpha$. Then, the following properties are equivalent:
  \begin{enumerate}
  \item $L_1$ is $(\Cs \circ \su)$-separable from $L_2$.
  \item $L_1$ is $(\Cs \circ \su_{2|M|})$-separable from $L_2$.
  \item \wfwa{L_1} is \Cs-separable from \wfwa{L_2}.
  \end{enumerate}
\end{corollary}

\begin{remark}\label{rem:whatneedstobedone}
  In practice, applying Theorem~\ref{thm:wfwords} to obtain an actual covering or separation algorithm for a given class of \flangs requires clearing the following preliminary steps:
  \begin{enumerate}
  \item Prove that this class is the \su-enrichment $\Cs \circ \su$ of some lattice of \flangs~\Cs closed under right quotient and inverse image.
  \item Obtain a covering or separation algorithm for \Cs.
  \end{enumerate}
  The main point here is that it is usually much simpler to achieve these steps than to obtain directly a covering algorithm for the class. For all examples we shall present, we only take care of the first item and obtain the second from previously known results.
\end{remark}

We shall apply Theorem~\ref{thm:wfwords} to obtain covering and separation algorithms for concrete classes of languages in the next two sections.

We devote the rest of this section to proving Theorem~\ref{thm:wfwords}. We keep our notation: $\alpha: A^* \to M$ is a monoid morphism, $S$ is the semigroup $\alpha(A^+)$, and $S^1=S\cup\{1_M\}$. Recall that the associated alphabet of well-formed \fwords is:
\[
  \awfa = (E(S) \cup \{\square\}) ~~\times~~ S^1 ~~\times~~ (E(S) \cup \{\square\}).
\]
Consider a nontrivial lattice of \flangs \Cs closed under right quotient and inverse image. Our objective is to show that when $L$ and \Lb are respectively a \flang and a multiset of \flangs recognized by $\alpha$, the following properties are equivalent:
\begin{enumerate}
\item $(L,\Lb)$ is $(\Cs \circ \su)$-coverable.
\item $(L,\Lb)$ is $(\Cs \circ \su_{2|M|})$-coverable.
\item $(\wfwa{L},\wfwa{\Lb})$ is \Cs-coverable.
\end{enumerate}

We prove that $(1) \Rightarrow (3) \Rightarrow (2) \Rightarrow (1)$. Observe that the direction $(2) \Rightarrow (1)$ is trivial since $\Cs \circ \su_{2|M|} \subseteq \Cs \circ \su$. Thus, we concentrate on proving that $(1) \Rightarrow (3)$ and $(3) \Rightarrow (2)$.

\subsection{\texorpdfstring{From $(\Cs \circ \su)$-covering to \Cs-covering}{From (C o SU)-covering to C-covering}}

We start with the implication $(1) \Rightarrow (3)$ in Theorem~\ref{thm:wfwords}. The argument is based on the following proposition.

\begin{proposition} \label{prop:wform:firstdir}
  For any $k \geq 1$, there exists a map $\gamma: \awfa^* \to A^*$ satisfying the two following properties:
  \begin{enumerate}
  \item For any $L \subseteq A^*$ recognized by $\alpha$ and any well-formed \fword $w \in \awfa^+$, we have $w \in \wfwa{L}$ if and only if $\ucroch{w} \in L$.
  \item For any \flang $K \in (\Cs \circ \su_k)(A)$, there exists $H_K \in \Cs(\awfa)$ such that for any well-formed \fword $w \in \awfa^+$, we have $w \in H_K$ if and only if $\ucroch{w} \in K$.
  \end{enumerate}
\end{proposition}

Before proving Proposition~\ref{prop:wform:firstdir}, we use it to show $(1) \Rightarrow (3)$ in Theorem~\ref{thm:wfwords}. Consider a \flang $L$ and a multiset of \flangs \Lb, all recognized by~$\alpha$. Assume that $(L,\Lb)$ is $(\Cs \circ \su)$-coverable. We have to prove that $(\wfwa{L},\wfwa{\Lb})$ is \Cs-coverable. By hypothesis on $(L,\Lb)$, we have a separating $(\Cs \circ \su)$-cover $\Kb_{A}$ for $(L,\Lb)$. We use it together with Proposition~\ref{prop:wform:firstdir} to construct a separating \Cs-cover $\Kb_{\awfa}$ for $(\wfwa{L},\wfwa{\Lb})$.

Since $\Kb_{A}$ contains finitely many \flangs all belonging to $\Cs \circ \su=\bigcup_{k\in\nat}\Cs\circ\su_k$ (by Remark~\ref{rem:stratification}), there is some $k \geq 1$ such that $\Kb_A\subseteq\Cs \circ \su_k$. Together with Proposition~\ref{prop:wform:firstdir}, this integer $k$ defines a map $\gamma: \awfa^* \to A^*$. In particular, for any $K \in \Kb_{A}$, Item~(2) of Proposition~\ref{prop:wform:firstdir} yields a \flang $H_K \in \Cs(\awfa)$. We define
\[
  \Kb_{\awfa} = \{H_K \mid K \in \Kb_{A}\}.
\]
To conclude the proof, we show that $\Kb_{\awfa}$ is a separating \Cs-cover for $(\wfwa{L},\wfwa{\Lb})$.

We first prove that $\Kb_{\awfa}$ is a \Cs-cover of $\wfwa{L}$. Let $w \in \wfwa{L}$, we have to find $H \in \Kb_{\awfa}$ such that $w \in H$. Since $w$ is well-formed, we know that $\ucroch{w} \in L$ by the first item in the proposition. Since $\Kb_{A}$ is a cover of $L$, we can find $K \in \Kb_{A}$ such that $\ucroch{w} \in K$. It then follows from the second item in the proposition that $w \in H_K$, which belongs to $\Kb_{\awfa}$ by definition. We conclude that $\Kb_{\awfa}$ is a cover of $\wfwa{L}$. Moreover, it is a \Cs-cover since all \flangs in $\Kb_{\awfa}$ belong to \Cs by Item~(2) of Proposition~\ref{prop:wform:firstdir}.

It remains to prove that $\Kb_{\awfa}$ is separating. Given any $H \in \Kb_{\awfa}$, we have to find $\wfwa{L'} \in \wfwa{\Lb}$ such that $H \cap \wfwa{L'} = \emptyset$. By definition, $H = H_K$ for some $K \in \Kb_{A}$. Moreover, since $\Kb_{A}$ is a separating cover of $(L,\Lb)$, there exists $L' \in \Lb$ such that $K \cap L' = \emptyset$. This entails that $H_K \cap \wfwa{L'} = \emptyset$. Indeed, otherwise, we would have $w \in H_K \cap \wfwa{L'}$ which would imply that $\ucroch{w} \in K \cap L'$ by the two items of Proposition~\ref{prop:wform:firstdir}, a contradiction. This terminates the proof of $(1) \Rightarrow (3)$ in Theorem~\ref{thm:wfwords}.

It now remains to prove Proposition~\ref{prop:wform:firstdir}.

\smallskip
\noindent
{\bf Proof of Proposition~\ref{prop:wform:firstdir}: definition of $\gamma$.} Fix a natural number $k$. We start by defining the map $\gamma:  \awfa^* \to A^*$ and we then show that it satisfies the desired properties. It turns out that $\gamma$ is a morphism. Hence, it suffices describe the image of letters in \awfa.

To any element $s \in \alpha(A^*)$, we associate an arbitrarily chosen \fword $\wcroch{s} \in A^*$ such that $\alpha(\wcroch{s}) = s$.  When $s \in S = \alpha(A^+)$, we require $\wcroch{s}$ to be {\bf nonempty} (note that this implies that $\wcroch{e}\ne\varepsilon$ when $e \in E(S)$, which is crucial in the proof). We are now ready to define our morphism $\gamma: \awfa^* \to A^*$, by defining the image of all four kinds of letters in~\awfa. Given  $s \in S$ and $e,f \in E(S)$, we define,
\[
  \left\{
    \begin{array}{lll}
      \gamma((e,s,f))             & = & \wcroch{e}^{k}\wcroch{s}\wcroch{f}^{k}, \\
      \gamma((\square,s,f))       & = & \wcroch{s}\wcroch{f}^{k},                \\
      \gamma((e,s,\square))       & = & \wcroch{e}^{k}\wcroch{s},                \\
      \gamma((\square,s,\square)) & = & \wcroch{s}.
    \end{array}
  \right.
\]
Now that we defined the morphism $\gamma: \awfa^* \to A^*$, it remains to prove that it satisfies the two properties of Proposition~\ref{prop:wform:firstdir}. We start with the first one, which is simpler.

\smallskip
\noindent
{\bf Proof of Proposition~\ref{prop:wform:firstdir}: first item.} Consider a \flang $L \subseteq A^*$ recognized by~$\alpha$. We have to show that $w \in \wfwa{L}$ iff $\ucroch{w} \in L$, for any well-formed \fword $w \in \awfa^+$.

Since $w$ is well-formed, $w \in \wfwa{L}$ if and only if $\eval(w) \in \alpha(L)$. Moreover, since~$\alpha$ recognizes $L$, $\ucroch{w} \in L$ if and only if $\alpha(\ucroch{w}) \in \alpha(L)$. Hence, it suffices to show $\eval(w) = \alpha(\ucroch{w})$. By definition,
\[
  \begin{array}{cll}
    w      & = & (\square,s_0,e_1) \cdot (e_1,s_1,e_2)\cdots (e_{n-1},s_{n-1},e_n) \cdot (e_n,s_n,\square), \\[1ex]
    \ucroch{w} & = & \wcroch{s_0}\wcroch{e_1}^{2k}\wcroch{s_1}\wcroch{e_2}^{2k} \cdots \wcroch{e_{n-1}}^{2k}\wcroch{s_{n-1}}\wcroch{e_n}^{2k}\wcroch{s_n}.
  \end{array}
\]
Hence, we have:
\[
  \begin{array}{cll}
    \eval(w)          & =  & s_0e_1e_1s_1e_2 \cdots e_{n-1}s_{n-1}e_ne_ns_n, \\[1ex]
    \alpha(\ucroch{w}) & = & s_0(e_1)^{2k}s_1(e_2)^{2k}\cdots (e_{n-1})^{2k}s_{n-1}(e_n)^{2k}s_n.
  \end{array}
\]
Since the $e_i \in E(S)$ are idempotents, we obtain indeed $\eval(w) = \alpha(\ucroch{w})$.

\smallskip
\noindent
{\bf Proof of Proposition~\ref{prop:wform:firstdir}: second item.} For the proof, we fix some arbitrary \flang $K \in (\Cs \circ \su_k)(A)$. We have to build a \flang $H_K \in \Cs(\awfa)$ such that:
\begin{equation} \label{eq:wform:goal1}
  \text{For any well-formed \fword $w \in \awfa^+$,} \quad \text{$w \in H_K$ if and only if $\ucroch{w} \in K$}.
\end{equation}
This is more involved. Recall that by definition of $\Cs \circ \su_k$, we have an $\su_k$-partition \Pb of $A^*$ and languages $L_P \in \Cs(\Pb \times A)$ for all $P \in \Pb$ such that:
\begin{equation}
  \label{eq:k-case2}
  K = \bigcup_{P \in \Pb} \left(P \cap \tau_\Pb\inv(L_P)\right).
\end{equation}

Since $\Cs(\awfa)$ is a lattice, it suffices to treat only two particular cases:
\begin{itemize}
\item $K \in \Pb$, and
\item $K = \tau_\Pb\inv(L_P)$ for some $P \in \Pb$.
\end{itemize}
Indeed, if for each $P\in\Pb$, we are able to exhibit $H^{}_P,H'_P\in \Cs(\awfa)$ such that for any well-formed \fword $w \in \awfa^+$, we have $w \in H_P$ iff $\ucroch{w} \in P$ and $w \in H'_P$  iff $\ucroch{w} \in \tau_\Pb\inv(L_P)$, then for $K$ given by~\eqref{eq:k-case2}, one can choose $K_H=\bigcup_{P \in \Pb} \left(H^{}_P\cap H'_P\right)$. We therefore treat these two cases.

\smallskip
\noindent
\textbf{Case~1}: $K \in \Pb$. Therefore, $K \in \su_k(A)$, since \Pb is an $\su_k$-partition of $A^*$. In this case, the argument is based on the following fact, which follows from the definition of $\gamma$.

\begin{fct} \label{fct:wform:prefixes}
  Consider two well-formed \fwords $w,w' \in \awfa^+$ with the same rightmost letter. Then, $\ucroch{w} \in K$ if and only if $\ucroch{w'} \in K$.
\end{fct}

\begin{proof}
  Since $K \in \su_k(A)$, it follows from the definition of $\su_k$ that given $u \in A^*$, whether $u \in K$ depends only on the suffixes of length at most $k$ in $u$. Moreover, by definition of the map $\gamma$, given a \emph{well-formed} \fword $w \in \awfa^+$, the suffixes of length at most $k$ in $\ucroch{w}$ depend only on the rightmost letter in $w$. The fact is then immediate.
\end{proof}

In view of Fact~\ref{fct:wform:prefixes}, there exist a sub-alphabet $\frB \subseteq \awfa$ such that for any well-formed \fword $w \in \awfa^+$, we have $\ucroch{w} \in K$ if and only if the rightmost letter in $w$ belongs to \frB. Thus, it suffices to define a \flang $H_K \in \Cs(\awfa)$ such that:
\[
  \text{For any well-formed \fword $w \in \awfa^+$,} \quad \text{$w \in H_K$ if and only if $w \in \awfa^* \cdot \frB$}.
\]
It will then be immediate that this \flang $H_K$ satisfies~\eqref{eq:wform:goal1} as desired. It remains to construct $H_K \in \Cs(\awfa)$ satisfying the above property.

For ensuring the condition $H_K\in\Cs(\awfa)$, we use the fact that \Cs is nontrivial. Indeed, this yields an alphabet $D$ such that $\Cs(D)$ contains some \flang $L$ satisfying $L \neq \emptyset$ and $L \neq D^*$. In particular, we have two \fwords $u,v \in D^*$ such that $u \in L$ and $v \not\in L$. Consider the morphism $\eta: \awfa^* \to D^*$ defined as follows. For any letter $b \in \awfa$:
\begin{itemize}
\item If $b$ is of the form $(e,s,\square)$ with $s \in S^1$ and $e \in E(S) \cup \{\square\}$ (\emph{i.e.}, $b$ is used as a rightmost letter in some well-formed \fword), then we define,
  \[
    \eta(b) = \left\{
      \begin{array}{ll}
        u & \text{if $b \in \frB$},   \\
        v & \text{if $b \not\in \frB$}.
      \end{array}
    \right.
  \]
\item Otherwise, $\eta(b) = \varepsilon$.
\end{itemize}
We define $H_K = \eta^{-1}(L)$. Clearly, $H_K \in \Cs(\awfa)$ since \Cs is closed under inverse image. It remains to show that it satisfies the desired property. Let $w \in \awfa^+$ be a well-formed \fword. By definition of well-formed \fwords, $w = w'b$ where $b$ is the unique letter in $w$ of the form $(e,s,\square)$ with $s \in S^1$ and $e \in E(S) \cup \{\square\}$. Thus, it follows that $\eta(w) = u \in L$ if $b \in \frB$ and $\eta(w) = v \not\in L$ otherwise. This exactly says that $w \in H_K$ if and only if $w \in \awfa^*\cdot\frB$, which concludes the proof of this case.

\smallskip
\noindent
\textbf{Case~2}: $K = \tau_\Pb\inv(L_P)$ for some $P \in \Pb$. In this case, the construction of $H_K$ is based on the following lemma.

\begin{lemma} \label{lem:wfwords:wwordsfinal}
  There exists a morphism $\beta: \awfa^* \to (\Pb \times A)^*$ such that for any well-formed \fword $w \in \awfa^+$, we have $\tau_\Pb(\ucroch{w}) = \beta(w)$.
\end{lemma}

Again, Before proving Lemma~\ref{lem:wfwords:wwordsfinal}, we use it to construct $H_K$ and finish the proof of Proposition~\ref{prop:wform:firstdir}. We have a \flang $L_P \in \Cs(\Pb \times A)$ such that $K = \tau_\Pb\inv(L_P)$. We define
\[
  H_K = \beta^{-1}(L_P).
\]
Since \Cs is closed inverse image, it is immediate from the definition that $H_K \in \Cs(\awfa)$. We now prove that $H_K$ satisfies~\eqref{eq:wform:goal1}: for any well-formed \fword $w \in \awfa^+$, $w \in H_K$ if and only if $\ucroch{w} \in K$. We use Lemma~\ref{lem:wfwords:wwordsfinal}. Given a well-formed \fword $w \in \awfa^+$, we have $w \in H_K$ if and only if $\beta(w) \in L_P$. The lemma then says that this is equivalent to $\tau_\Pb(\ucroch{w}) \in L_P$, \emph{i.e.}, to $\ucroch{w} \in K$ by hypothesis on $K$.

\smallskip

It remains to prove Lemma~\ref{lem:wfwords:wwordsfinal}. Let us first define the morphism $\beta: \awfa^* \to (\Pb \times A)^*$. We use the map $\delta_\Pb: A^* \times A^* \to (\Pb \times A)^*$ that we defined at the beginning of the section.

Given any letter $(e,s,f) \in \awfa$, we define its image $\beta((e,s,f))$. There are two cases depending on whether $e = \square$ or $e \in E(S)$.

\begin{enumerate}
\item If $e = \square$, we define $\beta((\square,s,f)) = \tau_\Pb(\gamma((\square,s,f)))$.
\item If $e \in E(S)$, we define $\beta((e,s,f)) = \delta_\Pb(\wcroch{e}^k,\gamma((e,s,f)))$.
\end{enumerate}

It remains to show that for any well-formed \fword $w \in \awfa^+$, we have $\tau_\Pb(\ucroch{w}) = \beta(w)$. We have $w = b_1 \cdots b_n$ with $b_1,\dots,b_n \in \awfa$. We show that for any $\ell \leq n$, we have:
\[
  \tau_\Pb(\ucroch{b_1 \cdots b_\ell}) = \beta(b_1 \cdots b_\ell).
\]
We shall argue by induction on~$\ell$. The case $\ell = n$, will then yield the desired result. In the base case $\ell = 1$, since $w$ is well-formed, we know that $b_1 = (\square,s,f)$ for some $s \in S^1$ and $f \in E(S) \cup \{\square\}$. Thus, it is immediate by definition of $\beta$ that we have $\tau_\Pb(\ucroch{b_1}) = \beta(b_1)$. Assume now that $\ell \geq 2$. Since $\beta$ is a morphism, we have
\[
  \beta(b_1 \cdots b_\ell) = \beta(b_1 \cdots b_{\ell-1}) \cdot \beta(b_\ell).
\]
It then follows from the induction hypothesis that,
\[
  \beta(b_1 \cdots b_\ell) = \tau_\Pb(\ucroch{b_1 \cdots b_{\ell-1}}) \cdot \beta(b_\ell).
\]
Since $w$ is well-formed, we know that there exist $f,g \in E(S) \cup \{\square\}$, $s,t \in S^1$ and $e \in E(S)$ such that $b_{\ell-1} = (g,t,e)$ and $b_{\ell} = (e,s,f)$. Hence, $\beta(b_\ell) = \delta_\Pb(\wcroch{e}^k,\gamma(b_\ell))$. Moreover, since $b_{\ell-1} = (g,t,e)$, it follows from the definition of $\gamma$ that $\wcroch{e}^k$ is a suffix of $\ucroch{b_1 \cdots b_{\ell-1}}$. Thus, since $\Pb$ is a $\su_k$-partition of $A^*$ and $\wcroch{e}^k$ has length at least $k$ (this is where $\wcroch{e}$ being nonempty is crucial), we have $\ppart{\wcroch{e}^k}=\ppart{\ucroch{b_1 \cdots b_{\ell-1}}}$. Hence,
\[
  \beta(b_\ell) = \delta_\Pb(\wcroch{e}^k,\gamma(b_\ell)) = \delta_\Pb(\ucroch{b_1 \cdots b_{\ell-1}},\gamma(b_\ell)).
\]
Altogether, this yields,
\[
  \beta(b_1 \cdots b_\ell) = \tau_\Pb(\ucroch{b_1 \cdots b_{\ell-1}}) \cdot \delta_\Pb(\ucroch{b_1 \cdots b_{\ell-1}},\gamma(b_\ell)).
\]
By Lemma~\ref{lem:deltadef}, this says that $\beta(b_1 \cdots b_\ell) = \tau_\Pb(\ucroch{b_1 \cdots b_{\ell}})$, which concludes the proof.

\subsection{\texorpdfstring{From \Cs-covering to $(\Cs \circ \su)$-covering}{From C-covering to (C o SU)-covering}}

We now turn to the direction $(3) \Rightarrow (2)$ in Theorem~\ref{thm:wfwords}. The argument is based on the following proposition which states a generic property of the morphism $\alpha: A^* \to M$.

\begin{proposition} \label{prop:wform:secdir}
  There exists a map $\eta: A^* \to \awfa^*$ such that:
  \begin{enumerate}
  \item For any $L \subseteq A^*$ recognized by $\alpha$, we have $L = \eta\inv(\wfwa{L})$.
  \item For any $K \in \Cs(\awfa)$, we have $\eta\inv(K) \in (\Cs \circ\su_{2|M|})(A)$.
  \end{enumerate}
\end{proposition}

Before we prove Proposition~\ref{prop:wform:firstdir}, we use it to finish the proof of Theorem~\ref{thm:wfwords}. Consider a \flang $L$ and a multiset of \flangs \Lb, all recognized by $\alpha$. Moreover, assume that $(\wfwa{L},\wfwa{\Lb})$ is \Cs-coverable. We have to show that $(L,\Lb)$ is $(\Cs \circ \su_{2|M|})$-coverable.

By hypothesis, there exists a separating \Cs-cover $\Kb_{\awfa}$ of $(\wfwa{L},\wfwa{\Lb})$, we use it to build a separating $(\Cs \circ\su_{2|M|})$-cover $\Kb_{A}$ of $(L,\Lb)$. Consider the map $\eta: A^* \to \awfa^*$ given by Proposition~\ref{prop:wform:secdir}. We define,
\[
  \Kb_{A} = \{\eta\inv(K) \mid K \in \Kb_{\awfa}\}.
\]
To conclude the proof, we show that $\Kb_{A}$ is a separating $(\Cs \circ\su_{2|M|})$-cover of $(L,\Lb)$.

Let us first prove that $\Kb_{A}$ is a $(\Cs \circ \su_{2|M|})$-cover of $L$. Let $w \in L$, we first have to find $H \in \Kb_{A}$ such that $w \in H$. Since $w \in L$, we know from the first item in Proposition~\ref{prop:wform:secdir} that $\croch{w} \in \wfwa{L}$. Hence, since $\Kb_{\awfa}$ is a cover of \wfwa{L}, there exists $K \in \Kb_{\awfa}$ such that $\croch{w} \in K$. It now follows that $w$ belongs to $\eta\inv(K) \in \Kb_{A}$. We conclude that $\Kb_{A}$ is a cover of $L$. Moreover, it is a $(\Cs \circ \su_{2|M|})$-cover by Item~(2) in Proposition~\ref{prop:wform:secdir}.

We now prove that $\Kb_{A}$ is separating. Let $H \in \Kb_{A}$, we have to find $L' \in \Lb$ such that $L' \cap H = \emptyset$. By definition $H = \eta\inv(K)$ for some $K \in \Kb_{\awfa}$. Since $\Kb_{\awfa}$ is a separating for $(\wfwa{L},\wfwa{\Lb})$, we know that there exists $L \in \Lb$ such that $K \cap \wfwa{L} = \emptyset$. It is immediate that $\eta\inv(K) \cap L = \emptyset$ since $L = \eta\inv(\wfwa{L})$ by Item~(1) in Proposition~\ref{prop:wform:secdir}.

It remains to prove Proposition~\ref{prop:wform:secdir}, to which we devote the rest of this section.

\smallskip
\noindent
{\bf Proof of Proposition~\ref{prop:wform:secdir}: definition of $\eta$.} We begin by defining the map $\eta: A^* \to \awfa^*$. Let us point out that $\eta$ is {\bf not} be a morphism (otherwise, since \Cs is closed under inverse image, all $\eta^{-1}(K)$ would belong to $\Cs(A)$, which is not the case in general). We start with a preliminary definition.

Given a \fword $w$, a position $x$ in $w$ (\emph{i.e.}, $x\in \{0,1,\ldots,|w|-1\}$) and a natural number $k \in \nat$, we define the \emph{\ktype} of $x$ as the following \fword of length at most $k$:
\begin{itemize}
\item If $x < k$, then the \ktype of $x$ is the prefix $w[0,x-1]$ of length $x$.
\item If $x \geq k$, then the \ktype of $x$ is the infix $w[x-k,x-1]$ of length $k$.
\end{itemize}

For the construction of $\eta$, we fix $k = |M|$. Moreover, we choose an arbitrary order on the set of idempotents $E(S)$ (recall that $S = \alpha(A^+)$).

Consider a nonempty \fword $w \in A^+$ and a position $x$ in $w$. We say that $x$ is \emph{distinguished} when there exists an idempotent $e \in E(S)$ such that the \ktype $u$ of $x$ satisfies $\alpha(u) \cdot e = \alpha(u)$. The following fact states that distinguished positions are frequent.

\begin{fct} \label{fct:wform:herecomesdis}
  Let $w \in A^+$ be such that $|w| \geq k$ and let $y \geq k-1$ be a position in $w$. Then, there exists a distinguished position $x$ in $w$ such that $y-(k-1) \leq x \leq y$.
\end{fct}

\begin{proof}
  This follows from the pigeonhole principle. By definition, $w[y-(k-1),y]$ is a \fword $a_1\cdots a_k$ of length $k$. For all $1 \leq j \leq k$, we define $w_j = a_1 \cdots a_j$. Moreover, we let $w_0 = \varepsilon$. By definition, we have $k = |M|$. Thus, we obtain from the pigeonhole principle that there exist $0 \leq j_1 < j_2 \leq k$ such that, $\alpha(w_{j_1}) = \alpha(w_{j_2})$.

  We claim that the position $x = y - (k - j_1-1)$ in $w$ is distinguished. Indeed, by definition, we have $w_{j_2} = w_{j_1} a_{j_1+1} \cdots a_{j_2}$. Therefore,
  \[
    \begin{array}{lll}
      \alpha(w_{j_1}) & = & \alpha(w_{j_2}) \\
                 & = & \alpha(w_{j_1}) \cdot \alpha(a_{j_1 + 1} \cdots a_{j_2}) \\
                 & = & \alpha(w_{j_1}) \cdot (\alpha(a_{j_1 + 1} \cdots a_{j_2}))^p \quad \text{for all $p \in \nat$}.
    \end{array}
  \]
  It is standard (and easy to check) that there exists $p \geq 1$ such that $(\alpha(a_{j_1 + 1} \cdots a_{j_2}))^p$ is an idempotent $e \in E(S)$. Therefore, since $w_{j_1}$ is a suffix of the \ktype of the position $x = y - (k -j_1-1)$, we know that $x$ is distinguished, as witnessed by the idempotent~$e$.
\end{proof}

We are now ready to define the map $\eta: A^* \to \awfa^*$. Consider a \fword $w \in A^*$. If $w$ does not contain any distinguished position, we define,
\[
  \croch{w} = (\square,\alpha(w),\square).
\]
Otherwise, $w$ has $n \geq 1$ distinguished positions, say $x_0 < \cdots < x_{n-1}$. We let $u_0,\cdots,u_{n-1}$ be their respective \ktypes. Finally, let $e_0,\dots,e_{n-1}\in E(S)$ be such that for all $i \geq 0$, $e_i$ is the smallest idempotent (according to the arbitrary order that we fixed over $E(S)$) such that $\alpha(u_{i}) \cdot e_i = \alpha(u_{i})$. We define $\croch{w} \in \awfa^+$ as the following well-formed \fword:
\begin{equation}\label{eq:crochw}
  \croch{w} = (\square,\alpha(w_0),e_0)\cdot (e_0,\alpha(w_1),e_1) \cdots (e_{n-2},\alpha(w_{n-1}),e_{n-1})\cdot (e_{n-1},\alpha(w_{n}),\square)
\end{equation}
\noindent
where $w_0,\dots,w_{n}$ are the unique \fwords such that $w$ may be decomposed as $w = w_0w_1 \cdots w_{n}$ where for all $i\geq0$, the word $w_{i+1}$ starts at position $x_i$ in $w$. In other terms, $w_0 =  w[0,x_0-1]$, $w_i = w[x_{i-1},x_{i}-1]$ for $1 \leq i \leq n-1$ and $w_{n} = w[x_{n-1},|w|-1]$. Observe that for any $w \in A^*$, $\croch{w} \in \awfa^+$ is well-formed by construction.

It remains to prove that this definition satisfies the two items in Proposition~\ref{prop:wform:secdir}.

\smallskip
\noindent
{\bf Proof of Proposition~\ref{prop:wform:secdir}: first item.} Consider a \flang $L \subseteq A^*$ recognized by $\alpha$, we have to show that $L = \eta\inv(\wfwa{L})$. This amounts to proving that for any $w \in A^*$, $w \in L$ if and only if $\croch{w} \in \wfwa{L}$.

Since $\alpha$ recognizes $L$, $w \in L$ if and only if $\alpha(w) \in \alpha(L)$. Moreover, since $\croch{w}$ is well-formed, by definition, $\croch{w} \in \wfwa{L}$ if and only if $\eval(\eta(w)) \in \alpha(L)$. Hence, it suffices to prove that $\eval(\croch{w}) = \alpha(w)$. This is immediate from the definition if $w$ has no distinguished position. Otherwise, $w$ may be decomposed as $w = w_0 \cdots w_{n}$ and~\eqref{eq:crochw} holds.
By choice of the idempotents used in the construction, we have $\alpha(w_1\cdots w_{i}) \cdot e_i =\alpha(w_1\cdots w_{i})$ for all $i \geq 0$. Hence, it is immediate from a simple induction that
\[
  \eval(\croch{w}) = \alpha(w_0)\cdot e_0 \cdot \alpha(w_1) \cdot e_1 \cdot \alpha(w_2) \cdot \cdots \cdot e_{n-1}\cdot \alpha(w_{n}) =  \alpha(w_1 \cdots w_{n}) = \alpha(w).
\]

\smallskip
\noindent
{\bf Proof of Proposition~\ref{prop:wform:secdir}: second item.} Given an arbitrary \flang $K \in \Cs(\awfa)$, we have to prove that $\eta\inv(K)$ belongs to $(\Cs \circ \su_{2k})(A)$ (recall that we fixed $k = |M|$).

By definition of $\Cs \circ \su_{2k}$, the first thing we have to do is choose some $\su_{2k}$-partition of $A^*$. Recall that $\eqsu{2k}$ denotes the canonical equivalence associated to $\su_{2k}$: given $w,w' \in A^*$, we have $w \eqsu{2k} w'$ when $w \in L \Leftrightarrow w' \in L$ for any $L \in \su_{2k}$. We denote by \Pb the partition of $A^*$ into $\eqsu{2k}$-classes. By Lemma~\ref{lem:canoeq}, \Pb is a $\su_{2k}$-partition of $A^*$. It now remains to exhibit languages $L_P \in \Cs(\Pb \times A)$ for all $P \in \Pb$ such that
\[
  \eta\inv(K) = \bigcup_{P \in \Pb} \left(P \cap \tau_\Pb\inv(L_P)\right).
\]
We start with preliminary definitions. We know that for any \fword $w \in A^*$ having $n$ distinguished positions, $\croch{w} \in \awfa^+$ has length $n+1$ and is built by decomposing $w$ according to these distinguished positions. We let $\tcroch{w} \in \awfa^*$ be the (possibly empty) prefix of \croch{w} made of the first $n$ letters of $\croch{w}$ (\tcroch{w} is not well-formed: the third component of the rightmost letter is not ``$\square$''). The argument is now based on the two following lemmas.

\begin{lemma} \label{lem:canonicone}
  Let $P \in \Pb$. Then there exists a letter $b_P \in \awfa$ such that for any $w \in P$, the rightmost letter in $\croch{w}$ is $b_P$, \emph{i.e.}, $\croch{w} = \tcroch{w} \cdot b_P$.
\end{lemma}

\begin{lemma} \label{lem:canonictwo}
  There exists a morphism $\beta: (\Pb \times A)^* \to \awfa^*$ such that for any $w \in A^*$, $\beta(\tau_\Pb(w)) = \tcroch{w}$.
\end{lemma}

Before proving these lemmas, let us use them to finish the proof that $\eta\inv(K)$ belongs to $(\Cs \circ \su_{2k})(A)$. For any $P \in \Pb$, we let $b_P \in \awfa$ be as defined in Lemma~\ref{lem:canonicone}. Moreover, let $\beta: (\Pb \times A)^* \to \awfa^*$ be the morphism described in Lemma~\ref{lem:canonictwo}. We claim that:
\[
  \eta\inv(K) = \bigcup_{P \in \Pb} \left(P \cap \tau_\Pb\inv(\beta\inv(K(b_P)\inv))\right).
\]
This concludes the proof: since \Cs is closed under right quotient and inverse image,  we know that for any $P \in \Pb$, $\beta\inv(K(b_P)\inv) \in \Cs(\Pb \times A)$. Thus,   it is immediate that $\eta\inv(K)$ belongs to $(\Cs \circ \su_{2k})(A)$ by definition.

Let us prove the claim. Consider a \fword $w \in A^*$ and let $P$ be the unique language in the partition \Pb of $A^*$ such that $w \in P$. It suffices to show that $w \in \eta\inv(K)$ if and only if $w \in \tau_\Pb\inv(\beta\inv(K(b_P)\inv))$. By Lemma~\ref{lem:canonicone}, we know that $w \in \eta\inv(K)$ if and only if $\tcroch{w} \cdot b_P \in K$, \emph{i.e.}, $\tcroch{w} \in K(b_P)\inv$. Finally, since $\beta(\tau_\Pb(w)) = \tcroch{w}$ by Lemma~\ref{lem:canonictwo}, this is equivalent to $w \in \tau_\Pb\inv(\beta\inv(K(b_P)\inv))$, which concludes the proof.

\smallskip
It remains to prove Lemmas~\ref{lem:canonicone} and \ref{lem:canonictwo}. 

\smallskip
\noindent
{\it Proof of Lemma~\ref{lem:canonicone}.} Since any $P \in \Pb$ is by definition a \eqsu{2k}-class, this amounts to proving that given $w,w' \in A^*$ such that $w \eqsu{2k} w'$, $\croch{w}$ and $\croch{w'}$ have the same rightmost letter. We consider two possible cases.

If $w$ has no distinguished position, then we have $|w| < k$  by Fact~\ref{fct:wform:herecomesdis}, hence $\{w\} \in \su_{2k}$. Since $w \eqsu{2k} w'$, we have $w' \in \{w\}$, \emph{i.e.}, $w = w'$.  The result is now immediate.

Assume on the contrary that $w$ contains at least one distinguished position. We let $x$ be the rightmost one, and $u$ be the \ktype of $x$. By definition, the rightmost letter in $\croch{w}$ is $(e,\alpha(v),\square)$ where $v= w[x,|w|-1]$ and $e \in E(S)$ the smallest idempotent such that $\alpha(u) \cdot e = \alpha(u)$. Note that $uv$ is a suffix of $w$ by definition. Since $x$ is the rightmost distinguished position by definition, it follows from Fact~\ref{fct:wform:herecomesdis} that $|v| \leq k$ (otherwise, there would be another distinguished position strictly to the right of $x$). It follows that $|uv| \leq 2k$. Thus, since $w \eqsu{2k} w'$, $uv$ is a suffix of $w'$ as well. It now follows from the definitions that the rightmost letter in $w'$ must be $(e,\alpha(v),\square)$ as well. This concludes the proof of Lemma~\ref{lem:canonicone}.

\smallskip
\noindent
{\it Proof of Lemma~\ref{lem:canonictwo}.} Let us start with a few simple observations. Consider some \fword $w \in A^*$. By definition, if $w$ has no distinguished position, then $\tcroch{w} = \varepsilon$. Otherwise, $w$ has $n \geq 1$ distinguished positions $x_0 <\cdots < x_{n-1}$ and,
\[
  \tcroch{w} = (\square,\alpha(w_0),e_0)\cdot (e_0,\alpha(w_1),e_1) \cdots
  (e_{n-2},\alpha(w_{n-1}),e_{n-1})
\]
where $w_0 =  w[0,x_0-1]$, $w_i = w[x_{i-1},x_{i}-1]$ for $1 \leq i \leq n-1$. Note that there is a natural bijection between the distinguished positions of $w$ and the positions of $\tcroch{w}$, which associates to any distinguished position $x_i$ in $w$ the position $\pcroch{x_i} = i$ in $\tcroch{w}$.

By definition any position $x$ in $w$ may also be viewed as a position of $\tau_\Pb(w) \in (\Pb \times A)^*$. Because of our choice of $k$ as $|M|$, the map $w \mapsto \tcroch{w}$ is designed so that for any position $x$ in $w$, whether $x$ is distinguished and if so the label of $\pcroch{x}$ in $\tcroch{w}$ depends only on the label of $x$ in $\tau_\Pb(w)$. Let us state this property in the following lemma.

\begin{lemma} \label{lem:wform:canonic}
  For any letter $(P,a) \in \Pb \times A$, one of the two following properties hold:
  \begin{enumerate}
  \item For any $w \in A^*$ and any position $x$ in $w$, if $x$ has label $(P,a)$ in $\tau_\Pb(w)$, then $x$ is not distinguished.
  \item There is a letter $c_{(P,a)} \in \awfa$ such that for any $w \in A^*$ and any position $x$ in~$w$, if $x$ has label $(P,a)$ in $\tau_\Pb(w)$, then $x$ is distinguished and $\pcroch{x}$ has label $c_{(P,a)}$ in~$\tcroch{w}$.
  \end{enumerate}
\end{lemma}

Before proving Lemma~\ref{lem:wform:canonic}, we first use it to define the morphism $\beta: (\Pb \times A)^* \to \awfa^*$ and finish the argument for Lemma~\ref{lem:canonictwo}. We have to define the image of each letter in $\Pb \times A$. Let $(P,a) \in \Pb \times A$ be a letter.
\begin{itemize}
\item If $(P,a)$ satisfies the first item Lemma~\ref{lem:wform:canonic}, we let $\beta((P,a)) = \varepsilon$.
\item If $(P,a)$ satisfies the first second item in Lemma~\ref{lem:wform:canonic},  we let $\beta((P,a)) = c_{(P,a)}$.
\end{itemize}
It is now immediate from Lemma~\ref{lem:wform:canonic} that $\beta$ satisfies the desired property: for any $w \in A^*$, $\beta(\tau_\Pb(w)) = \tcroch{w}$.

\smallskip

It remains to prove Lemma~\ref{lem:wform:canonic}. By definition of the partition \Pb, this amounts to proving that given $w,w' \in A^*$ and $x,x'$ positions in $w,w'$ such that $w[0,x-1] \eqsu{2k} w'[0,x'-1]$, $x$ is distinguished if and only if $x'$ is distinguished and in that case, \pcroch{x} and \pcroch{x'} in \tcroch{w} and \tcroch{w'} have the same label. If $w[0,x-1] \eqsu{2k} w'[0,x'-1]$, then $x$ and $x'$ have the same \ktype, and since the \ktype of a position determines whether it is distinguished or not, $x$ and $x'$ are either both distinguished, or none of them is.

We now concentrate on the second property. Assume that $x$ and $x'$ are distinguished, we show that the positions \pcroch{x} and \pcroch{x'} in \tcroch{w} and \tcroch{w'} carry the same label. We define $u$ as the common \ktype of $x$ and $x'$ and $e \in E(S)$ as the smallest idempotent such that $\alpha(u) \cdot e = \alpha(u)$. We distinguish two cases.

Assume first that $x$ is the leftmost distinguished position in $w$. It follows that the label of $\pcroch{x}$ is $(\square,\alpha(w[0,x-1]),e)$. We have to show that this is also the label of \pcroch{x'}. By hypothesis, we know that the prefix $w[0,x-1]$ contains no distinguished position. Thus, Fact~\ref{fct:wform:herecomesdis} yields that $w[0,x-1] < k$. Since $w[0,x-1] \eqsu{2k} w'[0,x'-1]$, it is then immediate that $w[0,x-1] = w'[0,x'-1]$. We conclude that $x'$ is also the leftmost distinguished position of $w'$ and that $\pcroch{x},\pcroch{x'}$ have the same label, namely $(\square,\alpha(w[0,x-1]),e)$.

Assume now that $x$ is not the leftmost distinguished position in $w$. Let $y$ be the distinguished position which directly precedes $x$ in $w$. Furthermore, let $v$ be the \ktype of $y$ and $f \in E(S)$ be the smallest idempotent such that $\alpha(v) \cdot f = \alpha(v)$.	By definition, the label of \pcroch{x} is $(f,\alpha(w[y,x-1]),e)$. We have to show that this is also the label of \pcroch{x'}. By Fact~\ref{fct:wform:herecomesdis}, we have $|w[y,x-1]| \leq k$ (otherwise, there would be a third distinguished position in $w$ strictly between $y$ and $x$, contradicting the definition of $y$). Therefore, $|v \cdot w[y,x-1]| \leq 2k$. Moreover, $v \cdot w[y,x-1]$ is a suffix of $w[0,x-1]$ by definition. Since $w[0,x-1] \eqsu{2k} w'[0,x'-1]$, we obtain that $v \cdot w[y,x-1]$ is also a suffix of $w'[0,x'-1]$. This shows that $\pcroch{x'}$ has label $(f,\alpha(w[y,x-1]),e)$ as well, which concludes the proof.

\section{Application to two-variable first-order logic}
\label{sec:fo2}
This is the first of two sections in which we illustrate Theorem~\ref{thm:wfwords} and use it to obtain algorithms for a particular class of languages. Here, we consider the two-variable fragment of first-order logic over words (defined in Section~\ref{sec:logic}). Specifically, we show that the covering is decidable for the strong variant: \fodp.  Let us state this result.

\begin{corollary}[of Theorem~\ref{thm:wfwords}] \label{cor:fod}
  Covering and separation are decidable for \fodp over words.
\end{corollary}

As we explained in Remark~\ref{rem:whatneedstobedone}, using Theorem~\ref{thm:wfwords} to obtain Corollary~\ref{cor:fod} requires clearing two preliminary steps.  First we need to that \fodp is the \su-enrichment of some lattice closed under right quotient and inverse image (namely \fodw in this case). Then, we need to show that covering and separation are decidable for \fodw.

Fortunately, the second step has already been achieved: it was shown in~\cite{pvzmfcs13} and in~\cite{pzcovering,pzcovering2} that \fodw-separation and \fodw-covering are decidable. Thus, we just have to show that \fodp is the \su-enrichment of \fodw. We state this in the following proposition.

\begin{proposition} \label{prop:fo2enrich}
  Over words, \fodp is the \su-enrichment of \fodw.
\end{proposition}

\begin{remark}
  It is important to point out that while the formulation is new, the underlying ideas behind Proposition~\ref{prop:fo2enrich} were already known. This connection between \fodw and \fodp was originally presented by~Thérien and Wilke~\cite{twfodeux}. However, the full proof of this result is scattered in the literature and relies on different terminology. Thus, it makes sense to detail it here.
\end{remark}

There are two inclusions to prove for showing Proposition~\ref{prop:fo2enrich}. We devote a subsection to each of them.

\subsection{From enrichment to successor} We show here that any language in the \su-enrichment of \fodw may be defined by an \fodp sentence. For this, let us fix an alphabet $A$ and consider a language $L \in \fodw \circ \su$ over $A$. By definition, there exists an \su-partition \Pb of $A^*$ such that,
\[
  L = \bigcup_{P \in \Pb} (P \cap \tau_\Pb\inv(L_P))
\]
where all languages $L_P \subseteq (\Pb \times A)^*$ are definable in \fodw. We show that $L$ can be defined by an \fodp sentence. Since we may freely use Boolean connectives in \fodp sentences, it suffices to show that for all $P\in \Pb$, both $P$ and $\tau_\Pb\inv(L_P)$ are  defined by an \fodp sentence. We start with the following preliminary lemma.

\begin{lemma} \label{lem:prelimfo2}
  For any $u \in A^*$, one may construct an \fodp formula $\varphi_u(x)$ (with one free variable $x$) such that for any $w \in A^*$ and any position $x$ in $u$, we have $w \models \varphi_u(x)$ if and only if $w[0,x-1] \in A^*u$.
\end{lemma}

\begin{proof}
  We use induction on $u$ to define $\varphi_u(x)$. If $u = \varepsilon$, it suffices to define $\varphi_u(x) = \top$. Otherwise, $u =va$ for some $v \in A^*$ and $a \in A$ and we define,
  \[
    \varphi_u(x) = \exists y\ (y+1 = x \wedge a(y) \wedge \varphi_v(y)).
  \]
  This concludes the proof of Lemma~\ref{lem:prelimfo2}.
\end{proof}

We now start the main argument. Let $P\in \Pb$, we first show that $P$ and $\tau_\Pb\inv(L_P)$ may both be defined by an \fodp sentence.

\medskip
\noindent
{\bf Case~1}: each $P\in\Pb$ may be defined by an \fodp sentence. By definition of \Pb, we know that $P \in \su(A)$, whence $P$ is a finite Boolean combination of languages of the form $A^*w$, with $w \in A^*$. Since \fodp is a Boolean algebra, it suffices to show that $A^*w$ can be defined in \fodp. If $w = \varepsilon$, then $A^*$ is defined by the sentence $\top$. Otherwise, $w = ua$ with $u \in A^*$ and $a \in A$, $A^*w$ is defined by $\exists x\ max(x) \wedge a(x) \wedge \varphi_u(x)$.

\medskip
\noindent
{\bf Case~2}: the language $\tau_\Pb\inv(L_P)$ may be defined by an \fodp sentence. Recall that $L_P \subseteq (\Pb \times A)^*$ is defined by some \fodw sentence $\xi$. We use the following fact which is an immediate consequence of Lemma~\ref{lem:prelimfo2} since all languages in \Pb belong to \su.

\begin{fct} \label{fct:computelabel}
  Given any $(P,a) \in \Pb \times A$, there exists a \fodp formula $\zeta_{(P,a)}(x)$ (over $A$) with one free variable such that for any $w \in A^*$ and any position $x$ in $w$, we have, $w \models \zeta_{(P,a)}(x)$ if and only if $x$ has label $(P,a) \in \tau_\Pb(w)$.
\end{fct}
It is now simple to construct an \fodp sentence defining $\tau_\Pb\inv(L_P)$ from the \fodw sentence $\xi$ defining $L_P \subseteq (\Pb \times A)$: we replace atomic subformulas of the form $(P,a)(x)$, for some $(P,a) \in \Pb \times A$, by the formula $\zeta_{(P,a)}(x)$. This concludes the proof for this~direction.

\subsection{\efgame games} To prove the converse direction in Proposition~\ref{prop:fo2enrich}, we need the \efgame games associated to \fodw and \fodp. We first define these games.

\begin{remark}
  For the sake of simplifying the \fodp-game, we shall assume that the predicates $min,max$ and $\varepsilon$ are not allowed in \fodp. This is not restrictive since $min(x)$ is defined by $\neg(\exists y\ y < x)$, $max(x)$ by $\neg(\exists y\ y>x)$ and $\varepsilon$ by $\forall x \bot$.
\end{remark}

The \emph{(quantifier) rank} of a first-order formula $\varphi$, denoted $\rk\varphi$, is defined as the largest number of quantifiers along a branch in the parse tree of~$\varphi$. Formally, $\rk\varphi=0$ if $\varphi$ is an atomic formula, $\rk{\neg\varphi}=\rk\varphi$, $\rk{\varphi_1\lor\varphi_2}=\max(\rk{\varphi_1},\rk{\varphi_2})$ and $\rk{\exists x\,\varphi}=\rk\varphi+1$. For any alphabet $A$, any natural number $k \in \nat$ and any two words $w,w' \in A^*$, we write:
\begin{itemize}
\item $w \kfodeq w'$ when $w$ and $w'$ satisfy the same \fodw sentences of rank $k$.
\item $w \kfodeqp w'$ when $w$ and $w'$ satisfy the same \fodp sentences of rank $k$.
\end{itemize}

It is immediate that both $\kfodeq$ and $\kfodeqp$ are equivalence relations over the set $A^*$. Moreover, one may verify the following standard lemma, which characterizes languages definable in \fodw and \fodp using these relations:

\begin{lemma}[Folklore] \label{lem:fo2equiv}
  Given any alphabet $A$, any language $L \subseteq A^*$ and any natural number $k \in \nat$, the following properties hold:
  \begin{itemize}
  \item $L$ may be defined by a \fodw sentence of rank $k$ if and only if $L$ is a union of $\kfodeq$-classes.
  \item $L$ may be defined by a \fodp sentence of rank $k$ if and only if $L$ is a union of $\kfodeqp$-classes.
  \end{itemize}
\end{lemma}

We now define the \efgame games associated to \fodw and \fodp, which give alternate definitions for the relations $\kfodeq$ and $\kfodeqp$.

\medskip
\noindent
{\bf \fodw game.} The board of the $\fodw$-game consists of two words $w$ and $w'$. It lasts a predefined number $k$ of rounds. There are two players called \emph{Spoiler and Duplicator}. Moreover, there are two pebbles and at any time during the game after the first round, one of them is placed on a position of $w$ and the other on a position $w'$, and these two positions have the same label (when the game starts, no pebble is on the board).

In the first round, Spoiler chooses a word (either $w$ or $w'$) and places a pebble on a position of this word. Duplicator must answer by placing the other pebble on a position of the other word having the same label. The remaining rounds are played as follows. Spoiler chooses a word (either $w$ or $w'$) and moves the pebble inside this word from its original position $x$ to a new position $y$. Duplicator must answer by moving the other pebble in the other word from its original position $x'$ to a new position $y'$ having the same label as $y$ and such that $x' < y'$ if and only if $x < y$.  

Duplicator wins if she manages to play for all $k$ rounds. Spoiler wins as soon as Duplicator is unable to play.

\medskip
\noindent
{\bf \fodp game.} The \fodp-game is defined similarly with an additional constraint for Duplicator when answering Spoiler's moves. When Spoiler makes a move, Duplicator must choose her answer $y'$ so that $x' + 1 = y'$ if and only if $x+1 = y$ and $y' + 1 = x'$ if and only if $y+1 = x$ (in addition to the constraints already presented for the \fodw game).

\medskip

We may now state the \efgame theorem for \fodw and \fodp: \kfodeq and \kfodeqp are characterized by the \fodw game and the \fodp game, respectively.

\begin{theorem}[Folklore] \label{thm:fo2ef}
  Let $A$ be an alphabet. Given any $k \in \nat$ and any words $w,w' \in A^*$, the two following properties hold:
  \begin{itemize}
  \item $w \kfodeq w'$ if and only if Duplicator has a winning strategy for playing $k$-rounds in the \fodw-game over $w$ and $w'$.
  \item $w \kfodeqp w'$ if and only if Duplicator has a winning strategy for playing $k$-rounds in the \fodp-game over $w$ and $w'$.
  \end{itemize}
\end{theorem}

\subsection{From successor to enrichment} We are now ready to show the remaining direction in Proposition~\ref{prop:fo2enrich}: any language that can be defined by an \fodp sentence belongs to the \su-enrichment of \fodw.

We start with a preliminary definition. Recall that given any $k \in \nat$, \keqsu denotes the canonical equivalence defined on $A^*$ associated to $\su_k$. We write $\Pb_k$ for the finite partition of $A^*$ into classes of \keqsu. Note that $\Pb_k$ is a \su-partition of $A^*$ by Lemma~\ref{lem:canoeq}. In the proof, we use the \su-partitions $\Pb_k$ to build languages in $\fodw \circ \su$. For the sake of simplifying the notation, given $k \in \nat$, we shall write $\tau_k$ for the map $\tau_{\Pb_k}: A^* \to (\Pb_k \times A^*)$.

\smallskip
Our argument to prove $\fodp \subseteq \fodw \circ \su$ is based on the following result.

\begin{proposition} \label{prop:fo2themainprop}
  Let $k \in \nat$, and $w,w' \in A^*$. If $\tau_{2k}(w) \kfodeq \tau_{2k}(w')$, then $w \kfodeqp w'$.
\end{proposition}

Before showing Proposition~\ref{prop:fo2themainprop}, we use it to conclude our argument for Proposition~\ref{prop:fo2enrich}. Let $L \subseteq A^*$ be defined by some \fodp sentence $\varphi$. We show that $L \in \fodw \circ \su$. By Lemma~\ref{lem:fo2equiv}, $L$ is a union of \kfodeqp-classes, where $k$ is the rank of $\varphi$. We define a language $H \subseteq (\Pb_{2k} \times A)^*$ as follows:
\[
  H = \{u \in {(\Pb_{2k} \times A)}^* \mid \text{there exists $w \in L$ such that } \tau_{2k}(w) \kfodeq u\}.
\]
By definition, $H$ is a union of \kfodeq-classes and can therefore by defined by some \fodw sentence of rank $k$ (see Lemma~\ref{lem:fo2equiv}). We show that,
\begin{equation} \label{eq:fo2enrichproof}
  L = \tau_{2k}\inv(H) = \bigcup_{P \in \Pb_{2k}} (P \cap \tau_{2k}\inv(H)),
\end{equation}
an expression showing that $L \in \fodw \circ \su$ (since $\Pb_{2k}$ is an \su-partition of $A^*$). To prove~\eqref{eq:fo2enrichproof}, we start with the left to right inclusion. Assume that $v \in L$. It is then immediate from the definition that $\tau_{2k}(v) \in H$ since $\tau_{2k}(v) \kfodeq \tau_{2k}(v)$, hence we get $v \in \tau_{2k}\inv(H)$, which establishes the first inclusion. For the converse inclusion, assume that $v \in \tau_{2k}\inv(H)$. We show that $v \in L$. Since $v \in \tau_{2k}\inv(H)$, we get by definition of $H$ that there exists $w \in L$ such that $\tau_{2k}(w) \kfodeq \tau_{2k}(v)$. Thus, we obtain from Proposition~\ref{prop:fo2themainprop} that $w \kfodeqp v$. Finally, since $w \in L$ and $L$ is a union of \kfodeqp-classes, we obtain that $v \in L$.

\medskip

It remains to prove Proposition~\ref{prop:fo2themainprop}. Let $k \in \nat$ and let $w,w' \in A^*$ be words such that $\tau_{2k}(w) \kfodeq \tau_{2k}(w')$. Our objective is to prove that $w \kfodeqp w'$. By Theorem~\ref{thm:fo2ef}, this amounts to describing a winning strategy for Duplicator in the $k$-round \fodp-game over $w$ and $w'$. We call this game \Gs. Duplicator's strategy involves playing another ``shadow'' \fodw-game over $\tau_{2k}(w)$ and $\tau_{2k}(w')$. Recall that by hypothesis and Theorem~\ref{thm:fo2ef}, she has a winning strategy for $k$ rounds in this shadow game. Depending on the moves that Spoiler makes in \Gs, Duplicator may have to simulate a move by Spoiler in the shadow game. Her strategy in the shadow game gives her an answer to this simulated move, which she is then able to use for computing a suitable answer in \Gs.

Recall that for any word $v \in A^*$ (including $w$ and $w'$), $\tau_{2k}(u)$ is a relabeling of $u$ over $(\Pb \times A)$. In particular, this means that any position of $w$ (resp. $w'$) corresponds to a position of $\tau_{2k}(w)$ (resp. $\tau_{2k}(w')$), and may be viewed as such, and \emph{vice versa}. This will be convenient to relate the moves performed in \Gs to the ones played in the shadow game.

Given two positions $x$ and $z$ of $w$ and a natural number $h \leq k-1$, we say that \emph{$x$ is $h$-safe for $z$} when all positions $y$ in $u$ such that $|x-y| \leq h$ satisfy $z - 2k \leq y \leq z$. Note that these positions in $w$ are all fully described by (the two components of) the label of $z$ in $\tau_{2k}(w)$. We extend the definition to positions $x',z'$ of $w'$ in the same way.

We may now describe Duplicator's strategy in \Gs. It involves enforcing an invariant $\Is(j)$ which has to hold after each round $j$. Let $1 \leq j \leq k$. Assume that $j$ rounds have been played in \Gs so far and let $x,x'$ be the positions of $w,w'$ on which the pebbles are currently placed in \Gs. Furthermore, let $z,z'$ be the positions of $\tau_{2k}(w),\tau_{2k}(w')$ on which the pebbles are currently placed in the shadow game. We say that $\Is(j)$ holds when the following conditions are met.
\begin{enumerate}
\item $x - z = x' - z'$.
\item $x$ is $(k-j)$-safe for $z$ and $x'$ is $(k-j)$-safe for $z'$.
\item Duplicator has a wining strategy for playing $k-j$ more rounds in the shadow game.
\end{enumerate}

We now describe a strategy allowing Duplicator to play and enforce $\Is(j)$ after each round $j$, $1\leq j\leq k$. Let us first explain how Duplicator may enforce $\Is(1)$ after round~1.

Assume that Spoiler puts a pebble on a position $x$ of $w$ (the case when Spoiler puts a pebble in $w'$ is symmetrical). Then, Duplicator simulates a moves by Spoiler in the shadow game by putting a pebble on position $z = min(|w|-1,x+k-1)$ in $\tau_{2k}(w)$. She then obtains an answer $z'$ in $\tau_{2k}(w')$ having the same label as $z$ from her strategy. By definition, $x$ is $(k-1)$-safe for $z$. Since $z$ and $z'$ have the same label in $\tau_{2k}(w)$ and $\tau_{2k}(w')$, by definition of the labels in $\tau_{2k}(w)$ and $\tau_{2k}(w')$, there must exist a position $x'$ in $w'$ such that $z - x = z' - x'$ and with the same label as $x$. This position $x'$ is Duplicator's answer, which is clearly correct. Moreover, $\Is(1)$ is satisfied.

We now assume that $j \geq 1$ rounds have already been played and that $\Is(j)$ holds. Let $x,x'$ be the positions of $w,w'$ on which the pebbles are currently placed in \Gs and let $z,z'$ be those in $\tau_{2k}(w),\tau_{2k}(w')$ on which the pebbles are currently placed in the shadow game. Assume that Spoiler moves the pebble from position $x$ in $w$ to a new position $y$ (as before, the other case is symmetrical). We describe a correct answer for Duplicator which satisfies $\Is(j+1)$. There are two cases depending on whether $y$ is $(k-j-1)$-safe for $z$ or not.

Assume first that $y$ is $(k-j-1)$-safe for $z$. In that case, Duplicator does not use the shadow game (the pebbles remain on $z$ and $z'$ for this round). Since $z$ and $z'$ have the same label in $\tau_{2k}(w)$ and $\tau_{2k}(w')$, there must exist a position $y'$ in $w'$ such that $z - y = z' - y'$ and with the same label as $y$. This position $y'$ is Duplicator's answer. One may verify that it is correct and that $\Is(j+1)$ is satisfied.

We now assume that $y$ is not $(k-j-1)$-safe for $z$. There are two sub-cases depending on whether $y < x$ or $x < y$. By symmetry, we only consider the case $x < y$. Note that by hypothesis, we have the following properties:
\begin{enumerate}[label=$\alph*)$]
\item\label{item:y--x+1} $y > x+1$ (since $x$ is $(k-j)$-safe, $x+1$ is $(k-j-1)$-safe).
\item\label{item:y--z} $y > z - k +j +1$ (as $y$ is strictly right of the rightmost $(k-j-1)$-safe position for~$z$).
\item\label{item:z-neq-w} $z \neq |w|-1$ (otherwise, $y > x$ would be $(k-j-1)$-safe for $z$)
\end{enumerate}

Duplicator first simulates a move by Spoiler in the shadow game: she moves the pebble in $\tau_{2k}(w)$ from $z$ to $z_1 = min(|w|-1,y + k - j - 1)$, \emph{i.e.}, to the leftmost position for which $y$ is $(k-j-1)$-safe. Note that $z_1 > z$. Indeed, either $z_1 = |w|-1$ and $z < z_1$ by Item~\ref{item:z-neq-w} above. Otherwise, $z_1 =y + k - j - 1>z$ by Item~\ref{item:y--z}. Hence, Duplicator's strategy in the shadow game yields an answer $z'_1 > z'$ in $\tau_{2k}(w')$ with the same label as~$z_1$.

Since $z_1,z'_1$ have the same label in $\tau_{2k}(w),\tau_{2k}(w')$, we now obtain a position $y'$ in $w'$ such that $z_1 - y = z'_1 - y' = k - j -1$ and with the same label as $y$. This position $y'$ is Duplicator's answer. Proving that it is correct and that $\Is(j+1)$ now holds amounts to showing that $y' > x'+1$.  Since $x'$ was $(k-j)$-safe for $z'$, we know that $x' \leq z'- k +j$. Since $z' < z'_1$, this yields $x' < z'_1- k +j$. Finally, we have $z'_1 - y' = k - j -1$ by definition. Altogether, this means that $x'+1 < y'$, as desired.

\smallskip
This concludes the proof of Proposition~\ref{prop:fo2themainprop}, and therefore of Proposition~\ref{prop:fo2enrich} as well.

\section{Application to quantifier alternation}
\label{sec:hiera}
In this section, we illustrate Theorem~\ref{thm:wfwords} with a second example: the quantifier alternation of first-order logic over words. We prove the following result.

\begin{corollary}[of Theorem~\ref{thm:wfwords}]\label{cor:alt}
  Over words, covering and separation are decidable for the levels \sipe{1}, \bspe{1}, \sipe{2}, \bspe{2} and \sipe{3} in the alternation hierarchy of first-order logic.
\end{corollary}

Note that Corollary~\ref{cor:alt} subsumes many difficult results from the literature. In particular, it shows that \emph{membership} is decidable for \bspe1 and for \sipe2. Direct proofs for both of these results are difficult~\cite{knast83,gssig2,glasserdd}.

As usual, obtaining Corollary~\ref{cor:alt} from Theorem~\ref{thm:wfwords} requires clearing two preliminary steps.
\begin{itemize}
\item First we prove that the fragments mentioned in the theorem are the \su-enrichment of lattices closed under right quotient and inverse image. As expected, we use \sio{1}, \bso{1}, \sio{2}, \bso{2} and \sio{3} in this case.
\item Then, covering and separation are shown to be decidable for these simpler classes.
\end{itemize}

Again, the second step has already been achieved: it is known that covering and separation are decidable for \sio{1}~\cite{cmmptsep,pzbpol}, \bso{1}~\cite{pvzmfcs13,cmmptsep,pzbpol}, \sio{2}~\cite{pzqalt,pzqaltj,pzbpol}, \bso{2}~\cite{pzbpol} and \sio{3}~\cite{pseps3,pseps3j}. Thus, we concentrate on proving the connections with \su-enrichment. We state them in the following proposition.

\begin{proposition}\label{prop:henrich}
  Given any $n \geq 1$, the following two properties hold over words:
  \begin{itemize}
  \item \sipe{n} is the \su-enrichment of \sio{n}.
  \item \bspe{n} is the \su-enrichment of \bso{n}.
  \end{itemize}
\end{proposition}

\begin{remark}
  As for two-variable first-order logic in the previous section, these properties are essentially already known. The underlying ideas behind the connection with \su-enrichment are due to Straubing~\cite{StrauVD}.
\end{remark}

In the rest of this section, we prove Proposition~\ref{prop:henrich}. We focus on the first item: \sipe{n} is the \su-enrichment of \sio{n}.   The proof for the second item is similar and left to the reader. There are two inclusions to prove, we devote one subsection to each of~them.

\subsection{From enrichment to successor}

Let $n \in \nat$ and consider some alphabet $A$. Let $L \subseteq A^*$ be a language belonging to $\sio{n} \circ \su$. We want to show that $L$ is definable in \sipe{n}. By definition, there exists an \su-partition \Pb of $A^*$ such that
\[
  L = \bigcup_{P \in \Pb} (P \cap \tau_\Pb\inv(L_P)),
\]
where all languages $L_P \subseteq (\Pb \times A)^*$ are definable in \sio{n}. We show that $L$ can be defined by a \sipe{n} sentence. Since we may freely use disjunction and conjunction in  \sipe{n} sentences, it suffices to show that for all $P\in \Pb$, both $P$ and $\tau_\Pb\inv(L_P)$ are defined by a \sipe{n} sentence.

\medskip
\noindent
{\bf Case~1.} We start with the language $P \in \Pb$. By definition $P \in \Pb(A)$. It follows from Lemma~\ref{lem:eqclasses} that~$P$ is a finite union of languages $\{w\}$ or $A^*w$ where $w \in A^*$. Since $\sipe{n}$ is closed under union, it suffices to show that these two kinds of languages may be defined in \sipe{n}, which is easy: if $w = \varepsilon$, then $\{\varepsilon\}$ and $A^*$ are defined by the sentences ``$\varepsilon$'' and ``$\top$'' respectively. Otherwise,  $w = a_1\cdots a_\ell$ for  $a_1,\dots,a_\ell \in A$. In that case, $\{w\}$ is defined by the following \sipe{1} sentence:
\[
  \exists x_1 \cdots \exists x_\ell\quad  min(x_1) \wedge max(x_\ell) \wedge\left(\bigwedge_{1 \leq i \leq \ell-1} x_{i}+1 = x_{i+1}\right) \wedge \left(\bigwedge_{1 \leq i \leq \ell} a_i(x_i)\right).
\]
Similarly, $A^*w$ is defined  by the following \sipe{1} sentence:
\[
  \exists x_1 \cdots \exists x_\ell\quad  max(x_\ell) \wedge \left(\bigwedge_{1 \leq i \leq \ell-1} x_{i}+1 = x_{i+1}\right) \wedge \left(\bigwedge_{1 \leq i \leq \ell} a_i(x_i)\right).
\]

\medskip
\noindent
{\bf Case~2.} We now consider languages of the form $\tau_\Pb\inv(L_P)$. By hypothesis $L_P \subseteq (\Pb \times A)^*$ is defined by some \sio{n} sentence $\Psi$. We exhibit a \sipe{n} sentence defining $\tau_\Pb\inv(L_P)$. For this, we first make sure that all atomic formulas of the form $(Q,a)(x)$ occurring in $\Psi$ are under no negation. This can be assumed since if $(Q,a)(x)$ is such an atomic formula, we have $(Q,a)(x)=\bigvee_{(Q',a')\not=(Q,a)}(Q',a')(x)$.

There are now two sub-cases, depending on whether $n$ is odd or even. If $n$ is odd, then the innermost block of quantifiers is an existential one. Therefore, replacing an atomic sub-formula $(Q,a)(x)$ which is not negated within a \sipe{n} sentence (and in particular within a \sio{n} sentence such as $\Psi$) by some \sipe{1} formula yields a \sipe{n} sentence again. We now use the following simple result.

\begin{fct}\label{fct:existcase}
  Given any $(Q,a) \in \Pb \times A$, there exists a \sipe{1} formula $\zeta_{(Q,a)}(x)$ over $A$ with one free variable such that for any $w \in A^*$ and any position $x$ in $w$, we have $w \models \zeta_{(Q,a)}(x)$ if and only if $x$ has label $(Q,a)$ in $\tau_\Pb(w)$.
\end{fct}

The proof of Fact~\ref{fct:existcase} is left to the reader (it is similar to that of Case~1 above). Recall that we have a \sio{n} sentence $\Psi$ defining $L_P$. Consider the \sipe{n} sentence $\varphi$ obtained from $\Psi$ by replacing any atomic formula of the form $(Q,a)(x)$ (for $(Q,a) \in \Pb \times A$) by the formula $\zeta_{(Q,a)}(x)$ given by Fact~\ref{fct:existcase}. Then, $\varphi$ is \sipe{n} defining $\tau_\Pb\inv(L_P)$. This concludes the proof for this sub-case.

\medskip

We now assume that $n$ is even. In that case, replacing an atomic sub-formula under no negation within a \sipe{n} sentence (and in particular within a \sio{n} sentence such as $\Psi$) by some \pipe{1} formula yields a \sipe{n} sentence. We shall need the following simple result.

\begin{fct}\label{fct:forallcase}
  Given any $(Q,a) \in \Pb \times A$, there exists a \pipe{1} formula $\xi_{(Q,a)}(x)$ over $A$ with one free variable such that for any $w \in A^*$ and any position $x$ in $w$, we have $w \models \xi_{(Q,a)}(x)$ if and only if $x$ has label $(Q,a)$ in $\tau_\Pb(w)$.
\end{fct}

\begin{proof}
  By Lemma~\ref{lem:eqclasses}, $Q \in \Pb$ is a finite union of languages $\{u\}$ or $A^*u$ where $u \in A^*$. Since \pipe{1} is closed under union, it suffices to consider the cases when $Q$ is one of these two kinds of language. If $u = \varepsilon$, then we let $\xi_{(\{\varepsilon\},a)}(x) = min(x) \wedge a(x)$ and $\xi_{(A^*,a)}(x) = a(x)$.

  Otherwise, there exist $a_1,\dots,a_\ell \in A$ such that $w = a_1\cdots a_\ell$. Observe that for any $m \in \nat$ we have a \pipe{1} sentence $\chi_m(x)$ which holds when $x \geq m+1$. Indeed, we may define this formula by induction on $m$. When $m = 0$, then $\chi_0(x) = \top$. Otherwise, $\chi_m(x) = \neg min(x) \wedge \forall y \ (y+1 = x \Rightarrow \chi_{m-1}(y))$. We may now define, $\xi_{(Q,a)}(x)$. If $Q = \{w\}$, we define $\xi_{(Q,a)}(x)$ as the following formula:
  \[
    a(x) \wedge \chi_\ell(x) \wedge \forall x_1 \cdots \forall x_\ell\ \left(\bigwedge_{i \leq \ell-1} x_{i}+1 = x_{i+1} \wedge x_\ell+1=x \right) \Rightarrow \left(\bigwedge_{i \leq \ell} a_i(x_i) \wedge min(x_1) \right).
  \]
  Finally, if $Q = A^*w$, we define $\xi_{(Q,a)}(x)$ as the following formula,
  \[
    a(x) \wedge \chi_\ell(x) \wedge \forall x_1 \cdots \forall x_\ell\ \left(\bigwedge_{i \leq \ell-1} x_{i}+1 = x_{i+1} \wedge x_\ell+1=x \right) \Rightarrow \left(\bigwedge_{i \leq \ell} a_i(x_i)\right).
  \]
  This concludes the proof of Fact~\ref{fct:forallcase}.
\end{proof}

Recall now that we have a \sio{n} sentence $\Psi$ defining $L_P$. Consider the \sipe{n} sentence $\varphi$ obtained from $\Psi$ by replacing any atomic formula of the form $(Q,a)(x)$ with $(Q,a) \in \Pb \times A$ (\emph{i.e.}, any label test)  by the formula $\xi_{(Q,a)}(x)$ given by Fact~\ref{fct:existcase}. One may verify that $\varphi$ is \sipe{n} and defines $\tau_\Pb\inv(L_P)$ which concludes the proof for this sub-case.

\begin{remark}\label{rem:forall-from-exists}
  When $n$ is even, an alternative proof is to first ensure that all atomic sub-formulas of the form $(Q,a)(x)$ (\emph{i.e.}, label tests) are under exactly one negation in $\Psi$, and to apply Fact~\ref{fct:existcase} again.
\end{remark}

\subsection{\efgame games} Before turning to the converse direction in Proposition~\ref{prop:henrich}, let us recall the definition of the \efgame games associated to the levels \sic{n}. It is parameterized by an arbitrary signature~$\sigma$ (which we shall instantiate later with the signatures of \sio{n} and \sipe{n}).

\medskip
\noindent
{\bf Quantifier rank and canonical preorders.} As for two variable first-order logic, the link with \efgame games is based on the notion of quantifier rank. Recall that the rank of a first-order sentence is the longest sequence of nested quantifiers in $\varphi$.

Using the quantifier rank, we associate a preorder relation to any level $\sic{n}(\sigma)$ in the quantifier alternation hierarchy. Given two words $w,w' \in A^*$ and $k \in \nat$, we write $w \ksieq^\sigma w'$ when
\[
  \text{For any $\sic{n}(\sigma)$ sentence of rank at most $k$:} \quad w \models \varphi \Rightarrow w' \models \varphi.
\]
The next lemma is folklore and simple to verify. It characterizes with the preorder $\ksieq^\sigma$ the languages that can be defined by a $\sic{n}(\sigma)$ sentence of rank $k$. An \emph{upper set for $\ksieq^\sigma$} is a language $L \subseteq A^*$ which is upward closed under $\ksieq^\sigma$: given any $w,w' \in A^*$, if $w \in L$ and $w \ksieq^\sigma w'$, then $w' \in L$.

\begin{lemma}[Folklore]\label{lem:preocarac}
  Consider two natural numbers $n \geq 1$ and $k\geq0$. For any language $L \subseteq A^*$, the following two properties are equivalent:
  \begin{itemize}
  \item $L$ can be defined by a $\sic{n}(\sigma)$ sentence of rank $k$.
  \item $L$ is an upper set for $\ksieq^\sigma$.
  \end{itemize}
\end{lemma}

We now define the \efgame game for $\sic{n}(\sigma)$ (called the $\sic{n}(\sigma)$ game). It yields an alternate (and easier to manipulate) definition of the preorders $\ksieq^\sigma$. The board of the game consists of two words $w$ and $w'$ in $A^*$ and there are two players called \emph{Spoiler} and \emph{Duplicator}. We speak of the $\sic{n}(\sigma)$-game over the pair $(w,w')$, or over $w$ and $w'$. Note that unlike in the \fod game, the ordering between the two words is relevant: $w$ is the first word and $w'$ is the second. Spoiler's goal is to prove that the words $w$ and $w'$ are different (wrt.\ \sio{n} or \sipe{n}) while Duplicator must prevent him from doing so. The game is set to last a predefined number $k$ of rounds and when it starts, each player owns $k$ pebbles. Moreover, we have the two following additional parameters that may change as the play progresses:
\begin{enumerate}
\item There is a distinguished word among $w,w'$, called the \emph{active word}. Initially, the active word is the first word, that is, $w$.
\item There is a counter $c$ called the \emph{alternation counter}. Initially, $c$ is set to $0$. It can only increase, and its maximal allowed value is $n-1$. It counts the number of times the active word was changed.
\end{enumerate}

A single round is played as follows. Spoiler has to place a pebble on the board (\emph{i.e.}, on a position of either $w$ or $w'$). However, there are constraints on the word that he may choose. Spoiler can always choose the active word, in which case both $c$ and the active word remain unchanged. On the other hand, Spoiler may choose the word that is not active only when $c < n - 1$. In that case, the active word is switched and $c$ is incremented by $1$.

Duplicator must answer by placing one of her own pebbles on some position of the other word. This answer must yield a \emph{correct configuration}. By \emph{configuration after round $\ell$}, we mean the set
\[
  C = \{(x_1,x'_1),\dots,(x_\ell,x'_\ell)\},
\]
where the elements $(x_i,x'_i)$ are the pairs of positions ($x_i$  in $w$ and $x'_i$ in $w'$) holding corresponding pebbles at rounds $1,\ldots,\ell$ (\emph{i.e.}, Spoiler placed a pebble on $x_i$ in a previous round and Duplicator answered by putting a pebble on $x'_i$, or \emph{vice versa}). Such a configuration is declared \emph{correct} if and only if for any predicate $P \in \sigma$ of arity $m$, given any $i_1,\dots,i_m \leq \ell$,
\[
  \text{$P(x_{i_1},\dots,x_{i_m})$ holds} \quad \text{if and only if} \quad \text{$P(x'_{i_1},\dots,x'_{i_m})$ holds}.
\]
Intuitively, a configuration is correct when it is impossible to point out a difference between the sequences of positions $x_1,\dots,x_\ell$ in $w$ and $x'_1,\dots,x'_\ell$ in $w'$ by using the predicates available in $\sigma$.

\medskip

When the game starts, the configuration is empty. \emph{Duplicator wins} if this initial configuration is correct (while empty, the initial configuration may not be correct when $\sigma$ contains constants such as ``$\varepsilon$'') and if she is able to answer all moves by Spoiler with a correct configuration until all $k$ rounds have been played. On the other hand \emph{Spoiler wins} if the initial configuration is not correct or as soon as Duplicator is unable to play. We now state the \efgame theorem  for the $\sic{n}(\sigma)$-game. It characterizes the preorder $\ksieq^\sigma$: two words are comparable, \emph{i.e}, $w\ksieq^\sigma w'$, when Duplicator has a winning strategy for $k$ rounds over $(w,w')$.

\begin{theorem}[Folklore]\label{thm:efgame}
  Let $n \geq 1$, $k \in \nat$ and $w,w' \in A^*$. Then $w \ksieq^\sigma w'$ if and only if Duplicator has a winning strategy for playing $k$ rounds in the $\sic{n}(\sigma)$ game over $(w,w')$.
\end{theorem}

This concludes the definition of \efgame games. Note that while the above presentation is generic to all signatures, we are only interested in two specific ones. Given $n \geq 1$ and $k\in \nat$, we write  \ksieq for the preorder associated to the \sio{n} sentences of rank $k$ and \ksieqp for the preorder associated to the \sipe{n} sentence of rank $k$. Finally, we shall need the following simple result about the relations  \ksieq.

\begin{lemma}\label{lem:remsuff}
  Let $n \geq 1$ and let $h,k \in \nat$ be natural numbers. Consider three \fwords $w,w',u \in A^*$ such that $|u| \leq h$ and $wu \sieq{k+h} w'u$. Then, we have $w \sieq{k} w'$.
\end{lemma}

\begin{proof}
  By hypothesis, we know than Duplicator has a winning strategy for playing $k+h$ rounds in the \sio{n}-game over $wu$ and $w'u$. Since $|u| \leq h$, it is simple to verify that as long as there are more than $h$ rounds remaining after the current one, if Spoiler places a pebble in one of the prefixes $w$ or $w'$, then Duplicator's strategy gives an answer in $w$ or $w'$. Otherwise, Spoiler would be able to win within the $h$ following rounds (it is important here that the signature of \sio{n} includes the linear order ``$<$''). Therefore, it is immediate that Duplicator gets a winning strategy for playing $k$ rounds in the \sio{n}-game over $w$ and $w'$. This means that $w \sieq{k} w'$, as desired.
\end{proof}

\subsection{From successor to enrichment} We are now ready to prove the remaining direction in Proposition~\ref{prop:henrich}. For any $n \geq 1$, we show that $\sipe{n} \subseteq \sio{n} \circ \su$.

Let us start with some preliminary definitions. Recall that for any $k \in \nat$, we denote by \keqsu the canonical equivalence on $A^*$ associated to $\su_k$.  We write $\Pb_k$ for the finite partition of $A^*$ into \keqsu-classes. Recall that $\Pb_k$ is an \su-partition of $A^*$ by Lemma~\ref{lem:canoeq}. We shall only use the \su-partitions $\Pb_k$ to build languages in $\sio{n} \circ \su$. For the sake of simplifying the notation, given $k \in \nat$, we write
\begin{itemize}
\item $\tau_k$ for the map $\tau_{\Pb_k}: A^* \to (\Pb_k \times A^*)$.
\item $\delta_k$ for the map $\delta_{\Pb_k}: A^* \times A^* \to (\Pb_k \times A^*)$.
\end{itemize}

\smallskip
We now prove that $\sipe{n} \subseteq \sio{n} \circ \su$. Our argument is based on the following proposition.

\begin{proposition}\label{prop:themainprop}
  Let $k\geq0$, $n \geq 1$ and $\ell = 2^k$ be three integers. Assume that we have $w,w' \in A^*$ such that $\tau_{\ell}(w) \sieq{k+\ell} \tau_{\ell}(w')$ and $w \eqsu{\ell} w'$. Then, we have $w \sieqp{k} w'$.
\end{proposition}

Before we show Proposition~\ref{prop:themainprop}, let us use it to conclude this direction of the proof. Let $L \subseteq A^*$ be a language defined by some \sipe{n} sentence $\varphi$. We show that $L \in \sio{n} \circ \su$. By Lemma~\ref{lem:preocarac}, $L$ is an upper set for the preorder \sieqp{k}, where $k$ is the rank of $\varphi$.

Let $\ell = 2^k$. For any $P \in \Pb_{\ell}$, we let $H_P \subseteq (\Pb_{\ell} \times A)^*$ be the following upper set for $\sieq{k+\ell}$:
\[
  H_P = \{u \in (\Pb_{\ell} \times A)^* \mid \text{there exists $w \in P \cap L$ such that } \tau_{\ell}(w) \sieq{k+\ell} u\}.
\]
Since $H_P$ is an upper set for \sieq{k+\ell}, Lemma~\ref{lem:preocarac} entails that it can be defined by a \sio{n} sentence (of rank $k+\ell$). To conclude, we will show that
\begin{equation}\label{eq:enrichproof}
  L = \bigcup_{P \in \Pb_{\ell}} (P \cap \tau_{\ell}\inv(H_P)).
\end{equation}
It will then be immediate that $L \in \sio{n} \circ \su$, since $\Pb_{\ell}$ is an \su-partition of $A^*$.

It remains to prove~\eqref{eq:enrichproof}. We start with the left to right inclusion. Assume that $v \in L$. Since $\Pb_{\ell}$ is a partition of $A^*$, there exists some unique $P \in \Pb_{\ell}$ such that $v \in P$. It is then immediate from the definition that $\tau_{\ell}(v) \in H_P$ since $\tau_{\ell}(v) \sieq{k+\ell} \tau_{\ell}(v)$. Thus, we get $v \in P \cap \tau_{\ell}\inv(H_P)$ which concludes the proof of this inclusion.

We turn to the right to left inclusion. Assume that $v \in P \cap \tau_{\ell}\inv(H_P)$ for some $P \in \Pb_{\ell}$. We want to show that $v \in L$. Since $v \in \tau_{\ell}\inv(H_P)$, we obtain by definition of $H_P$ some word $w \in P \cap L$ such that $\tau_{\ell}(w)\sieq{k+\ell}\tau_{\ell}(v)$. Moreover, since $v$ and $w$ both belong to $P$, we have $w \eqsu{\ell} v$. Thus, since $\ell = 2^k$ by definition, we obtain the relation $w \sieqp{k} v$ from Proposition~\ref{prop:themainprop}. Finally, since $w \in L$ and since $L$ is an upper set for \sieqp{k}, we get $v \in L$, as desired.

\medskip

It remains to prove Proposition~\ref{prop:themainprop}, to which we devote the end of the section. Let $k\geq0$, $n \geq 1$ and $\ell = 2^k$. Consider two words $w,w' \in A^*$ such that $\tau_{\ell}(w) \sieq{k+\ell} \tau_{\ell}(w')$ and $w \eqsu{\ell} w'$. We prove that $w \sieqp{k} w'$. As expected, we use an \efgame argument and describe a winning strategy for Duplicator in the \sipe{n}-game over $w$ and $w'$. Recall that there are $k$ rounds to play, that the alternation counter $c$ starts at $0$ and has to remains bounded by $n-1$. We use an induction on $n$ and $k$ (in any order) to describe Duplicator's winning strategy.

\smallskip

Assume first that $k = 0$, which means that $\ell = 1$. In that case, there are no rounds to play and it suffices to show that Duplicator wins automatically (\emph{i.e.}, that $w$ and $w'$ satisfy the same constants in the signature of \sipe{n}). There is only one constant in the signature of \sipe{n}: ``$\varepsilon$''. It is immediate that $w \models \varepsilon$ if and only if $w' \models \varepsilon$ since we know that $w \eqsu{1} w'$ by hypothesis. This concludes the case $k = 0$.

\smallskip

We now assume that $k \geq 1$. We need to describe a strategy for Duplicator in order to play $k$ rounds in the \sipe{n}-game over $w$ and $w'$. Consider a move by Spoiler in the first round. We show that Duplicator is able to answer this move and then to win the remaining $k-1$ rounds. The argument depends on whether Spoiler plays his first move in $w$ or in $w'$. If Spoiler plays in $w'$, we use induction on $n$. In that case, the alternation counter is incremented (in particular, this may only happen when $n \geq 2$). One may verify from the definition that the game now corresponds to a \sipe{n-1}-game over $(w',w)$. Hence, it suffices to show that Duplicator has a winning strategy for playing $k$ rounds in this simpler game. This is immediate from induction on $n$. Indeed, we know that $\tau_{\ell}(w) \sieq{k+\ell} \tau_{\ell}(w')$ and $w \eqsu{\ell} w'$ by hypothesis. One may verify that this implies $\tau_{\ell}(w') \simieq{k+\ell} \tau_{\ell}(w)$ and $w' \eqsu{\ell} w$. Hence, we obtain from induction on $n$ that $w' \ksimieqp w$ which yields the desired strategy for Duplicator.

\smallskip

It remains to handle the case when Spoiler plays his first move on some position $x$ of the word~$w$. This requires more work. We may decompose $w$ according to the position $x$: $w =uav$ where the highlighted letter $a$ is at position $x$. We use the following lemma to describe an answer for Duplicator.

\begin{lemma}\label{lem:innerefgame}
  The word $w'$ has a decomposition $w' = u'av'$ such that $u \sieqp{k-1} u'$ and $\mbox{v \sieqp{k-1} v'}$.
\end{lemma}

Lemma~\ref{lem:innerefgame} provides Duplicator's answer to Spoiler's first move: consider the decomposition $w'=u'av'$ given by the lemma and let $x'$ be the position of $w'$ corresponding to the highlighted letter $a$ in this decomposition. We choose $x'$ as Duplicator's answer. One may verify that this answer is correct (\emph{i.e.}, $x$ and $x'$ satisfy the same predicates in the signature of \sipe{n}).

\begin{remark}
  For showing that $min(x)$ holds if and only if $min(x')$, one needs to use the fact that $u \sieqp{k-1} u'$ (which means that $u = \varepsilon$ if and only if $u' = \varepsilon$). Symmetrically, the fact that $max(x)$ holds if and only if $max(x')$ is based on $v \sieqp{k-1} v'$.
\end{remark}

It remains to show that Duplicator has a winning strategy for playing $k-1$ more rounds in the \sipe{n}-game over $w = uav$ and $w' = u'av'$ from the configuration $C = \{(x,x')\}$. This can be verified using our hypothesis that $u \sieqp{k-1} u'$ and $v \sieqp{k-1} v'$. Indeed, by induction, we get strategies for playing $k-1$ rounds over $u$ and $u'$, and $v$ and $v'$ respectively. These strategies are easily combined into a single one for playing $k-1$ rounds over $w = uav$ and $w' = u'av'$.

\smallskip

We finish with the proof of Lemma~\ref{lem:innerefgame}. Recall that we have $w = uav$ where the highlighted letter ``$a$'' is at position $x$. We consider two cases depending on the length of $v$ (\emph{i.e.}, on whether ``$x$'' is close to the ``right border'' of $w$). We let $h = 2^{k-1}$. Note that by definition, we have $\ell = 2h$.

\smallskip
\noindent
{\bf First case.} Assume first $|v| < h$. In that case, $av$ is a suffix of length at most $h$ of $w$. Since we know that $w \eqsu{\ell} w'$ and $\ell = 2h$ by hypothesis, it follows that $av$ is a suffix of $w'$ as well. In other words, we obtain that $w'$ admits a decomposition $w' = u'av$ for some $u' \in A^*$. It is immediate that $v \sieqp{k-1} v$. It remains to show that $u \sieqp{k-1} u'$.

We prove that $\tau_{h}(u) \sieq{k-1+h} \tau_{h}(u')$ and $u \eqsu{h} u'$. It will then be immediate by induction on $k$ in Proposition~\ref{prop:themainprop} that we have $u \sieqp{k-1} u'$, as desired. We start with the equivalence $u \eqsu{h} u'$. Recall that by hypothesis we have,
\[
  uav = w \eqsu{\ell} w' = u'av.
\]
Thus, $uav$ and $u'av$ have the same suffixes of length at most $\ell = 2h$.  Moreover, since $|av| \leq h$, it is immediate that $u$ and $u'$ have the same suffixes of length at most $h$, which means that $u \eqsu{h} u'$.

It remains to show that $\tau_{h}(u) \sieq{k-1+h} \tau_{h}(u')$. By hypothesis, we have, $\tau_{\ell}(w) \sieq{k+\ell} \tau_{\ell}(w')$. Since $\ell \geq h$, this implies $\tau_{h}(w) \sieq{k+\ell} \tau_{h}(w')$. Moreover, we have $w = uav$ and $w' = u'av$. Therefore, by Lemma~\ref{lem:deltadef}:
\[
  \tau_{h}(u) \cdot \delta_{h}(u,av)  \sieq{k+\ell} \tau_{h}(u') \cdot \delta_{h}(u',av)
\]
We just proved that $u \eqsu{h} u'$, which, together with the definition of $\delta_h$, implies $\delta_{h}(u,av) = \delta_{h}(u',av)$. Altogether, this means that there exists $z \in (\Pb_{h} \times A)^*$ such that $|z| \leq h$ and,
\[
  \tau_{h}(u) \cdot z  \sieq{k+\ell} \tau_{h}(u') \cdot z
\]
It now follows from Lemma~\ref{lem:remsuff} that $\tau_{h}(u) \sieq{k+\ell-h} \tau_{h}(u')$, and since $\ell=2h$, we get in particular $\tau_{h}(u) \sieq{k-1+h} \tau_{h}(u')$, as desired.

\medskip
\noindent
{\bf Second case.} We now assume that $|v| \geq h$. Since the highlighted ``$a$'' in $w = uav$ is at position $x$, it follows by hypothesis on $v$ that $y = x + h$ is also a position of $w$. Hence, we may further decompose $w$ according to $y$:  $w = uav_1bv_2$, where the highlighted $b$ is at position $y$. In other words $v = v_1b v_2$. Observe that by definition $v_1 = w[x+1,y-1]$ which means that $|av_1| = h$. We use the following fact, whose proof relies on our hypothesis that $\tau_{\ell}(w) \sieq{k+\ell} \tau_{\ell}(w')$.

\begin{fct}\label{fct:useshadow}
  There exists $u',v'_2 \in A^*$ such that $w' = u'av_1bv'_2$ and the following properties are satisfied:
  \begin{itemize}
  \item $uav_1 \eqsu{\ell} u'av_1$.
  \item $\tau_{\ell}(uav_1) \sieq{k+\ell-1} \tau_{\ell}(u'av_1)$.
  \item $\delta_{\ell}(uav_1b,v_2) \sieq{k+\ell-1} \delta_{\ell}(u'av_1b,v'_2)$.
  \end{itemize}
\end{fct}
\begin{proof}
  By definition, $w$ and $\tau_{\ell}(w)$ share the same set of positions. Thus, we may view $x$ and $y$ as positions in $\tau_{\ell}(w)$. In particular, we get from Lemma~\ref{lem:deltadef} that:
  \[
    \tau_{\ell}(w) = \tau_{\ell}(uav_1) \cdot (\wpart{uav_1}{\Pb_{\ell}},b) \cdot \delta_{\ell}(uav_1b,v_2).
  \]
  Since $\tau_{\ell}(w) \sieq{k+\ell} \tau_{\ell}(w')$, Duplicator has a winning strategy for playing $(k+\ell)$ rounds in the \sio{n}-game over $\tau_{\ell}(w)$ and $\tau_{\ell}(w')$. She may simulate a move by Spoiler in this game by placing a pebble on the position $y$ in $\tau_{\ell}(w)$. Her strategy then yields an answer $y'$ in $\tau_{\ell}(w)$. Recall that we may view $y'$ as a position of $w'$. We decompose $w'$ as $w' =  z'cv'_2$ where the highlighted letter $c \in A$ is at position $y'$. It then follows from Lemma~\ref{lem:deltadef} that,
  \[
    \tau_{\ell}(w') = \tau_{\ell}(z') \cdot (\wpart{z'}{\Pb_{\ell}},c) \cdot \delta_{\ell}(z'c,v'_2).
  \]
  By definition of $y'$, we know that $y$ and $y'$ have the same label in $\tau_{\ell}(w)$ and $\tau_{\ell}(w')$. Thus, it is immediate that $b = c$ and $\wpart{uav_1}{\Pb_{\ell}} = \wpart{z'}{\Pb_{\ell}}$. Note that by definition, the latter property means that $uav_1 \eqsu{\ell} z'$. In particular, since $|av_1| = h \leq \ell$, it follows that we have $z'= u'av_1$ for some $u' \in A^*$. Altogether, we have found a decomposition $w' = u'av_1bv'_2$ with $uav_1 \eqsu{\ell} u'av_1$.

  Moreover, we know that Duplicator has a strategy for playing $k+\ell-1$ more rounds in the \sio{n}-game over $\tau_{\ell}(w)$ and $\tau_{\ell}(w')$ from the configuration $\{(y,y')\}$. It follows that $\tau_{\ell}(uav_1) \sieq{k+\ell-1} \tau_{\ell}(u'av_1)$ and $\delta_{\ell}(uav_1b,v_2) \sieq{k+\ell-1} \delta_{\ell}(u'av_1b,v'_2)$ which concludes the proof.
\end{proof}

We may now come back to the proof of Case~2 and describe our decomposition of $w'$. We let $w' = u'av_1bv'_2$ be the decomposition given by Fact~\ref{fct:useshadow}. Finally, we define $v' = v_1bv'_2$. We now have our decomposition $w' = u'av'$. It remains to show that $u \sieqp{k-1} u'$ and $v \sieqp{k-1} v'$.

\smallskip

Let us start with $u \sieqp{k-1} u'$. We prove that $\tau_{h}(u) \sieq{k-1+h} \tau_{h}(u')$ and $u \eqsu{h} u'$. It will then be immediate from induction on $k$ in Proposition~\ref{prop:themainprop} that $u \sieqp{k-1} u'$ as desired. For the equivalence $u \eqsu{h} u'$, we know from the first item in Fact~\ref{fct:useshadow} that,
\[
  uav_1 \eqsu{\ell} u'av_1
\]
Thus, since $|av_1| = h$ and $\ell \geq 2h$, it is immediate that $u \eqsu{h} u'$. We turn to $\tau_{h}(u) \sieq{k-1+h} \tau_{h}(u')$. By the second item in Fact~\ref{fct:useshadow} we have:
\[
  \tau_{\ell}(uav_1) \sieq{k+\ell-1} \tau_{\ell}(u'av_1).
\]
Since $\ell \geq h$, one may verify that this implies $\tau_{h}(uav_1) \sieq{k+\ell-1} \tau_{h}(uav'_1)$.  Using Lemma~\ref{lem:deltadef}, we obtain
\[
  \tau_{h}(u) \cdot \delta_{h}(u,av_1)  \sieq{k+\ell-1} \tau_{h}(u') \cdot \delta_{h}(u',av_1).
\]
Moreover, since $u \eqsu{h} u'$, the definition of $\delta_h$ entails that $\delta_{h}(u,av_1) = \delta_{h}(u',av_1)$. Let $z = \delta_{h}(u,av_1)$. By definition, we have $|z| \leq h$ and,
\[
  \tau_{h}(u) \cdot z  \sieq{k-1+\ell} \tau_{h}(u') \cdot z.
\]
Since $\ell \geq 2h$, we have $k-1+\ell \geq k-1+h+h$ and it now follows from Lemma~\ref{lem:remsuff} that $\tau_{h}(u) \sieq{k-1+h} \tau_{h}(u')$, as desired.

\smallskip

We finish with the inequality $v \sieqp{k-1} v'$. We reuse the same approach, by showing that $\tau_{h}(v) \sieq{k-1+h} \tau_{h}(v')$ and $v \eqsu{h} v'$. The result will then follow by the induction on $k$ in the proof of Proposition~\ref{prop:themainprop}. Recall that $v = v_1 bv_2$ and $v' = v_1bv'_2$.

For proving that $v \eqsu{h} v'$, recall that by hypothesis we have $w \eqsu{\ell} w'$ with $\ell = 2h$. Thus, we have $w \eqsu{h} w'$. Moreover, $|v_1b| = h$ by hypothesis. Thus, $v = v_1bv_2$ and $v' = v_1bv'_2$ are suffixes of $w$ and $w'$ of length larger than $h$. Altogether, it follows that $v \eqsu{h} v'$.

It remains to show that, $\tau_{h}(v) \sieq{k-1+h} \tau_{h}(v')$. Since $v = v_1 bv_2$ and $v' = v_1bv'_2$, we get from Lemma~\ref{lem:deltadef} that:
\[
  \tau_{h}(v_1bv_2) = \tau_{h}(v_1) \cdot \delta_{h}(v_1b,v_2)\quad \text{and} \quad \tau_{h}(v_1bv'_2) =\tau_{h}(v_1) \cdot \delta_{h}(v_1b,v'_2).
\]
It is straightforward to verify that $\sieq{k-1+h}$ is compatible with concatenation. Since clearly, $\tau_{h}(v_1) \sieq{k-1+h} \tau_{h}(v_1)$, it suffices to show that we have $\delta_{h}(v_1b,v_2) \sieq{k-1+h} \delta_{h}(v_1b,v'_2)$. By the third item in Fact~\ref{fct:useshadow}, we have
\[
  \delta_{\ell}(uav_1b,v_2) \sieq{k-1+\ell} \delta_{\ell}(u'av_1b,v'_2).
\]
Since $\ell \geq h$, one may verify that this implies $\delta_{h}(uav_1b,v_2) \sieq{k+\ell-1} \delta_{h}(u'av_1b,v'_2)$. Moreover, since $|v_1b| = h$, the definition of $\delta_h$ gives us that $\delta_{h}(uav_1b,v_2) = \delta_{h}(v_1b,v_2)$ and $\delta_{h}(u'av_1b,v'_2)= \delta_{h}(v_1b,v'_2)$. Therefore, we obtain:
\[
  \delta_{h}(v_1b,v_2) \sieq{k-1+\ell} \delta_{h}(v_1b,v'_2).
\]
In particular, this implies $\delta_{h}(v_1b,v_2) \sieq{k-1+h} \delta_{h}(v_1b,v'_2)$,
which concludes the proof.

\section{\texorpdfstring{The reduction for \iwords}{The reduction
    for infinite words}}

\label{sec:theiw}
In this section, we generalize our reduction to the setting of \ilangs. We follow the same outline as the one we used for languages of finite words in Section~\ref{sec:wfwords}. First, we adapt the definition of enrichment to classes of \ilangs. In this setting, enrichment combines objects of different nature: given a class of \ilangs \Cs and a   class of languages \Ds (such as \su), we define the \Ds-enrichment of \Cs (still denoted by $\Cs \circ \Ds$). Then, we generalize the reduction theorem (Theorem~\ref{thm:wfwords}): given any lattice of \ilangs \Cs closed under inverse image, $(\Cs \circ \su)$-covering reduces to \Cs-covering.

\begin{remark}
  Our statements and proofs in this section are very similar to the ones we presented for words in Section~\ref{sec:wfwords}. In fact, aside from one specific technical result, our main theorem for \iwords and its proof are both straightforward generalizations of Theorem~\ref{thm:wfwords}. For this reason, we shall often leave the proofs of technical sub-results to the reader and refer to the corresponding statement in Section~\ref{sec:wfwords}.
\end{remark}

\subsection{\texorpdfstring{Enrichment for classes of \ilangs}{Enrichment for classes of omega languages}} We generalize enrichment to classes of \ilangs. Let us first adapt \Pb-taggings.

\medskip
\noindent
{\bf \Pb-taggings.} Let $A$ be an alphabet and \Pb a finite partition of $A^*$. We define a canonical map $\tau_\Pb: A^\omega \to (\Pb \times A)^\omega$.  Let $w \in A^\omega$ be an \iword: $w = a_1a_2a_3 \cdots$ with $a_i \in A$ for all $i$. We let $\tau_\Pb(w)$ be the \iword $\tau_\Pb(w) = b_1b_2b_3 \cdots$ where,
\[
  b_1 = (\ppart{\varepsilon},a_1) \quad \text{~~~and~~~} \quad b_i = (\ppart{a_1 \cdots a_{i-1}},a_i) \quad \text{for $i \geq 2$}
\]

\medskip
\noindent
{\bf Enrichment.} Consider a class of \ilangs \Cs and a class of \flangs \Ds (do note that \Ds is a class of languages and not of \ilangs). The \Ds-enrichment of \Cs, denoted by $\Cs \circ \Ds$ is now defined as the following class of \ilangs . For any alphabet $A$, $(\Cs \circ \Ds)(A)$ contains all \ilangs of the following form:
\[
  \tau_\Pb\inv(L) \quad \text{where \Pb is a \Ds-partition of $A^*$ and $L \in \Cs(\Pb \times A)$}.
\]
\begin{remark}
  The definition is actually simpler in this setting. Indeed, since we are dealing with \ilangs, it makes no sense to consider intersections with elements of\/ \Pb, which are \fword languages.
\end{remark}

As before, we are mainly interested in \su-enrichment since our theorem applies to this special case. As for finite words,   \su-enrichment for classes of \ilangs captures the intuitive connection between strong and weak logical fragments. One may show that over \iwords as well, \fodp is the \su-enrichment of \fodw and for any $n \geq 1$ \sipe{n} and \bspe{n} are respectively the \su-enrichments of \sio{n} and \bso{n}. Since the proofs are essentially identical\footnote{In fact, the proofs are even simpler in this setting. Since \iwords have no ``right border'', there are less cases to treat.} to those we presented in Sections~\ref{sec:fo2} and~\ref{sec:hiera} for finite words, they are left to the reader.

\medskip

We now turn to the variant for \iwords of our main theorem: given any lattice  \Cs of \ilangs which is closed under right quotient and inverse image, $(\Cs \circ \su)$-covering reduces to \Cs-covering. Both the reduction and its proofs are adapted from what we did for finite words in Section~\ref{sec:wfwords}. We start by generalizing well-formed words.

\subsection{\texorpdfstring{Languages of well-formed \iwords}{Languages of well-formed omega words}} Similarly, to what happened for finite words in Section~\ref{sec:wfwords}, using our reduction for \iwords requires working with the algebraic definition of regular \ilangs, which is based on \isemis. We first briefly recall the definition of \isemis and refer the reader to the book of Perrin and Pin~\cite{PerrinPin:Infinite-Words:2004:a} for more details.

\medskip
\noindent
{\bf \isemis.} An \emph{\isemi} is a pair $(S_+,S_{\omega})$ where $S_+$ is a semigroup and $S_{\omega}$ is a set. Moreover, $(S_+,S_{\omega})$ is equipped with two additional products: a \emph{mixed product} $S_+ \times S_{\omega} \rightarrow S_{\omega}$ that maps $s \in S_+$ and $t \in S_{\omega}$ to an element denoted $st \in S_\omega$, and an \emph{infinite product} $(S_{+})^{\omega} \rightarrow S_{\omega}$ that maps an infinite sequence $s_1,s_2,\dots \in (S_{+})^{\omega}$ to an element of $S_{\omega}$ denoted by $s_1s_2\cdots$. We require these products as well as the semigroup product of $S_+$ to satisfy all possible forms of associativity (see~\cite{PerrinPin:Infinite-Words:2004:a} for details). Finally, we denote by $s^{\omega}$ the element $sss\cdots$. Clearly, $(A^+,A^{\omega})$ is an \isemi for any alphabet~$A$. The notion of morphism is adapted to \isemis in the natural way.

An \isemi is said to be \emph{finite} if both $S_+$ and $S_{\omega}$ are finite. Note that even if an \isemi is finite, it is not obvious that a finite representation of the infinite product exists. However, it was proven by Wilke~\cite{womega} that the infinite product is fully determined by the mapping $s \mapsto s^\omega$, yielding a finite representation for finite \isemis.

An \ilang~$L \subseteq A^\omega$ is said to be \emph{recognized} by an \isemi $(S_+,S_{\omega})$ if there exist $F \subseteq S_{\omega}$ and a morphism $\alpha : (A^+,A^\omega) \rightarrow (S_+,S_\omega)$ such that $L = \alpha^{-1}(F)$. It is well known that an \ilang is regular if and only if it is recognized by a \emph{finite} \isemi.

\medskip
\noindent
{\bf Well-formed \iwords.} We may now adapt the notion of well-formed words to \iwords. To any morphism $\alpha: (A^+,A^\omega) \to (S_+,S_\omega)$ into a finite \isemi, we associate a new alphabet \awfa of \emph{well-formed \iwords}. Then, given any \ilang $L\subseteq A^\omega$ recognized by $\alpha$, we associate a new \ilang $\wfwa{L}\subseteq \awfa^\omega$.

We denote by $S$ the semigroup $S = \alpha(A^+) \subseteq S_+$. Moreover, we write $E(S)$ for the set of idempotent elements in $S$. Let ``$\square$'' be some symbol which does not belong to $S$. The \emph{alphabet of well-formed \iwords associated to $\alpha$}, denoted by \awfa, is defined as follows:
\[
  \awfa = (E(S) \cup \{\square\}) \times S \times E(S)
\]

\begin{remark}
  This definition is simpler than the one for \fwords. We do not need letters of the form $(e,s,\square)$ as \iwords do not have a ``right border''.
\end{remark}

The definition of well-formed \iwords is the natural one. We say that an \iword $w \in \awfa^\omega$ is \emph{well-formed} when it can be written as follows:
\[
  w = (\square,s_0,f_0) \cdot (e_1,s_1,f_1) \cdot (e_2,s_2,f_2) \cdots
\]
with $f_i=e_{i+1} \in E(S)$ for all $i \in \nat$. It is immediate by definition that the language of all well-formed \iwords in $\awfa^\omega$ is regular.

\begin{fct} \label{fct:wfwords:wfiregular}
  The language of all well-formed \iwords in $\awfa^\omega$ is regular.
\end{fct} 

We now associate a new \ilang over \awfa to each \ilang $L$ recognized by $\alpha$: the \emph{language of well-formed \iwords} associated to $L$. As the name suggests, it is made exclusively of well-formed \iwords.

We define a canonical morphism $\eval: (\awfa^+,\awfa^\omega) \to (S_+,S_\omega)$ by giving the image of the two kinds of letters in \awfa. Let $s \in S$ and $e,f \in E(S)$, we define,
\[
  \begin{array}{lll}
    \eval((e,s,f)) = esf      & \qquad & \eval((\square,s,f)) = sf.
  \end{array}
\]
Consider an \ilang $L$ recognized by $\alpha$. We define $\wfwa{L}\subseteq \awfa^\omega$ as follows:
\[
  \wfwa{L} = \bigl \{w \in \awfa^\omega \mid w \text{ is well-formed and } \eval(w) \in \alpha(L)\bigr\}.
\]
Clearly, $\wfwa{L}$ is the intersection of the language of all well-formed \iwords with an \ilang recognized by \eval. Thus, it is regular.

\begin{fct} \label{fct:wfwords:wfiregular2}
  For any \ilang $L$ recognized by $\alpha$, \wfwa{L} is regular.
\end{fct}

Finally, we lift the definition to multisets \Lb made of \ilangs recognized by $\alpha$ and write $\wfwa{\Lb}$ for the multiset $\wfwa{\Lb} = \{\wfwa{L} \mid L \in \Lb\}$.

\subsection{The reduction theorem for \iwords} We may now adapt Theorem~\ref{thm:wfwords} to \iwords. We state an effective reduction from $(\Cs \circ \su)$-covering to \Cs-covering which holds for any lattice of \ilangs \Cs closed under inverse image  (note that unlike in the setting of finite \fwords, we do not require \Cs to be nontrivial here).

\begin{theorem} \label{thm:wfiwords}
  Let $\alpha: (A^+,A^\omega) \to (S_+,S_\omega)$ be an \isemi morphism and let \Cs be a lattice of \ilangs closed under inverse image. Moreover, let $L$ be an \ilang and \Lb be a multiset of \ilangs, all recognized by $\alpha$. Then, the following properties are equivalent:
  \begin{enumerate}
  \item $(L,\Lb)$ is $(\Cs \circ \su)$-coverable.
  \item $(L,\Lb)$ is $(\Cs \circ \su_{n})$-coverable, where $n = 2^{3|S_+|+1}$.
  \item $(\wfwa{L},\wfwa{\Lb})$ is \Cs-coverable.
  \end{enumerate}
\end{theorem}

As announced, the statement of Theorem~\ref{thm:wfiwords} is a natural analogue of Theorem~\ref{thm:wfwords}. Let us point out that in this case as well, we shall present a constructive proof.

\begin{remark}
  The constant used in the second item of Theorem~\ref{thm:wfiwords} is much larger than the corresponding one in Theorem~\ref{thm:wfwords}. This is explained by technical difficulties that are specific to \iwords and arise when proving the implication $(3) \Rightarrow (2)$ of the theorem.
\end{remark}

Finally, note that as before, we may adapt Theorem~\ref{thm:wfiwords} to accommodate the simpler separation problem.

\begin{corollary} \label{cor:wfiwords}
  Let $\alpha: (A^+,A^\omega) \to (S_+,S_\omega)$ be an \isemi morphism and let \Cs be a lattice of \ilangs closed under inverse image. Moreover, let $L_1,L_2$ be two  \ilangs recognized by $\alpha$. Then, the following properties are equivalent:
  \begin{enumerate}
  \item $L_1$ is $(\Cs \circ \su)$-separable from $L_2$.
  \item $L_1$ is $(\Cs \circ \su_{n})$-separable from $L_2$  where $n = 2^{3|S_+|+1}$.
  \item \wfwa{L_1} is \Cs-separable from \wfwa{L_2}.
  \end{enumerate}
\end{corollary}

While we shall not detail the applications of Theorem~\ref{thm:wfiwords} as much as we did for Theorem~\ref{thm:wfwords}, let us briefly outline them. As we explained above, one may show that over \iwords, \fodp is the \su-enrichment of \fodw and for any $n \geq 1$ \sipe{n} and \bspe{n} are respectively the \su-enrichments of \sio{n} and \bso{n}. It was shown in~\cite{ppzinf16} that over \iwords, separation is decidable for the levels \sio{2} and \sio{3} of the quantifier alternation hierarchy. Thus, we obtain from Theorem~\ref{thm:wfiwords} that separation is decidable for \sipe{2} and \sipe{3} over \iwords.

\begin{remark}
  We do not speak about covering since there are no published results for covering in the setting of \iwords. We also do not mention separation \fodp, \sipe{1} and \bspe{1} over \iwords for the same reason. However, let us point out that our problem here is just the lack of bibliography. It is actually possible to generalize the existing results and show that covering and separation are both decidable for \fodw, \sio{1}, \bso{1}, \sio{2}, \bso{2} and \sio{3} overs \iwords. Thus, Theorem~\ref{thm:wfiwords} can be used to obtain the same results over \iwords as the ones obtained from Theorem~\ref{thm:wfwords} over words.
\end{remark}

The remainder of this section is devoted to proving Theorem~\ref{thm:wfiwords}. We fix an arbitrary morphism $\alpha: (A^+,A^\omega) \to (S_+,S_\omega)$ and we let $S$ be the semigroup $S = \alpha(A^+)$. Recall that the associated alphabet of well-formed words is defined as follows:
\[
  \awfa = (E(S) \cup \{\square\}) \times S \times E(S).
\]
Let \Cs be a lattice of \ilangs closed under inverse image. Our objective is to show that when $L$ and \Lb are respectively a \ilang and a multiset of \ilangs, all recognized by $\alpha$, the following properties are equivalent:
\begin{enumerate}
\item $(L,\Lb)$ is $(\Cs \circ \su)$-coverable.
\item $(L,\Lb)$ is $(\Cs \circ \su_{n})$-coverable where $n=2^{3|S_+|+1}$.
\item $(\wfwa{L},\wfwa{\Lb})$ is \Cs-coverable.
\end{enumerate}

We prove that $(1) \Rightarrow (3) \Rightarrow (2) \Rightarrow (1)$. As before, observe that the direction $(2) \Rightarrow (1)$ is trivial since it is clear that $\Cs \circ \su_{n} \subseteq \Cs \circ \su$. Thus, we may concentrate on $(1) \Rightarrow (3)$ and $(3) \Rightarrow (2)$.

The argument for the direction $(1) \Rightarrow (3)$ is basically identical to the one we used when proving the corresponding implication in Theorem~\ref{thm:wfwords}. For this reason, we shall only briefly sketch it. On the other hand, we provide more details for the proof of the implication $(3) \Rightarrow (2)$, which slightly departs from what we did for \fwords.

\subsection{\texorpdfstring{From $(\Cs \circ \su)$-covering to \Cs-covering}{From (C o SU)-covering to C-covering}}

We start with the direction $(1) \Rightarrow (3)$. As we explained above, the argument is essentially the same as the one we presented for the corresponding direction in Theorem~\ref{thm:wfwords}. In fact, it is even made simpler by the fact that the  definition of $\Cs \circ \su$ is less involved for classes of \ilangs.

The argument is based on the following proposition which is adapted from the one we used to handle the corresponding direction for \fwords: Proposition~\ref{prop:wform:firstdir}.

\begin{proposition} \label{prop:ifirstdir}
  For any integer $k \geq 1$, there exists a map $\gamma: \awfa^\omega \to A^\omega$ satisfying the two following properties:
  \begin{enumerate}
  \item For any \ilang $L \subseteq A^\omega$ recognized by $\alpha$ and any well-formed \iword $w \in \awfa^\omega$, we have $w \in \wfwa{L}$ if and only if $\ucroch{w} \in L$.
  \item For any \ilang $K \in (\Cs \circ \su_k)(A)$, there exists $H_K \in \Cs(\awfa)$ such that for any well-formed \iword $w \in \awfa^\omega$, $w \in H_K$ if and only if $\ucroch{w} \in K$.
  \end{enumerate}
\end{proposition}

One may show the direction $(1) \Rightarrow (3)$ in Theorem~\ref{thm:wfiwords} from this proposition using an argument which is similar to the one used for proving the same direction in Theorem~\ref{thm:wfwords} from Proposition~\ref{prop:wform:firstdir}. We leave this argument to the reader and prove Proposition~\ref{prop:ifirstdir}. Let us fix $k \geq 1$ for the~proof.

\medskip
\noindent
{\bf Proof of Proposition~\ref{prop:ifirstdir}: definition of $\gamma$.}   We start by defining the map $\gamma: \awfa^\omega \to A^\omega$ and then show that it satisfies the desired properties. We actually define a morphism $\gamma: \awfa^+ \to A^+$ which we lift as a map $\gamma: \awfa^\omega \to A^\omega$. Hence, it suffices to describe the image of any letter in \awfa.

For any element $s \in S$ (recall that $S = \alpha(A^+)$), we associate an arbitrarily chosen {\bf nonempty} word $\wcroch{s} \in A^+$ such that $\alpha(\wcroch{s}) = s$ (note that such a word exists by definition of $S$). We are now ready to define our morphism $\gamma: \awfa^+ \to A^+$. Recall that there are two kinds of letters in \awfa. Given  $s \in S$ and $e,f \in E(S)$, we define,
\[
  \begin{array}{lll}
    \gamma((e,s,f))             & = & \wcroch{e}^{k}\wcroch{s}\wcroch{f}^{k} \\
    \gamma((\square,s,f))       & = & \wcroch{s}\wcroch{f}^{k}
  \end{array}
\]
Given $w = b_1b_2 \cdots \in \awfa^\omega$, we now define $\gamma(w) = \gamma(b_1)\gamma(b_2) \cdots \in A^\omega$. It remains to prove that $\gamma$ satisfies the two properties stated in Proposition~\ref{prop:wform:firstdir}.

\medskip
\noindent
{\bf Proof of Proposition~\ref{prop:ifirstdir}: first item.} Consider an \ilang $L \subseteq A^\omega$ which is recognized by~$\alpha$. We have to show that for any well-formed \iword $w \in \awfa^\omega$, $w \in \wfwa{L}$ if and only if $\ucroch{w} \in L$.

Since $w$ is well-formed, we have $w \in \wfwa{L}$ if and only if $\eval(w) \in \alpha(L)$. Moreover, since $\alpha$ recognizes $L$, we have $\ucroch{w} \in L$ if and only if $\alpha(\ucroch{w}) \in \alpha(L)$. Hence, it suffices to prove that $\eval(w) = \alpha(\ucroch{w})$. By definition,
\[
  \begin{array}{cll}
    w      & = & (\square,s_0,e_1) \cdot (e_1,s_1,e_2) \cdot (e_2,s_2,e_3) \cdots                                               \\
    \ucroch{w} & = & \wcroch{s_0}\wcroch{e_1}^{2k}\wcroch{s_1}\wcroch{e_2}^{2k} \wcroch{s_2} \cdot \wcroch{e_3}^{2k}\cdots
  \end{array}
\]
Hence, we have,
\[
  \begin{array}{cll}
    \eval(w)          & =  & s_0e_1e_1s_1e_2s_2e_3  \cdots \\
    \alpha(\ucroch{w}) & = & s_0(e_1)^{2k}s_1(e_2)^{2k}s_2(e_3)^{2k}\cdots
  \end{array}
\]
Therefore, since each element $e_i \in E(S)$ is an idempotent, $\eval(w) = \alpha(\ucroch{w})$.

\medskip
\noindent
{\bf Proof of Proposition~\ref{prop:ifirstdir}: second item.} Consider $K \in (\Cs \circ \su_k)(A)$, we have to build a new \ilang $H_K \in \Cs(\awfa)$ satisfying the following property:
\begin{equation} \label{eq:igoal1}
  \text{For any well-formed \iword $w \in \awfa^\omega$,} \quad \text{$w \in H_K$ if and only if $\ucroch{w} \in K$}
\end{equation}
The argument is simpler that what we did for \fwords since the definition of $\Cs \circ \su_{k}$ is less involved. By definition, there exists an $\su_k$-partition \Pb of $A^*$ and $L \in \Cs(\Pb \times A)$ such that,
\[
  K = \tau_\Pb\inv(L).
\]
The construction of $H_K$ is based on the following lemma (which is adapted from Lemma~\ref{lem:wfwords:wwordsfinal} used in the  case of finite words). Recall that any morphism $\beta: \awfa^* \to (\Pb \times A)^*$ may be lifted as a map $\beta: \awfa^\omega \to (\Pb \times A)^\omega$.

\begin{lemma} \label{lem:iwwordsfinal}
  There exists a morphism $\beta: \awfa^* \to (\Pb \times A)^*$ such that for any well-formed \iword $w \in \awfa^\omega$, we have $\tau_\Pb(\ucroch{w}) = \beta(w)$.
\end{lemma}

The proof of Lemma~\ref{lem:iwwordsfinal} is identical to the one of Lemma~\ref{lem:wfwords:wwordsfinal}. It is left to the reader. Let us use the lemma to construct $H_K$ and finish the proof of Proposition~\ref{prop:ifirstdir}. We have an \ilang $L \in \Cs(\Pb \times A)$ such that $K = \tau_\Pb\inv(L)$. Let us define
\[
  H_K = \beta^{-1}(L).
\]
Since \Cs is closed inverse image, we obtain that $H_K \in \Cs(\awfa)$. We now prove that $H_K$ satisfies~\eqref{eq:wform:goal1} using Lemma~\ref{lem:iwwordsfinal}. Given a well-formed \iword $w \in \awfa^\omega$, we have $w \in H_K$ if and only if $\beta(w) \in L$. The lemma then says that this is equivalent to $\tau_\Pb(\ucroch{w}) \in L$, \emph{i.e.}, to $\ucroch{w} \in K$ by hypothesis on $K$.

\subsection{\texorpdfstring{From \Cs-covering to $\Cs \circ \su$-covering}{From C-covering to (C o SU)-covering}}

We now turn to the direction $(3) \Rightarrow (2)$ in Theorem~\ref{thm:wfwords}. While the proof remains very similar to the one for the corresponding direction in Theorem~\ref{thm:wfwords}, there is a significant technical difference. The argument is based on the following proposition adapted from Proposition~\ref{prop:wform:secdir}. We let $n = 2^{3|S_+|+1}$ for the proof.

\begin{proposition} \label{prop:isecdir}
  There exists a map $\eta: A^\omega \to \awfa^\omega$ satisfying the two following properties:
  \begin{enumerate}
  \item For any \ilang $L \subseteq A^\omega$ recognized by $\alpha$, we have $L = \eta\inv(\wfwa{L})$.
  \item For any \ilang $K \in \Cs(\awfa)$, the \ilang $\eta\inv(K)$ belongs to $(\Cs \circ\su_{n})(A)$.
  \end{enumerate}
\end{proposition}

As before, one may show the direction $(3) \Rightarrow (2)$ in Theorem~\ref{thm:wfiwords} from this proposition using an argument which is identical to the one used for proving the same direction in Theorem~\ref{thm:wfwords} from Proposition~\ref{prop:wform:secdir}. Therefore, we leave it to the reader. We concentrate on proving Proposition~\ref{prop:isecdir}.

\medskip
\noindent
{\bf Proof of Proposition~\ref{prop:isecdir}: definition of $\eta$.} We begin by defining $\eta: A^\omega \to \awfa^\omega$. While similar, the definition is slightly different from the one we used when proving Proposition~\ref{prop:wform:secdir}.

We first generalize the notion of \emph{\ktype} to \iwords. Given an \iword $w$, a position $x$ in $w$ and a natural number $k\in\nat$, a \emph{\ktype of $x$} is the following \fword of length at most $k$:
\begin{itemize}
\item If $x \leq k$, then the \ktype of $x$ is the prefix $w[1,x-1]$ of length $x-1$.
\item If $x > k$, then the \ktype of $x$ is the infix $w[x-k,x-1]$ of length $k$.
\end{itemize}

For the construction of $\eta$, we fix $k = 2^{3|S_+|}$, so that $n = 2k$. Moreover, we choose an arbitrary order on the set of idempotents $E(S)$. We now generalize the notion of \emph{distinguished} position to \iwords. The definition differs from the one we used for finite words. This change is needed to prove the first item in Proposition~\ref{prop:isecdir} (on the other hand, it is harmless for the second item: its proof is identical to the one for finite \fwords). It is also the reason for using a larger constant $k$ in this setting.

Consider an \iword $w$ and a position $x$ in $w$. Moreover, let $u$ be the \ktype of $x$. We say that $x$ is \emph{distinguished} when there exists a nonempty suffix $v \in A^+$ of $u$ such that $\alpha(v) \in E(S)$.

\begin{remark}
  As for finite words, when $x$ is distinguished, we have an idempotent $e \in E(S)$ such that $\alpha(u) \cdot e = \alpha(u)$ (namely, $e=\alpha(v)$ with $v$ defined as above). However, in the case of \iwords, we have a stronger	property: $u$ has a suffix whose image under $\alpha$ is an idempotent. We need this to prove the first item in Proposition~\ref{prop:isecdir}.
\end{remark}

We now generalize Fact~\ref{fct:wform:herecomesdis}: distinguished positions occur frequently in \iwords.

\begin{fct} \label{fct:freqi}
  Let $w \in A^\omega$, let $k=2^{3|S_+|}$
  and let $y \geq k-1$ be some position of $w$. Then, there exists a distinguished position $x$ in $w$ such that $y - (k-1) \leq x \leq y$.
\end{fct}

\begin{proof}
  It is known that every word $v \in A^+$ of length greater than $k = 2^{3|S_+|}$ contains an infix whose image under $\alpha$ is an idempotent (this is an immediate consequence of Simon's factorization forest theorem~\cite{simonfacto,cfacto,kfacto}). Hence, since the infix, $w[y-(k-1),y]$ has length $k$, the result follows.
\end{proof}

We may now define the map $\eta: A^\omega \to \awfa^\omega$. Let $w \in A^\omega$. It is immediate from Fact~\ref{fct:freqi} that $w$ contains infinitely many distinguished positions. Let $x_0 < x_1 < x_2 < \cdots$ be these distinguished positions. For all $i \geq 0$, we let $u_i$ be the \ktype of $x_i$ and $e_i \in E(S)$ be the smallest idempotent (according to the arbitrary order that we fixed on idempotents) such that $u_i$ has a nonempty suffix whose image under $\alpha$ is $e_i$. We define $\croch{w} \in \awfa^\omega$ as the \iword:
\[
  \croch{w} = (\alpha(w_0),e_0)\cdot (e_0,\alpha(w_1),e_1) \cdot
  (e_{1},\alpha(w_{2}),e_2)\cdot (e_{2},\alpha(w_{3}),e_3) \cdots \in \awfa^\omega
\]
\noindent
where $w_0 = w[0,x_0-1]$ and for all $i \geq 1$, $w_i = w[x_{i-1},x_i - 1]$. Note that $\croch{w}$ is well-formed by definition. It remains to show that $\eta$ satisfies the two items in Proposition~\ref{prop:isecdir}.

\medskip
\noindent
{\bf Proof of Proposition~\ref{prop:isecdir}: first item.} This is where the technical differences with the proof of Proposition~\ref{prop:wform:secdir} occur.  Consider a \ilang $L \subseteq A^\omega$ recognized by $\alpha$ and an \iword $w \in A^\omega$. We have to show that $w \in L$ if and only if $\croch{w} \in \wfwa{L}$.

Since $\alpha$ recognizes $L$, we have $w \in L$ if and only if $\alpha(w) \in \alpha(L)$. Moreover, since $\croch{w}$ is well-formed, by definition, $\croch{w} \in \wfwa{L}$ if and only if $\eval(\croch{w}) \in \alpha(L)$. Hence, it suffices to prove that $\eval(\croch{w}) = \alpha(w)$. This requires more work than for \fwords. By definition, we know that $w$ may be decomposed as $w = w_0w_1w_2 \cdots$ and
\[
  \croch{w} = (\alpha(w_0),e_0)\cdot (e_0,\alpha(w_1),e_1) \cdot
  (e_{1},\alpha(w_{2}),e_2)\cdot (e_{2},\alpha(w_{3}),e_3) \cdots
\]
such that for all $i \geq 0$, $w_0 \cdots w_i$ has a suffix of length at most $k$ whose image under $\alpha$ is $e_i$. In particular, it is immediate by definition of \eval that,
\[
  \eval(\croch{w})= \alpha(w_0) e_0 \alpha(w_1) e_1 \alpha(w_2) e_2 \alpha(w_3) e_3 \cdots
\]
Therefore, we need to show that,
\[
  \alpha(w_0)\alpha(w_1)\alpha(w_2)\alpha(w_3) \cdots = \alpha(w_0) e_0 \alpha(w_1) e_1 \alpha(w_2) e_2 \alpha(w_3) e_3 \cdots
\]
Using a standard Ramsey argument, we obtain an infinite sequence of indices $i_1,i_2,i_3,\dots$ together with $s,t \in S$ and $f,g \in E(S)$ such that,
\begin{enumerate}
\item $sf = s$ and $tg=t$.
\item $\alpha(w_1) \cdots \alpha(w_{i_1}) = s$ and $\alpha(w_1)e_1
  \cdots \alpha(w_{i_1})e_{i_1} = t$.
\item For all $j > 1$, $\alpha(w_{i_{j-1}+1}) \cdots \alpha(w_{i_j})
  = f$ and $\alpha(w_{i_{j-1}+1})e_{i_{j-1}+1} \cdots
  \alpha(w_{i_j})e_{i_j} = g$.
\end{enumerate}
Therefore, $\alpha(w_0)\alpha(w_1)\alpha(w_2) \cdots = sf^\omega$ and $\alpha(w_0) e_0 \alpha(w_1) e_1 \alpha(w_2) e_2 \cdots = tg^\omega$. Hence, it suffices to prove that $sf^\omega = tg^\omega$. This is a consequence of the following lemma.

\begin{lemma} \label{lem:eqinf}
  We have $s=t$, $fg = f$ and $gf = g$.
\end{lemma}

Before we prove the lemma, we show that $sf^\omega = tg^\omega$ and conclude the argument for the first item in Proposition~\ref{prop:isecdir}. Using the lemma and the fact that
$sf =s$, we obtain, 
\[
  sf^\omega = s(fg)^\omega = sf (gf)^\omega = s(gf)^\omega = tg^\omega
\]
It remains to prove Lemma~\ref{lem:eqinf}. We prove the three equalities separately.

\medskip
\noindent
{\it First Equality: $s=t$.} By hypothesis, we know that for all $i \geq 0$, $w_0 \cdots w_i$ has a suffix of length at most $k$ whose image under $\alpha$ is $e_i$.  It follows that $\alpha(w_0 \cdots w_i) \cdot e_i = \alpha(w_0 \cdots w_i)$.  Thus, it is immediate from a simple induction that indeed,
\[
  s = \alpha(w_1) \cdots \alpha(w_{i_1}) = \alpha(w_1)e_1 \cdots \alpha(w_{i_1})e_{i_1} = t.
\]

\medskip
\noindent
{\it Second Equality: $fg = f$.} Let $j \geq 3$ be a large enough integer so that the word $w_{i_1+1}w_{i_1+2} \cdots w_{i_{j-1}}$ has length at least $k = 2^{3|S_+|}$. By definition and using the fact that $f$ is idempotent, we have,
\[
  \begin{array}{lll}
    \alpha(w_{i_1+1}w_{i_1+2} \cdots w_{i_{j-1}}) & = & f, \\
    \alpha(w_{i_{j-1}+1}) \cdots \alpha(w_{i_j}) & = & f, \\
    \alpha(w_{i_{j-1}+1})e_{i_{j-1}+1} \cdots \alpha(w_{i_j})e_{i_j} &= & g.
  \end{array}
\]
Therefore it suffices to show that,
\[
  \alpha(w_{i_1+1} \cdots w_{i_{j-1}}) \cdot \alpha(w_{i_{j-1}+1}) \cdots \alpha(w_{i_j}) = \alpha(w_{i_1+1} \cdots w_{i_{j-1}}) \cdot  \alpha(w_{i_{j-1}+1})e_{i_{j-1}+1} \cdots \alpha(w_{i_j})e_{i_j}.
\]
By hypothesis, we know that for all $i \geq 0$, $w_0 \cdots w_i$ has a suffix {\bf of length at most $k$} whose image under $\alpha$ is $e_i$. Thus, since $w_{i_1+1}w_{i_1+2} \cdots w_{i_{j-1}}$ has length at least $k$ by hypothesis, it follows that for all $i \in \{i_{j-1}+1,\dots,i_j\}$, we have,
\[
  \alpha(w_{i_1+1} \cdots w_{i}) = \alpha(w_{i_1+1} \cdots w_{i}) \cdot e_i.
\]
Hence, the result follows from a simple induction.

\medskip
\noindent
{\it Third Equality: $gf=g$.} Consider the word $w_{i_{1}+1} \cdots w_{i_2}$, which is mapped to $f$ by $\alpha$. By definition, for all $i_{1}+1 \leq j \leq i_2$, we have a suffix $v_j$ of length at most $k$ of $w_{1} \cdots w_{j}$ such that $\alpha(v_j) = e_j$. We consider two sub-cases.

First assume that for all $i_{1}+1 \leq j \leq i_2$, the \fword $v_j$ is a suffix of $w_{i_{1}+1}\cdots w_{j}$. In that case, it is immediate from as simple induction that we have,
\[
  f = \alpha(w_{i_{1}+1}\cdots w_{i_2}) =
  \alpha(w_{i_{1}+1})e_{i_1+1}\cdots \alpha(w_{i_2})e_{i_2} = g
\]
Hence, we get $gf = g$. Otherwise, we consider the largest index $j$ with $i_{1}+1 \leq j \leq i_2$ such that $v_j$ is not a suffix of $w_{i_{1}+1} \cdots w_{j}$. Observe that,
\begin{enumerate}[label=$\alph*)$]
\item Since $j$ has been chosen to be maximal, we get from a simple induction that,
  \begin{equation}\label{eq:gf=g}
    \alpha(w_{i_{1}+1} \cdots w_{j})\alpha(w_{j+1})e_{j+1} \cdots
    \alpha(w_{i_2})e_{i_2} = \alpha(w_{i_{1}+1}\cdots w_{i_2}) = f.
  \end{equation}
\item\label{item:vj} Since $v_{j}$ is by definition a suffix of $w_{1} \cdots w_{j}$ but not of $w_{i_{1}+1} \cdots w_{j}$, it follows that $w_{i_{1}+1} \cdots w_{j}$ is a suffix of $v_{j}$.
\item Since $\alpha(v_{j}) = e_{j}$, we obtain from \ref{item:vj} some $r \in S$ such that $r \cdot \alpha(w_{i_{1}+1} \cdots w_{j})= e_{j}$.
\end{enumerate}
Multiplying~\eqref{eq:gf=g} by $r$ on the left, we obtain therefore:
\[
  e_{j}\alpha(w_{j+1})e_{j+1} \cdots \alpha(w_{i_2})e_{i_2} = rf.
\]
We may now multiply this equality on the left by $\alpha(w_{i_{1}+1})e_{i_1+1} \cdots e_{j-1}\alpha(w_{j})$, which yields,
\[
  \alpha(w_{i_{1}+1})e_{i_1+1}\cdots \alpha(w_{i_2})e_{i_2} = \alpha(w_{i_{1}+1})e_{i_1+1} \cdots \alpha(w_{j})rf,
\]
that is,
\[
  g = \alpha(w_{i_{1}+1})e_{i_1+1} \cdots \alpha(w_{j})rf.
\]
In other words, we have found some element $r' \in S$ such that $g = r'f$. It follows that $gf = r'ff=r'f=g$, which concludes the proof.

\medskip
\noindent
{\bf Proof of Proposition~\ref{prop:isecdir}: second item.} Given an arbitrary \ilang $K \in \Cs(\awfa)$, we have to prove that $\eta\inv(K)$ belongs to $(\Cs \circ \su_{2k})(A)$ (recall that we fixed $n = 2k$). The proof of this item is essentially a simplified version of the corresponding argument for finite words (it is simpler since \iwords have no right ``border'').

Recall that \eqsu{2k} denotes the canonical equivalence associated to $\su_{2k}$. We let \Pb be the partition of $A^*$ into $\eqsu{2k}$-classes. By definition of $\Cs \circ \su_{2k}$, it suffices to exhibit an \ilang $L \in \Cs(\Pb \times A)$ such that $\eta\inv(K) = \tau_\Pb\inv(L)$. We use the following lemma (which is adapted from Lemma~\ref{lem:canonictwo}). Recall that we may lift a morphism $\beta: (\Pb \times A)^* \to \awfa^*$ as a map $\beta: (\Pb \times A)^\omega \to \awfa^\omega$.

\begin{lemma} \label{lem:icanonictwo}
  There exists a morphism $\beta: (\Pb \times A)^* \to \awfa^*$ such that for any \iword $w \in A^\omega$, we have $\beta(\tau_\Pb(w)) = \eta(w)$.
\end{lemma}

The proof of Lemma~\ref{lem:icanonictwo} is identical to the one of Lemma~\ref{lem:canonictwo} (using Fact~\ref{fct:freqi} instead of Fact~\ref{fct:wform:herecomesdis}). We leave it to the reader. It remains to finish the proof of Proposition~\ref{prop:isecdir}.

Let $\beta: (\Pb \times A)^* \to \awfa^*$ be the morphism defined in Lemma~\ref{lem:icanonictwo}. We claim that
\[
  \eta\inv(K) = \tau_\Pb\inv(\beta\inv(K)).
\]
This will conclude the proof, since $\beta\inv(K)$ belongs to $\Cs(\Pb \times A)$ by closure under inverse image. It remains to prove the claim. Let $w \in A^\omega$. By definition of $\beta$ in Lemma~\ref{lem:icanonictwo}, we have $w \in \eta\inv(K)$ if and only if $\beta(\tau_\Pb(w)) \in K$. This equivalent to $\tau_\Pb(w) \in \beta^{-1}(K)$. This exactly says that $w \in (\Pb \times A)^\omega \cap \tau_\Pb\inv(\beta\inv(K))$, as desired.

\section{Conclusion}
\label{sec:conc}
We presented generic reduction theorems for the \su-enrichment operation on classes of languages and \ilangs. Given any such class \Cs satisfying appropriate closure properties, we
reduce covering and separation for $\Cs \circ \su$ to the same problem for \Cs.

These theorems have many applications: for most logical fragments, \su-enrichment is the language theoretic counterpart of a natural logical operation: if \Cs is the class corresponding to some logical fragment, it is often the case that its \su-enrichment $\Cs \circ \su$ corresponds to the stronger fragment obtained by adding the predicates ``$+1$'', ``$min$'', ``$max$'' and ``$\varepsilon$'' to the signature. We showed this in the setting of finite words for the most prominent fragments of first-order logic, namely the two-variable fragment \fodw and the levels \sio{n} and \bso{n} in the quantifier alternation hierarchy. Combined with our reduction theorem and already known results, this shows that covering and separation are decidable for \fodp and the levels \sipe{1}, \bspe{1}, \sipe{2}, \bspe{2} and \sipe{3}. Note that several of these results were unknown, and that others have difficult combinatorial proofs only for membership algorithms (this is the case for \bspe1~\cite{knast83} and for \sipe2~\cite{gssig2}).

An interesting follow-up to our work would be to obtain a similar reduction theorem for another natural operation: \md-enrichment $\Cs \mapsto \Cs \circ \md$. Here, \md stands for the class of modulo languages. A language $L$ belongs to \md if and only if there exists a natural number $d \geq 1$ such that membership of a word $w$ in $L$ depends only $|w| \mod d$. This operation is important as it is the language theoretic counterpart of another natural logical operation. If \Cs corresponds to some logical fragment, then $\Cs \circ \md$  corresponds to the stronger fragment obtained by adding the \emph{modular predicates} to the signature. Essentially, they consist in unary predicates which can be used to test the number of a position modulo some constant.

\bibliographystyle{ACM-Reference-Format}

\end{document}